\documentclass[%
11pt,
onecolumn,
tightenlines,
superscriptaddress,
notitlepage,
preprintnumbers,
nofootinbib,
amsmath,amssymb,amsthm,
aps,
eqsecnum
]{revtex4-2}

\usepackage[utf8]{inputenc}
\usepackage{mathtools}
\usepackage{amsmath}
\usepackage{amsthm}
\usepackage{amsbsy}
\usepackage{amssymb}
\usepackage{amscd}
\usepackage{amsfonts}
\usepackage{bm}
\usepackage{graphics}
\usepackage{graphicx}
\usepackage[plain]{fancyref}
\usepackage{soul}
\usepackage{hyperref}
\setcitestyle{numbers,sort&compress}
\usepackage{verbatim}
\usepackage[margin=25mm]{geometry}
\usepackage{paralist}
\usepackage{wasysym}
\usepackage{xcolor}
\usepackage{MnSymbol}
\usepackage{lineno}

\usepackage{booktabs}
\usepackage{array}
\usepackage{tabularx}
\usepackage{longtable}
\usepackage{ragged2e}

\graphicspath{ {./figs/} }

\newcommand*{\fancyrefapplabelprefix}{app}
\frefformat{plain}{\fancyrefapplabelprefix}{appendix~#1}
\Frefformat{plain}{\fancyrefapplabelprefix}{Appendix~#1}
\frefformat{main}{\fancyrefapplabelprefix}{appendix~#1}
\Frefformat{main}{\fancyrefapplabelprefix}{Appendix~#1}

\DeclareMathOperator{\csch}{csch}

\newcommand{\partialx}[2]{\frac{\partial #1}{\partial #2}}

\newcommand{\evalAt}[2]{\left. #1 \right\rvert_{#2}}

\newcommand{\rmag}{r} 
\newcommand{\vvec}[1]{\bm{#1}}
\newcommand{\rvec}{\vvec{\rmag}}
\newcommand{\rrmsvec}{\vvec{\rmag}_{rms}}
\newcommand{\dir}[1]{\hat{#1}}
\newcommand{\rdir}{\dir{\rvec}}
\newcommand{\tens}[1]{\bm{#1}}

\newcommand{\Wchains}{W_{\text{ch}}}
\newcommand{\PEQ}{U_{\text{net}}}
\newcommand{\RVEdomain}{\Omega_0}
\newcommand{\rvedomain}{\Omega}
\newcommand{\cLen}{\ell}
\newcommand{\mLen}{b} 
\newcommand{\n}{n} 
\newcommand{\nSet}{\mathcal{N}}
\newcommand{\nPolydisperseSet}{\n_1, \dots, \n_{\numChains}}
\newcommand{\nA}{n_{\alpha}} 
\newcommand{\nB}{n_{\beta}} 
\newcommand{\nmean}{\overline{\n}} 
\newcommand{\ngeomean}{\nmean_{\text{geo}}} 
\newcommand{\eff}{\eta} 
\newcommand{\chainStretch}{\gamma} 
\newcommand{\critChainLength}{\rmag_{crit}} 
\newcommand{\rmsChainLength}{\rmag_{rms}} 

\newcommand{\crossDensity}{M} 

\newcommand{\partitionFunction}{\mathcal{Z}} 
\newcommand{\partitionFunctionFR}{\partitionFunction^{FR}} 
\newcommand{\partitionFunctionFA}{\partitionFunction^{FA}} 

\newcommand{\trans}[1]{#1^{T}} 
\newcommand{\collection}[1]{\left\{#1\right\}} 
\newcommand{\ncollection}[1]{\left(#1\right)} 

\newcommand{\Conv}{\text{Conv}}
\newcommand{\kB}{k_B} 
\newcommand{\T}{T} 

\newcommand{\chainConformationProbDens}{p}

\newcommand{\chainDistanceProbDens}{P}

\newcommand{\chainFreeEnergy}{w}
\newcommand{\freeEnergyDensity}{\mathcal{W}}

\newcommand{\clinkFreeEnergy}{W_{\inst}}
\newcommand{\clinkFreeRotationFreeEnergy}{\clinkFreeEnergy^{FR}}
\newcommand{\clinkFrameAveragingFreeEnergy}{\clinkFreeEnergy^{FA}}
\newcommand{\rveFreeEnergy}{W}
\newcommand{\rveFullNetworkFreeEnergy}{\rveFreeEnergy^{FN}}

\newcommand{\clinkFreeEnergyPolydisperse}[1]{W_{\inst,\ncollection{#1}}}

\newcommand{\clinkFreeEnergyApprox}[1]{\clinkFreeEnergyPolydisperse{#1}}
\newcommand{\clinkFreeRotationFreeEnergyApprox}[1]{\clinkFreeEnergyApprox{{#1}}^{FR}}
\newcommand{\clinkFrameAveragingFreeEnergyApprox}[1]{\clinkFreeEnergyApprox{{#1}}^{FA}}

\newcommand{\clinkInnerFreeEnergy}{\widehat{W}_{\inst}}
\newcommand{\clinkInnerFreeRotationFreeEnergy}{\clinkInnerFreeEnergy^{FR}}
\newcommand{\clinkInnerFrameAveragingFreeEnergy}{\clinkInnerFreeEnergy^{FA}}

\newcommand{\allClinksFreeEnergyDensity}{\freeEnergyDensity_{\chainSpace}}

\newcommand{\GaussChainConformationProbDens}{\chainConformationProbDens_{\text{G}}}
\newcommand{\KGChainConformationProbDens}{\chainConformationProbDens_{\text{KG}}}
\newcommand{\KGFreeEnergy}{\chainFreeEnergy_{\text{KG}}}
\newcommand{\GaussFreeEnergy}{\chainFreeEnergy_{\text{G}}}

\newcommand{\forcemag}{f}
\newcommand{\forcevec}{\vvec{\forcemag}}

\newcommand{\pFrame}{\vvec{P}}
\newcommand{\pDir}{\hat{\vvec{v}}}
\newcommand{\uVec}{\dir{\vvec{\upsilon}}}
\newcommand{\nuVec}{\dir{\vvec{\nu}}}
\newcommand{\probDens}{\rho}
\newcommand{\probDensF}[1]{\probDens_{#1}}
\newcommand{\probDensInst}{\probDensF{\inst}}

\newcommand{\probDensChain}{\probDens_{\n}}

\newcommand{\probDensClinker}{\probDens_{\numChains}}

\newcommand{\unitSphere}{\mathbb{S}^2}

\newcommand{\df}[1]{\text{d}#1} 
\newcommand{\intOverSphere}[1]{\int_{\unitSphere} \df{A} \: #1 } 
\newcommand{\volUnitSphere}{4 \pi}

\newcommand{\Lang}{\mathcal{L}}
\newcommand{\invLang}{\Lang^{-1}}
\newcommand{\invLangPrime}{\left(\Lang^{-1}\right)^{\prime}}

\newcommand{\numChains}{\nodeDegree}
\newcommand{\numChainsSet}{\mathcal{K}}

\newcommand{\defMap}{\bm{\Phi}}
\newcommand{\pullback}[1]{\tilde{#1}}
\newcommand{\Rmag}{\pullback{\rmag}_{rms}}

\newcommand{\Rdir}{\uVec}

\newcommand{\generic}{\Box}

\newcommand{\inst}{c}

\newcommand{\F}{\tens{F}}

\newcommand{\cGreenSym}{C}
\newcommand{\cGreen}{\tens{\cGreenSym}}
\newcommand{\strTens}{\tens{V}}
\newcommand{\polRot}{\tens{R}}
\newcommand{\genRot}{\tens{Q}}
\newcommand{\genRotRef}{\genRot_0}

\newcommand{\genRotStar}{\genRot^*}
\newcommand{\genRotMono}{\genRot_m}
\newcommand{\delRotMag}{\delta \rodMag}
\newcommand{\delRot}{\delta \rodVec}
\newcommand{\changeCoord}[1]{\bar{#1}}

\newcommand{\rodMag}{\omega}
\newcommand{\rodVec}{\vvec{\rodMag}}
\newcommand{\eulerx}{\alpha}
\newcommand{\eulery}{\beta}
\newcommand{\eulerz}{\xi}
\newcommand{\chainStretchStar}{\chainStretch^*}
\newcommand{\strDiag}{\bm{\Lambda}}

\newcommand{\SOThree}{SO\left(3\right)}
\newcommand{\volSOThree}{8 \pi^2}
\newcommand{\numSOThreeQuad}{N_{\SOThree}}
\newcommand{\weightFactorSOThreeQuad}{\upsilon}
\newcommand{\quadPointSOThreeQuad}{\genRotRef}
\newcommand{\numSphQuad}{N_{sph}}
\newcommand{\weightFactorSphQuad}{\nu}
\newcommand{\polarAngle}{\theta}
\newcommand{\azimuthalAngle}{\phi}
\newcommand{\spinAngle}{\psi}
\newcommand{\numSpinQuad}{N_{\spinAngle}}

\newcommand{\Gradd}{\bm{\nabla}}

\newcommand{\takeGrad}[1]{\Gradd #1}

\newcommand{\Xj}{X}
\newcommand{\xj}{x}

\newcommand{\Xvec}{\vvec{\Xj}} 
\newcommand{\xvec}{\vvec{\xj}} 
\newcommand{\ycmag}{y} 
\newcommand{\yc}{\vvec{\ycmag}} 
\newcommand{\ycQO}{\yc_{\genRotRef}}
\newcommand{\ycStar}{\yc^*} 
\newcommand{\ycStarQO}{\ycStar_{\genRotRef}} 
\newcommand{\ycStarFrameAveraging}{\yc^{*FA}} 
\newcommand{\ycFrameAveraging}{\yc^{FA}} 
\newcommand{\ycMono}{\yc_m} 
\newcommand{\delYcMag}{\delta \ycmag} 
\newcommand{\delYc}{\delta \yc} 
\newcommand{\delYcQO}{\delYc_{\genRotRef}} 
\DeclareMathOperator{\diag}{diag} 
\DeclareMathOperator{\trace}{Tr} 
\newcommand{\pStretchSymbol}{\lambda} 
\newcommand{\pStretch}[1]{\pStretchSymbol_{#1}} 

\newcommand{\identity}{\tens{I}}
\newcommand{\euclid}[1]{\hat{\vvec{e}}_{#1}}

\newcommand{\orderOf}[1]{\mathcal{O}\left(#1\right)}

\DeclareMathOperator{\Tr}{Tr}

\newcommand{\Reals}{\mathbb{R}}

\newcommand{\chainSpace}{\mathcal{C}}

\newcommand{\Ball}[2]{B_{#1}\left(#2\right)}
\newcommand{\dirac}[2]{\delta\left(#1-#2\right)}

\newcommand{\nullvec}{\vvec{0}}
\newcommand{\smallparam}{\epsilon}
\newcommand{\auxVar}[1]{a_{#1}}
\newcommand{\auxSign}[1]{s_{#1}}
\newcommand{\E}{E}
\newcommand{\PKSym}{\Sigma}
\newcommand{\PK}{\bm{\PKSym}}

\newcommand{\genFunc}{F}
\newcommand{\genGroup}{\mathcal{G}}

\newcommand{\nodeDegree}{k}

\newcommand{\kNetTors}{\kappa}

\newcommand{\Eqref}[1]{Eq. \eqref{#1}}

\newcommand{\Figref}[1]{Fig. \ref{#1}}
\makeatletter
\newcommand*{\gnuplotinput}[2][1.0]{%
  \begingroup
  \let\@gnplt@input@includegraphics=\includegraphics
  \def\includegraphics##1{\@gnplt@input@includegraphics[scale=#1]{#2}}%
  \let\@gnplt@input@picture=\picture
  \def\picture{\unitlength=#1\unitlength\relax\@gnplt@input@picture}%
  \input{#2}%
  \endgroup
}
\makeatother

\newtheorem{proposition}{Proposition}

\newtheorem*{lemma}{Lemma}

\theoremstyle{remark}
\newtheorem{remark}{Remark}

\newcommand*{\fancyrefproplabelprefix}{prop}
\frefformat{plain}{\fancyrefproplabelprefix}{proposition~#1}
\Frefformat{plain}{\fancyrefproplabelprefix}{Proposition~#1}
\frefformat{main}{\fancyrefproplabelprefix}{proposition~#1}
\Frefformat{main}{\fancyrefproplabelprefix}{Proposition~#1}

\newcommand*{\fancyrefdeflabelprefix}{def}
\frefformat{plain}{\fancyrefdeflabelprefix}{definition~#1}
\Frefformat{plain}{\fancyrefdeflabelprefix}{Definition~#1}
\frefformat{main}{\fancyrefdeflabelprefix}{definition~#1}
\Frefformat{main}{\fancyrefdeflabelprefix}{Definition~#1}

\newcommand*{\fancyreflemlabelprefix}{lem}
\frefformat{plain}{\fancyreflemlabelprefix}{lemma~#1}
\Frefformat{plain}{\fancyreflemlabelprefix}{Lemma~#1}
\frefformat{main}{\fancyreflemlabelprefix}{lemma~#1}
\Frefformat{main}{\fancyreflemlabelprefix}{Lemma~#1}

\newcommand*{\fancyrefasslabelprefix}{ass}
\frefformat{plain}{\fancyrefasslabelprefix}{assumption~#1}
\Frefformat{plain}{\fancyrefasslabelprefix}{Assumption~#1}
\frefformat{main}{\fancyrefasslabelprefix}{assumption~#1}
\Frefformat{main}{\fancyrefasslabelprefix}{Assumption~#1}

\newcommand*{\fancyrefconvlabelprefix}{conv}
\frefformat{plain}{\fancyrefconvlabelprefix}{convention~#1}
\Frefformat{plain}{\fancyrefconvlabelprefix}{Convention~#1}
\frefformat{main}{\fancyrefconvlabelprefix}{convention~#1}
\Frefformat{main}{\fancyrefconvlabelprefix}{Convention~#1}

\newcommand{\capTitle}[1]{\emph{#1}}

{}

\renewcommand{\hl}[1]{#1}
\newenvironment{hlbreakable}%
{}%
{}

\begin{document}

\title{Polydisperse polymer networks with irregular topologies: Mechanics of cross-link distributions}

\author{Jason Mulderrig}
\email{mulderrig.jason@gmail.com}
\affiliation{Materials \& Manufacturing Directorate, Air Force Research Laboratory, Wright-Patterson AFB, Ohio 45433, USA}
\affiliation{National Research Council (NRC) Research Associateship Programs, The National Academies of Sciences, Engineering, and Medicine}

\author{Michael Buche}
\affiliation{Materials and Failure Modeling, Sandia National Laboratories, Albuquerque, New Mexico 87185, USA}

\author{Matthew Grasinger}
\email{matthew.grasinger.1@us.af.mil}
\affiliation{Materials \& Manufacturing Directorate, Air Force Research Laboratory, Wright-Patterson AFB, Ohio 45433, USA}

\preprint{To appear in Journal of the Mechanics and Physics of Solids: \url{https://doi.org/10.1016/j.jmps.2026.106706}}



\begin{abstract}
  The structure of polymer networks, defined by chain lengths and connectivity patterns, fundamentally influences their bulk properties.
  While existing polymer network models connect chain properties to emergent network behavior, they are often limited to monodisperse networks with regular connectivities.
  In this work, we introduce a novel modeling framework that shifts the focus from individual polymer chains to cross-links and their connected chains as the fundamental unit of analysis.
  The key features of this framework are the relaxation of the cross-link junction position to satisfy local force balance, and physically intuitive means for satisfying material frame indifference.
  \begin{hlbreakable}By relaxing to equilibrium, the cross-link structure is allowed to deform non-affinely and its chains adopt a more energetically efficient distribution of stretches and forces than that permitted from equal stretch or equal force homogenization theories.\end{hlbreakable}
  We explore two distinct limiting behaviors for the orientation of the frame of a cross-link: \begin{inparaenum}[(1)] \item the free rotation limit, which assumes the cross-link rotates to minimize free energy, and \item the frame averaging limit, which incorporates structural heterogeneity by averaging over all possible cross-link orientations. \end{inparaenum}
  It is found that an increase in variance in monomer numbers generally leads to network softening, while in bimodal networks, the onset of strain stiffening is controlled by shorter chains and the stiffening response is modulated by the ratio of short to long chains.
  By deriving closed-form approximations for both limits valid in the regimes of small deformation or small polydispersity, we offer an efficient computational approach to modeling the mechanics of complex, polydisperse networks.
  An aim of this framework is to take a step toward the rational modeling and design of heterogeneous polymer networks with structures tailored for specific properties.
\end{abstract}

\maketitle

\pagebreak
\tableofcontents

\section{Introduction} \label{sec:intro}

Constitutive models for solid polymer networks fall into two main categories: \begin{inparaenum}[(1)] \item phenomenological models based on strain invariants and principal stretches, and \item micromechanical models that connect macromolecular and network structural properties to the continuum scale. \end{inparaenum}
In the first category, there is the Neo-Hookean, Mooney-Rivlin~\cite{mooney1940theory,rivlin1948large}, and Ogden models~\cite{ogden1972large}.
In this work, we focus on micromechanical models.
Statistical mechanics help to elucidate the entropic underpinnings of the elasticity of many polymers~\cite{treloar1975physics}.
The continuum response of a broader polymer network can be constructed by using the single-chain elasticity in a polymer network model.
Network models allow one to relate macroscopic variables, such as deformation, to individual chains within the network.

Many polymer network models exist in the literature.
Some have a discrete number of chains arranged in a representative volume element (RVE) while others model the network via a probability density of polymer chains.
There are also competing assumptions for how macroscopic deformation is related to the deformations of chains within the network.
Traditional discrete network models include the $3$-chain~\cite{james1943theory}, the $4$-chain~\cite{treloar1943elasticity,treloar1943Belasticity,treloar1946elasticity,treloar1954photoelastic,flory1943statistical} (i.e., Flory-Rehner model), and the $8$-chain~\cite{arruda1993threee} models.
Continuous network models began with the full network model of Wu and van der Giessen~\cite{wu1992improved,wu1993improved} (inspired by~\cite{treloar1979non}) which assumes that all chains in the network are of the same length (in the reference configuration) and that the directions of chains are uniformly distributed over the unit sphere.
In response to the (at times) underwhelming fit of the full network model to benchmark data on rubber elasticity, it was improved upon by the microsphere model of Miehe et al.~\cite{miehe2004micro}.
A key feature of the microsphere model was the development of a new assumption for how macroscopic deformation is related to chains within the network. 
The assumption involved a ``stretch fluctuation field'' in the unit sphere of chain directions.
The stretch fluctuation field is determined by minimizing the average free energy of the network subject to homogenization-based constraints.
For certain fitted parameters, the microsphere model recovers the $8$-chain model exactly.
	
Since these earlier developments, a number of variants have been proposed that aim to retain much of the simplicity and physical-basis of earlier models while also providing a better fit to the `S'-shaped stress-strain curves of the well-known Treloar data for rubber elasticity~\cite{erman1980moments,bechir2010three,xing2026multi,miroshnychenko2009heuristic,xiang2018general,kroon20118,zhan2023new,davidson2013nonaffine,khiem2016analytical,itskov2024review,mirzapour2023micro,yang2025hyperelastic,li2026new,amores2019average,amores2021network,amores2023model,tan2024new,xing2025pseudo,elias2002constitutive,elias2006non,ouardi20253d}.
Many of these variants decompose the free energy density of the network additively into a contribution which represents the ``cross-linked network'' -- often modeled using the $8$-chain or a full network model -- and a topological constraint contribution due to chain entanglements within the network.
The contribution of topological constraints has taken different forms: some based on chain statistics and micromechanics~\cite{erman1980moments,xiang2018general,kroon20118,davidson2013nonaffine,khiem2016analytical,itskov2024review,mirzapour2023micro,yang2025hyperelastic,li2026new} while others appear more phenomenological~\cite{bechir2010three,xing2026multi,miroshnychenko2009heuristic}. 
For full network type models, a new macro-to-micro kinematic assumption has recently been developed, based on an affine stretch projection and a microscale analog of the Biot stress, which performs well for capturing complex and multiaxial stress-strain relationships~\cite{zhan2023new}.
Many constitutive models for rubber-like elasticity exist in the literature that attempt to balance between competing objectives: \begin{inparaenum}[(1)] \item contain as few fitting parameters as possible, \item have the ability to reproduce complex deformation behavior, and \item in the interest of informing material design, have model parameters which are physically interpretable and can be reasonably connected to the underlying microstructure of the material. \end{inparaenum}

Polymer network models have been applied to a vast number of different polymer networks related to adhesives~\cite{zhao2022network}, biomechanics~\cite{grekas2021cells,song2022hyperelastic,alastrue2009anisotropic}, dynamically bonded polymer networks~\cite{lamont2021rate}, magneto-active polymers~\cite{moreno2022effects} and electro-active polymers~\cite{grasingerIPtorque,grasingerIPflexoelectricity,grasinger2019multiscale,cohen2016electroelasticity,friedberg2023electroelasticity}.
Network models have also been used for capturing phenomena such as viscoelasticity~\cite{bergstrom1998constitutive,zhao2022network}, excluded volume induced strain-hardening and incompressibility~\cite{khandagale2023statistical}, and chain scission and fracture~\cite{mulderrig2021affine,mao2017rupture,buche2021chain}.
As a result, continuous network models have been generalized to anisotropic chain distributions~\cite{alastrue2009anisotropic,grasinger2020architected} and polydisperse networks (i.e., networks consisting of chains of different lengths)~\cite{mulderrig2021affine}.
However, many open questions persist regarding how to relate macroscopic deformation to chain-scale (or micro-) deformation in polydisperse polymer networks~\cite{mulderrig2021affine}.
Some standard assumptions such as the affine deformation assumption~\cite{treloar1975physics}, which assumes that chain end-to-end vectors get mapped under the deformation gradient, or equal stretch theories such as the $8$-chain model~\cite{arruda1993threee}, appear to overestimate the stiffness of polydisperse networks for small deformations, but under predict the onset of strain stiffening.
Alternatively, equal force theories assume that chains of different contour lengths that are aligned in the same direction sustain equal forces~\cite{von2002mesoscale,verron2017equal,li2020variational,mulderrig2021affine}.
Other interesting perspectives have emerged which consider more network-wide ensembles of monomers (as opposed to strictly chain-level ensembles) and take an Eulerian-inspired approach (as opposed to Lagrangian) at the continuum level~\cite{zhan2025statistical,wang2025statistical}.
Certain aspects of these various approaches seem more suitable for polydisperse polymer networks, but all lack explicit consideration of how chains of different lengths are ultimately connected together at cross-links and how this gives rise to different distributions of stresses within the network.
Connections cannot be made between network topology and continuum behavior.

There are alternative approaches that are also micromechanical in nature.
Explicit discrete network models\footnote{These are sometimes referred to as ``discrete network models''. We add the word ``explicit'' to make clear that this class of models are distinct from the polymer network models that are the main focus of this work.} offer a compelling alternative for connecting network topology and polydispersity to emergent bulk properties~\cite{araujo2025force,araujo2024micromechanical,araujo2026role,wagner2021network,wagner2022mesoscale,wagner2025foundational,lei2021mesoscopic,lei2022network,kothari2018mechanical,ghareeb2020adaptive,cardona2025topogen,huang2025topological,hartquist2025fracture,deng2023nonlocal,hartquist2025scaling}. These models represent the material as a lattice-like structure, where individual structural members exhibit hyperelastic behavior.
To probe the stress or failure response, the boundaries of a simulation box are deformed, and the system's energy is minimized. 
Recently, this type of approach has been extended to include the concept of ``tension blobs'', and utilizes statistical graph theory to characterize the trade-off between the coordination of the network in distributing loads, on the one hand, and entropy on the other~\cite{you2025network}.
Explicit discrete network models are particularly powerful for elucidating structure-property relationships, as they can capture complex, non-local interactions and detail how loads are transferred between disparate parts of the network, such as between regions with differing cross-link densities and functionalities. 
However, this non-local formulation results in a large system of coupled, nonlinear equations, making analytical solutions generally intractable and requiring (potentially computationally expensive) numerical methods. 
Another alternative evokes the concept of maximal advance paths, which makes a connection between a network's macroscopic affine path and its microscopic averaged deformation~\cite{tkachuk2012maximal,rastak2018non,tkachuk2022elastic}.
Again, this method can characterize non-local interactions across the network, but at the cost of requiring numerics.
In this work, we pivot to a more local description of polydisperse network mechanics. This allows us to make analytical progress and, where numerical solutions are still necessary, has the potential to result in a reduction in computational expense and complexity.

The primary focus of this work is on discrete polymer network models.
In contrast to standard polymer network models, which treat polymer chains as fundamental units of the network, here we shift perspective to consider cross-links as more fundamental.
In this context, a discrete network model is a direct analog of a cross-link.
The standard macro-to-micro kinematic assumption for discrete polymer network models -- motivated by coordinate system invariance -- is that the RVE rotates such that its edges lie along the principal directions of stretch.
However, recent theoretical work and experimental observations underscore the importance of treating the RVE orientation relative to a given deformation more carefully~\cite{kuhl2005remodeling,grasingerIPtorque,cohen2018generalized,grasinger2023polymer}.
In this work, we instead conceptualize the orientation of the RVE as thermally fluctuating. In other words, we consider the RVE as embedded in an elastic background which gives it a preferred orientation and a torsional stiffness.
We explore and compare the new model in two limits that are both interesting and analytically tractable: \begin{inparaenum}[(1)]
    \item in the limit of negligible torsional stiffness, the kinematic behavior of the new model reduces to the recently developed free rotation assumption~\cite{grasinger2023polymer} -- which assumes that the cross-link rotates into the frame that minimizes its free energy after deformation, and
    \item in the limit of large torsional stiffness, the model is equivalent to averaging over an isotropic distribution of frames (i.e., cross-link orientations).
\end{inparaenum}
\begin{hlbreakable}In each limit, the cross-link locally demonstrates non-affine deformation and its chains exhibit a more efficient load-sharing distribution than that from restrictive equal stretch or equal force theories.\end{hlbreakable}
Through this modeling framework, polydispersity in polymer chain length and inhomogeneous cross-link connectivity can be related to the behavior of the broader polymer network.

\paragraph*{Structure.} The structure of the work is as follows:
\begin{itemize}
\item In \Fref{sec:prelude}, as a prelude, some broadly used polymer chain models and concepts in continuum mechanics are reviewed.
\item \Fref{sec:new-model-1} develops a new framework for discrete polydisperse polymer network models by beginning with the statistical mechanics of a single cross-link embedded in an elastic background, and then following its implications up to the continuum scale through applications of statistical methods and approximations across scales.
\item \Fref{sec:FR-methods} and \Fref{sec:FA-methods} concern properties of, and approximation strategies for, the ``free rotation'' and ``frame averaging'' limits of the broader framework, respectively.
\item \Fref{sec:gaussian} explores the implications of polydispersity and topological differences on the elasticity of a bimodal network of Gaussian chains.
\item \Fref{sec:s-curves} considers how network structural properties give rise to different features in the characteristic `S'-shape stress-strain curves of polymer network elasticity out to large deformation regimes.
\item \Fref{sec:conclusion} concludes the work.
\item \begin{hlbreakable} As an aid to the reader, we provide a nomenclature table in \Fref{app:nomenclature}.
\end{hlbreakable}
\end{itemize}

\section{Prelude} \label{sec:prelude}
\paragraph*{Polymer chain statistical mechanics.} 
A statistical mechanics-based approach is often used as the foundation for the micromechanical modeling of polymer chain elasticity.
Consider a freely-jointed chain (FJC) of $\n$ rigid monomers\footnote{We here use the term ``monomers'', not in the chemical sense, but in the mechanical sense where the term is often used interchangeably with ``Kuhn segments'' or ``links''.} bonded end-to-end, each with a length $\mLen$.
This chain has a contour length, $\cLen = \n \mLen$.
This chain also has an associated end-to-end vector, $\rvec$ (i.e., the vector from the beginning of the chain to its end), and its magnitude, $\rmag = \left|\rvec\right|$, is called the end-to-end chain length.
The ratio of $\rmag$ to $\cLen$ defines the (absolute) chain stretch, $\chainStretch = \rmag / \cLen$ \hl{(see \Figref{fig:polydiserse-polymer-network-microstructure} for an illustration)}.
A statistical mechanics description begins by defining a probability density of polymer chain conformations, $\chainConformationProbDens(\rmag)$.\footnote{\hl{The freely-jointed and Gaussian chain models assume ideal chain statistics which is justified by the screening of excluded-volume interactions in the dense, solvent-free state for dry networks, and by theta-solvent conditions for solvated networks such as gels. Modeling good or poor solvent conditions would require modifying $\chainConformationProbDens\left(\rmag\right)$.}}
For freely-jointed Gaussian chains, 
\begin{equation} \label{eq:gauss-prob-density-conformations}
    \GaussChainConformationProbDens\left(\rmag\right) \propto \exp{\left(-\frac{3 \rmag^2}{2 \n \mLen^2}\right)}.
\end{equation}
For freely-jointed chains statistically described by the Kuhn and Gr\"{u}n approach~\cite{kuhn1942beziehungen},
\begin{equation} \label{eq:FJC}
	\KGChainConformationProbDens\left(\rmag\right) \propto \exp{\left(-\n \left(\frac{\rmag}{\cLen} \invLang\left(\frac{\rmag}{\cLen}\right) + \ln \left(\invLang\left(\frac{\rmag}{\cLen}\right) \csch \invLang\left(\frac{\rmag}{\cLen}\right) \right) \right)\right)},
\end{equation}
where $\invLang$ is the inverse of the Langevin function, $\Lang\left(x\right) = \coth x - 1/x$.\footnote{The details behind the numerical implementation of the inverse Langevin function are provided in \Fref{app:numerical-implementation-inv-langevin-func}.}
Notably, the Kuhn and Gr\"{u}n approach captures the finite extensibility of the chain; that is, $\KGChainConformationProbDens\left(\rmag\right)\rightarrow 0$ as $\rmag \rightarrow \cLen$.
Additional formulations for $\chainConformationProbDens(\rmag)$ have also been derived for the extensible freely-jointed chain~\cite{mulderrig2023statistical} and the primitive chain from the tube model~\cite{doi1988theory} (to name a few).

To determine the most probable end-to-end chain length from $\chainConformationProbDens(\rmag)$, we first need to consider the probability of finding the end of a polymer chain within a spherical shell (centered about the chain beginning at the origin) of radius $\rmag$ and thickness $d\rmag$~\cite{treloar1975physics}.
We denote this probability as $\chainDistanceProbDens(\rmag) d\rmag$,
\begin{equation} \label{eq:chain-dist-prob-dens}
    \chainDistanceProbDens(\rmag) d\rmag = 4 \pi \rmag^2 \chainConformationProbDens(\rmag) d\rmag.
\end{equation}
The most probable end-to-end chain length has been conventionally associated with the root-mean-square chain conformation; the corresponding $\rmsChainLength$ is given by
\begin{equation} \label{eq:rms-chain-length}
    \rmsChainLength = \sqrt{\frac{\int_{0}^{\critChainLength}\rmag^2 \chainDistanceProbDens(\rmag) d\rmag}{\int_{0}^{\critChainLength} \chainDistanceProbDens(\rmag) d\rmag}}.
\end{equation}
Note that $\chainDistanceProbDens(\rmag)$ is defined over the domain $\rmag\in[0,\critChainLength)$:
for Gaussian chains, $\critChainLength=\infty$;
for Kuhn and Gr\"{u}n chains, $\critChainLength=\cLen$;
for extensible freely-jointed chains, $\critChainLength\gtrapprox \cLen$ (e.g.,~\cite{mulderrig2023statistical}).
For Gaussian chains, $\rmsChainLength = \sqrt{n} \mLen$ (found as the expectation of the distance from the origin for a random walk of $\n$ steps, each of length $\mLen$)~\cite{treloar1975physics}.
For other chain models, $\rmsChainLength$ can be obtained by evaluating \Eqref{eq:rms-chain-length} numerically.

\paragraph*{Chain energetics.}

With a statistical description of polymer chain conformations at hand, we employ the principal thermodynamic connection formula $\chainFreeEnergy\left(\rmag\right) = -\kB \T \ln{\chainConformationProbDens\left(\rmag\right)}$ to yield the chain free energy $\chainFreeEnergy$.
For Gaussian chains and Kuhn and Gr\"{u}n chains, the free energy is respectively given by,
\begin{align} 
    \GaussFreeEnergy\left(\rmag\right) & = \frac{3}{2}\kB \T\frac{\rmag^2}{\n \mLen^2}, \label{eq:gauss-free-energy} \\
    \KGFreeEnergy\left(\rmag\right) & = \n \kB \T \left(\frac{\rmag}{\cLen} \invLang\left(\frac{\rmag}{\cLen}\right) + \ln \left(\invLang\left(\frac{\rmag}{\cLen}\right) \csch \invLang\left(\frac{\rmag}{\cLen}\right) \right) \right), \label{eq:kuhn-grun-free-energy}
\end{align}
where $\kB$ is Boltzmann's constant and $\T$ is the absolute temperature.
In agreement with the finite extensibility of the Kuhn and Gr\"{u}n chain, $\KGFreeEnergy \rightarrow \infty$ as $\rmag \rightarrow \cLen$.

In addition to $\GaussFreeEnergy$ and $\KGFreeEnergy$, there are a myriad of other chain free energies in the literature that are applicable to different types of polymers under various loading conditions:
extensible freely-jointed chains~\cite{mulderrig2021affine,buche2021chain,buche2022freely,mulderrig2023statistical} have been an important development for understanding the large deformation and fracture of polymer networks;
tube and slip-link chain models~(see relevant overviews in~\cite{darabi2021generalized,kumar2023tube}) have played an important role in capturing the effect of topological constraints (e.g., entanglements) in rubber elasticity;
the worm-like chain~\cite{marko1995stretching,kuhl2005remodeling,marantan2018mechanics} has found many applications in biopolymers;
and electrically responsive chains~\cite{grasinger2020statistical,grasinger2022statistical,grasingerIPflexoelectricity,cohen2016electroelasticity,friedberg2023electroelasticity,khandagale2024statistical} have been used to model to soft robotics, wearable electronics, etc. (to name a few).\footnote{For chains interacting with external fields, the symmetry of the free energy response may be broken such that its functional dependence is on $\rvec$ as opposed to $\rmag$. However, in this work, we restrict our attention to chains with isotropic elasticity.}

\paragraph*{Chain mechanics.}
Given the free energy response of a single chain, its mechanical response (i.e., the force along the chain) can be obtained directly by differentiation:
\begin{equation} \label{eq:force}
	\forcevec = \frac{\partial \chainFreeEnergy}{\partial \rvec}.
\end{equation}
The choice of a chain mechanics model (given the choice of chain free energy) along with the specification of a polymer network model are the basis for many multiscale models of polymer networks.
Such models seek to make connections between molecular-/meso-scale features and corresponding continuum-scale behaviors.

\paragraph*{Continuum mechanics.}
To establish notation, we introduce some fundamental concepts in continuum solid mechanics.
Let $\Xvec$ be a material point of a solid body in the reference configuration and $\xvec = \defMap\left(\Xvec\right)$ its corresponding point in the deformed configuration.
Here $\defMap$ is the mapping that describes the deformation of the solid body.
The deformation gradient, $\F = \takeGrad{\defMap}$, maps line elements -- which are infinitesimal changes in position, $\delta \xvec$ -- from the reference to the current configuration (i.e., $\delta \xvec = \F \delta \Xvec$).
The right Cauchy-Green tensor, $\cGreen = \trans{\F} \F$, describes the stretches of line elements under $\F$.
Physically, we require that $\det \F > 0$.
Then by the polar decomposition theorem, $\F$ can be uniquely decomposed into $\F = \polRot \strTens$ where $\polRot \in \SOThree$, $\strTens = \sqrt{\cGreen}$ is positive definite, and $\SOThree$ is the group of three-dimensional rotations (i.e., proper orthogonal transformations).
Moving forward, we refer to $\strTens$ as the right stretch tensor and its eigenvalues and eigenvectors as the principal stretches, $\pStretch{i}$, and principal directions, $\pDir_i$, respectively.
Because $\strTens$ is symmetric, the principal directions can be chosen such that $\pDir_i \cdot \pDir_j = \delta_{ij}$ and this set of directions constitutes a principal frame.
The tensor
\begin{equation}
	\pFrame = \begin{pmatrix}
		| & | & | \\
		\pDir_1 & \pDir_2 & \pDir_3 \\
		| & | & |
	\end{pmatrix}
\end{equation}
is a rotation that diagonalizes the stretch tensor, $\strDiag = \pFrame \strTens \trans{\pFrame} = \diag \left(\pStretch{1}, \pStretch{2}, \pStretch{3}\right)$.

Analogous to the single chain case, the elasticity of the solid polymer network can be modeled via a free energy density $\freeEnergyDensity = \freeEnergyDensity\left(\F\right)$ where the nominal stress (i.e., first Piola–Kirchhoff stress) is given by $\PK = \partial \freeEnergyDensity / \partial \F$.

\section{Polydisperse polymer network models} \label{sec:new-model-1}

Consider a polydisperse polymer network, such as the \hl{examples} depicted in \Figref{fig:polydiserse-polymer-network-microstructure}.
\begin{hlbreakable}Polydispersity and topological irregularity arise in polymer networks from the polymerization and cross-linking procedures taken to synthesize these materials.
From this, several sources of irregularity emerge, as follows: \begin{inparaenum}[(1)] \item scattered chain orientations, \item polydispersity in monomer number, \item distribution of initial chain conformations (e.g., chain pre-stretches~\cite{araujo2024micromechanical,araujo2026role}), and \item variation in the number of chains connected together at cross-linker sites (i.e., variation in cross-linker site topology). \end{inparaenum}\end{hlbreakable}
To simplify the scope of this work, we will explicitly account for phenomena (1), (2), and (4) in our modeling formulation. 
Incorporating phenomenon (3) into our formulation remains an open problem for future research.

We begin our modeling approach by first characterizing the space of polydisperse polymer network cross-links. We do so by the set $\chainSpace = \collection{\inst_i}_{i=1}^{\left|\chainSpace\right|}$,\footnote{The set notation used here is shorthand for $\collection{x_i}_{i=1}^n \equiv \collection{x_i: 1 \leq i \leq n} = \collection{x_1, \dots, x_n}$. We use this shorthand set notation throughout this work, and also analogously employ this shorthand notation to tuples where convenient.}  
where each cross-link, $\inst$, is a representative volume element (RVE). 
We further characterize each individual cross-link with its own set, which consists of the chains that it connects.
For FJCs, this takes the form $\inst = \collection{\left(\n_i, \cLen_i, \Xvec_i, \chainFreeEnergy_i\right)}_{i=1}^{\numChains}$,
where $\n_i$ is the number of monomers of the $i$th chain, $\cLen_i$ is its contour length, $\Xvec_i$ is its end position (in the undeformed body), $\chainFreeEnergy_i$ is its free energy function (which will depend, in general, on $\n_i, \cLen_i$, etc.), and where $\numChains$ is the number of chains that are connected at the junction point, $\yc$.\begin{hlbreakable}\footnote{\begin{hlbreakable}Here, we account for the possibility that the monomer length $\mLen$ (which $\cLen$ is sensitive to) and the chain free energy function description $\chainFreeEnergy$ (e.g., $\GaussFreeEnergy$ versus $\KGFreeEnergy$) may vary between the chains in the cross-link.
Variation in $\mLen$ and $\chainFreeEnergy$ could arise from a diversity in the chemical composition of the chains used to synthesize the network.
Because of this, the approach can readily be used for polymer networks composed of diverse blends of chain chemistries.\end{hlbreakable}}\end{hlbreakable}
It is required that $\left|\Xvec_i - \Xvec_j\right| \leq \cLen_i + \cLen_j \text{ for all } i \neq j$.
Let $\RVEdomain$ be the domain of the RVE associated with the cross-link in its undeformed reference state; and, as such, it is defined as the convex hull of the chain ends, $\collection{\Xvec_i}_{i=1}^{\numChains}$: $\RVEdomain = \Conv\left({\collection{\Xvec_i}_{i=1}^{\numChains}}\right) \subset \Reals^3$. 
Note an otherwise implicit assumption: which end of the chain is its ``beginning'' and which is its ``end'' is arbitrary.
By convention, we say the chain ends are at the boundary of the RVE while the chain beginnings are taken to be at the junction point itself.
Then each chain end-to-end vector in the RVE is given by $\rvec_i = \Xvec_i - \yc$.
The stochastic structure of the network is then encoded by a probability density function, $\probDensInst : \inst \mapsto \left[0, \infty\right)$, that describes the likelihood of the network having a cross-link structure, $\inst \in \chainSpace$.
The function is normalized such that $\int \df{\inst} \: \probDensInst = 1$.
\begin{figure}
	\centering
	\includegraphics[width=\linewidth]{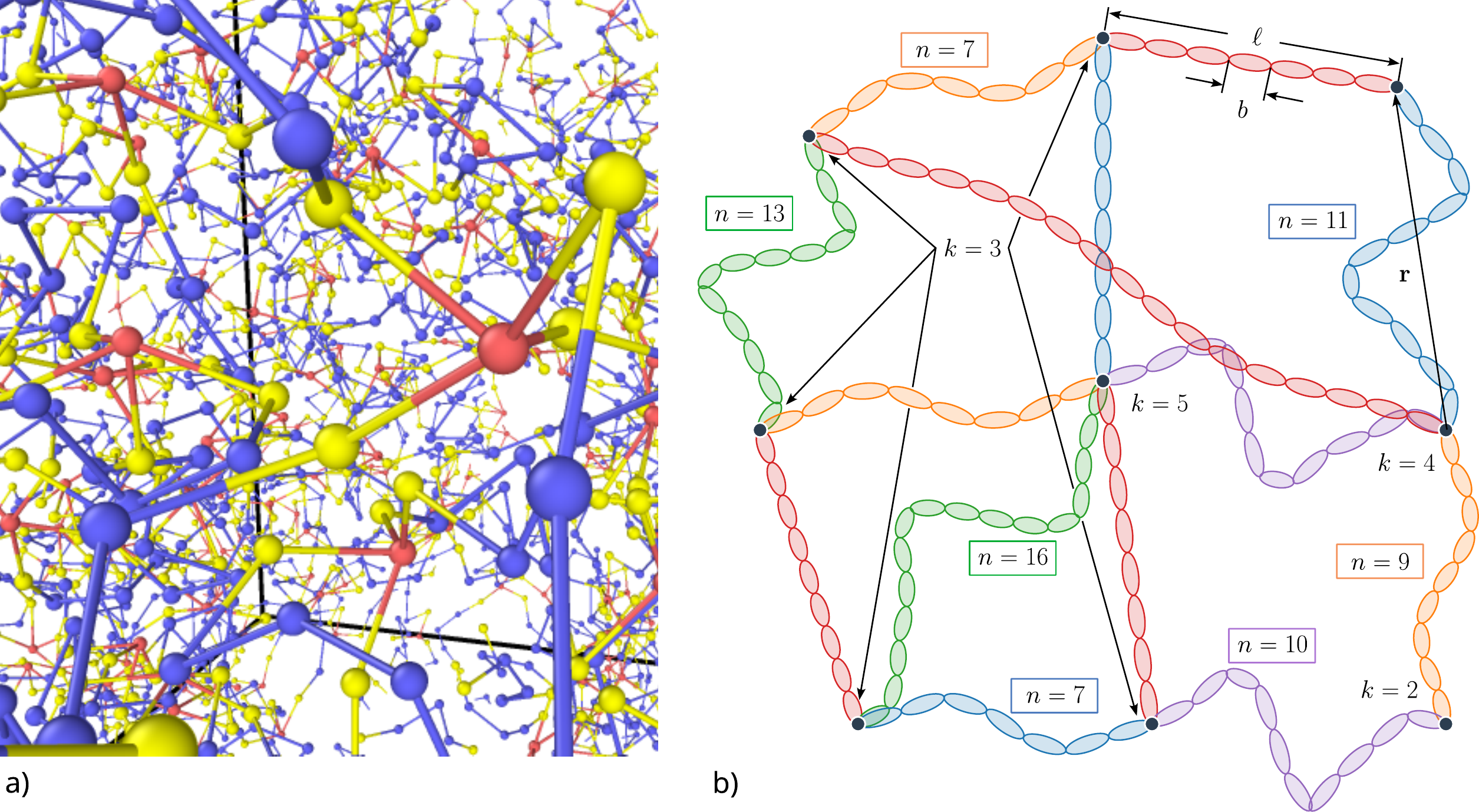}
	\caption{
            \capTitle{\hl{Polydisperse polymer networks.}}
                a) \hl{A realistic polydisperse polymer network microstructure generated via \emph{in silico} network synthesis.} The network is constructed by end-linking coarse-grained precursor chains in the \texttt{LAMMPS} molecular dynamics simulation package~\cite{plimpton1995fast,thompson2022lammps}, following procedures modified from those developed by Riggleman and colleagues~\cite{ye2020molecular,barney2022fracture,zhang2024predicting}. Red beads represent cross-linker sites, yellow beads represent precursor chain end monomers with the capability of bonding to cross-linker sites, and blue beads represent inner-chain monomers. 
                \hl{b) Schematic illustration of the structural features relevant to this work. Junctions have varying coordination number, $\numChains$ (i.e., differing numbers of chains meet at each junction) and chains have varying numbers of monomers, $\n$. As a result, chains have differing contour lengths, $\cLen$, which is the taut length of the chain and, for freely jointed chains, is given by $\n \mLen$ where $\mLen$ is the monomer length. The end-to-end vector, $\rvec$, for a single chain is also shown.}
	}
	\label{fig:polydiserse-polymer-network-microstructure}
\end{figure}

\subsection{Cross-link mechanics} \label{sec:cross-link-mechanics}
The crux of our discrete polymer network modeling approach is that we treat each cross-link as an RVE, which responds to external loads, and then model the response of the network as an average over the statistical distribution of cross-links.
Three major conceptual challenges remain: \begin{inparaenum}[(1)] \item How should the probability density function, $\probDensInst$, be parameterized in practice? \item How should a cross-link, $\inst$, be physically instantiated? And \item given some $\inst$, how does it response to external loads? \end{inparaenum}
To make progress, we make the following assumptions:
\begin{enumerate}[(1)]
\item \label{ass:chain-ends}
    The (average) chain ends connected to a cross-link, $\inst$, are well-defined and at $\collection{\Xvec_i}_{i=1}^{\numChains}$ in the undeformed body.
    Further, the restraints imposed on a particular cross-link $\inst$ by the remainder of the polymer network can be accurately captured by fixing the nearest neighbor cross-links to their average positions.\footnote{In general, the nearest neighbor cross-links may not be nearest in space, but rather, they are nearest with respect to the continuous network topological structure. Conversely, it is conceivable that cross-links that are distantly connected along the network structure may interpenetrate the same physical space. For a rich discussion of chain interpenetration in polymer networks, see~\cite{flory1985molecular,flory1985network}.}
\item \label{ass:RVE-formation}
    When a cross-link forms, \begin{inparaenum}[(a)] \item chain conformations have their expected length (given by their $\rmsChainLength$) from the initial junction point, and \item due to entropic and excluded volume effects, chains orientations are ``maximally spaced out'' (provided that the orientations may also achieve a force balance, i.e., are not linearly independent); this is made more precise in \Fref{sec:cross-link-structures} and \Fref{app:polydisperse-cross-link-structures}, where details regarding instantiating the chain end positions are discussed. \end{inparaenum}
\item \label{ass:isotropy}
    \begin{hlbreakable}Within a polymer network, the chain properties and RVE geometry of a cross-link are independent of its orientation; that is, \begin{equation} \label{eq:crosslink-isotropy}
        \probDensInst\left(\collection{\left(\n_i, \cLen_i, \genRotRef \Xvec_i, \chainFreeEnergy_i\right)}_{i=1}^{\numChains}\right) = 
        \probDensF{\genRotRef}\left(\genRotRef\right) \probDens\left(\collection{\left(\n_i, \cLen_i, \Xvec_i, \chainFreeEnergy_i\right)}_{i=1}^{\numChains}\right) = \probDensF{\genRotRef}\left(\genRotRef\right)\probDens\left(\inst\right)
    \end{equation}
    where $\probDens\left(\inst\right)$ is the frame-invariant cross-link structure probability density,
    and $\probDensF{\genRotRef}\left(\genRotRef\right)$ is the probability density of the cross-link having some orientation, $\genRotRef \in \SOThree$, 
    relative to a reference set of chain end positions, $\collection{\Xvec_i}_{i=1}^{\numChains}$.\end{hlbreakable}
    Here, when referring to the ``orientation'' of a cross-link, we mean its preferred orientation within an elastic background; that is, it is the ground state of its orientation.
    It is important that $\probDensF{\genRotRef}\left(\genRotRef\right)$ satisfy frame indifference and the underlying symmetry of the material.
    We are interested in isotropic materials in this work and, thus, we will restrict our attention to \begin{equation} \label{eq:prob-dens-preferred-clnk-orientation}
        \probDensF{\genRotRef}\left(\genRotRef\right) = \text{const.} = \frac{1}{\left|\SOThree\right|} = \frac{1}{\volSOThree}
    \end{equation} 
    such that all cross-link orientations are assumed equally likely.

    \begin{hlbreakable}At this point, the parameterization of $\probDensInst$ and $\probDens\left(\inst\right)$ remain unspecified. We discuss a methodology for parameterizing each in \Fref{sec:cross-link-statistics}.\end{hlbreakable}
\item \label{ass:affine}
    Macroscopic network deformation affinely displaces the relative positions of the chain ends (and, as a consequence, neighboring cross-links), $\collection{\Xvec_i}_{i=1}^{\numChains}\rightarrow\collection{\F \Xvec_i}_{i=1}^{\numChains}$. Meanwhile, the cross-link junction is allowed to occupy all possible physically-permissible positions in response to random thermal fluctuations of the connected chains.
\item \label{ass:fluctations}
    Thermal fluctuations of the network can be characterized through \begin{enumerate}[(a)]
        \item fluctuations of the internal degrees of freedom of each chain (e.g., for a FJC, the link orientations),
        \item rigid rotations of the cross-link such that $\collection{\Xvec_i}_{i=1}^{\numChains}\rightarrow\collection{\genRot \Xvec_i}_{i=1}^{\numChains}$ for $\genRot \in \SOThree$ in the undeformed body, and $\collection{\F \Xvec_i}_{i=1}^{\numChains}\rightarrow\collection{\F \genRot \Xvec_i}_{i=1}^{\numChains}$ in the deformed body,
        and \item displacements of the junction point, $\yc$.
    \end{enumerate}
    These fluctuation modes are depicted in \Figref{fig:fluctuations}.
    It is assumed that all other fluctuation modes have a negligible contribution to the network behavior (e.g., because they are effectively inaccessible due to higher energy or other constraints).
\item \label{ass:embedded}
    Cross-links are embedded in an elastic background that gives them a preferred orientation and a torsional stiffness with respect to the preferred orientation.
\item \label{ass:thermalize}
    The internal degrees of freedom of each of the chains thermalize much more quickly than network fluctuations; that is, fluctuations due to rotations of the cross-link or displacements of the junction point.
    This assumption is justified by observing that rotating the cross-link or displacing its junction point is coupled to all of the degrees of freedom of the chains.
    The cross-link is a larger, more complex structure than its individual chains.
    Relaxation times of polymers often scale with their size and the complexity of their structure~\cite{rouse1953theory,doi1996introduction}.
    Thus, one would expect that deformation modes of the RVE related to the global structure of the cross-link should have longer relaxation times.
    The result of this assumption is two-fold: \begin{inparaenum}[(1)] \item the free energy contributions from the chains can be taken as a sum of their individual free energies, $\sum_{i=1}^\numChains \chainFreeEnergy_i$, and \item in the statistical mechanics formulation of the network, the internal chain degrees of freedom can be treated separately before averaging over fluctuating cross-link rotations and junction positions. \end{inparaenum}
\end{enumerate}
Assumptions \ref{ass:chain-ends}-\ref{ass:isotropy} will facilitate the construction of cross-links and associated cross-link probability density function, $\probDensInst$.
Assumptions \ref{ass:affine}-\ref{ass:thermalize} will facilitate modeling the mechanical response of a cross-link.
\begin{figure}
	\centering
	\includegraphics[width=\linewidth]{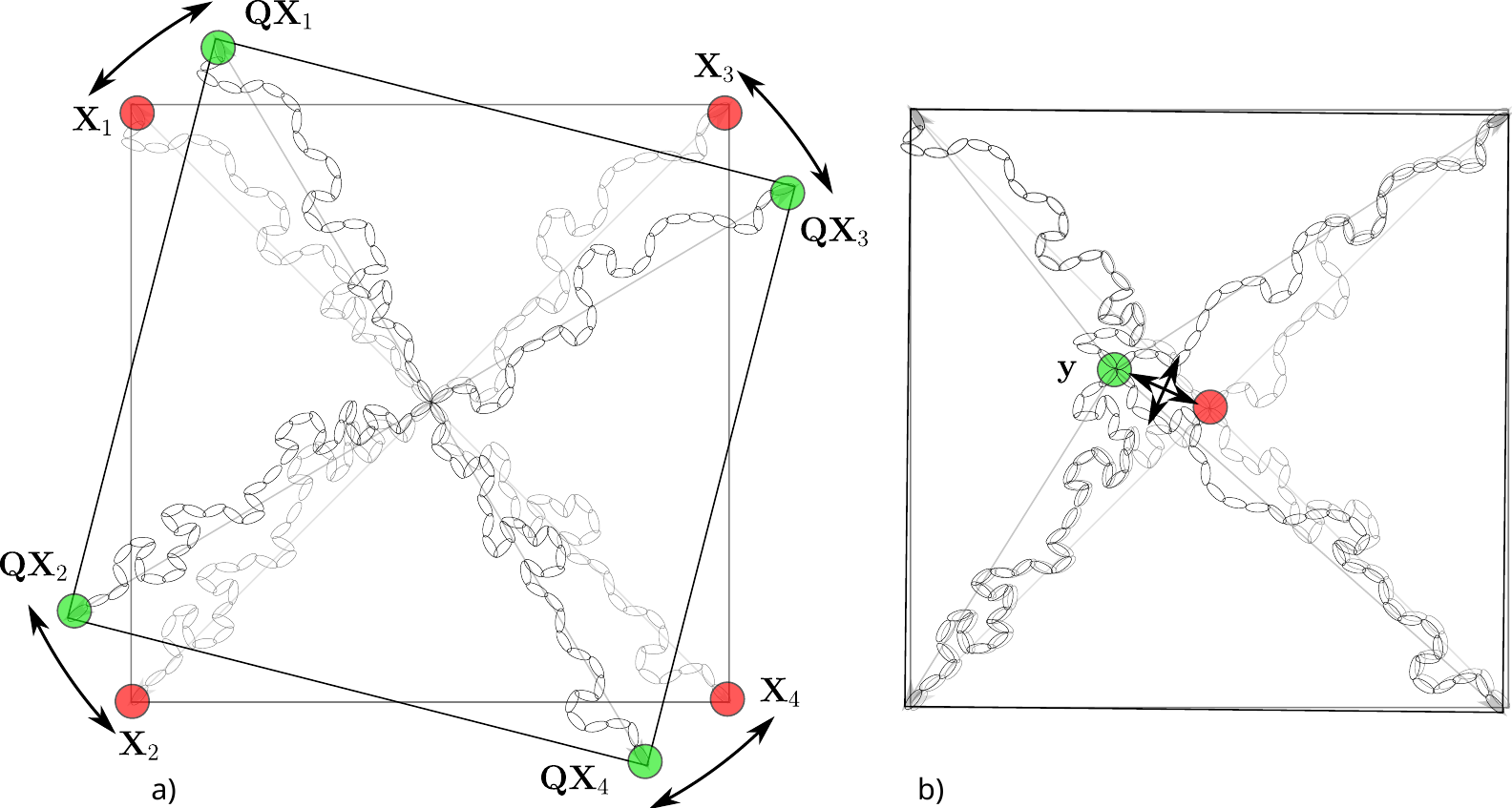}
	\caption{
            \capTitle{Thermal fluctuations of cross-links.}
            Thermal fluctuations cause a) rigid rotations of the end positions of the cross-link, $\Xvec \rightarrow \genRot \Xvec$ and b) perturbations of the junction position, $\yc$.
	}
	\label{fig:fluctuations}
\end{figure}

\paragraph*{Statistical mechanics of a cross-link.}
The free energy of a thermally fluctuating system at constant temperature is obtained (via the principal thermodynamic connection formula) as $-\kB \T \ln \partitionFunction$ where $\partitionFunction$ is the partition function.
Let $\rvedomain = \Conv\left({\collection{\F \genRot \Xvec_i}_{i=1}^{\numChains}}\right) \subset \Reals^3$ be the rotated and deformed RVE domain.
Given Assumptions \ref{ass:affine}-\ref{ass:thermalize}, the partition function of a cross-link is given by
\begin{equation} \label{eq:partition-function-1}
   \partitionFunction\left(\F, \genRotRef\right) = \int_{\yc \in \rvedomain} \df{\yc} \int_{\genRot \in \SOThree} \df{\genRot} \: \exp\left(-\frac{1}{\kB \T}\sum_{i=1}^\numChains \chainFreeEnergy_i\left(\right|\F \genRot \Xvec_i - \yc\left|\right) - \frac{\PEQ\left(\genRot\right)}{\kB \T}\right)
\end{equation}
where $\PEQ\left(\genRot\right)$ is the torsional elastic energy from the surrounding network when rotating the cross-link by some amount $\genRot$ prior to deformation.
Here we assume a simple form for the torsional elasticity as
\begin{equation}
    \PEQ\left(\genRot\right) = \frac{\kNetTors}{2} \lVert \genRot - \genRotRef \rVert^2
\end{equation}
where $\kNetTors$ is some torsional stiffness modulus, and recall: $\genRotRef$ is the preferred orientation of the cross-link in the undeformed body relative to the reference positions, $\collection{\Xvec_i}_{i=1}^\numChains$.
Even with this simplified form of $\PEQ$, exact evaluation of \Eqref{eq:partition-function-1} is prohibitively difficult.
Instead, we use saddle point approximation to explore two interesting limits that are both more analytically tractable.
Towards approximating \Eqref{eq:partition-function-1}, let
\begin{align}
    \label{eq:FR-min}
    \left\{\genRotStar, \ycStar\right\} &= \arg \inf_{\genRot \in \SOThree, \yc \in \rvedomain} \sum_{i=1}^{\numChains} \chainFreeEnergy_i\left(\left|\F \genRot \Xvec_i - \yc\right|\right), \\
    \label{eq:FA-min}
    \ycStarQO &= \arg \inf_{\yc \in \rvedomain} \sum_{i=1}^{\numChains} \chainFreeEnergy_i\left(\left|\F \genRotRef \Xvec_i - \yc\right|\right).
\end{align}
The two limits of interest are: \begin{inparaenum}[(1)] \item the \emph{free rotation limit} where, for each $\genRotRef \in \SOThree$, there exists a $\genRotStar$ such that \Eqref{eq:FR-min} and
\begin{equation} \label{eq:free-rotation-limit}
    \frac{\kNetTors}{2 \kB \T} \lVert \genRotStar - \genRotRef \rVert^2 \ll 1,
\end{equation}
are both satisfied\footnote{
This condition does not necessarily imply $\kNetTors / \kB \T \ll 1$. It could alternatively be met through high enough degeneracy of $\genRotStar$ (e.g., due to symmetry) such that for each $\genRotRef$ there is a $\genRotStar$ that is sufficiently close.
}; and \item
the \emph{frame averaging limit} where
\begin{equation} \label{eq:frame-averaging-limit}
    \frac{\kNetTors}{2 \kB \T} \gg 1.
\end{equation}
\end{inparaenum}
In the free rotation limit, the chain energy term (i.e., first term) dominates the argument of the exponential in \Eqref{eq:partition-function-1}, such that its minimum determines $\genRot$ at the saddle point; in the frame averaging limit, the torsional elastic energy term dominates, and it is instead this term whose minimum determines $\genRot$ at the saddle point.

\paragraph*{Free rotation limit.}
First assume \Eqref{eq:FR-min} and \Eqref{eq:free-rotation-limit} hold.
In this limit, one can take the saddle point approximation at $\genRot = \genRotStar,~\yc = \ycStar$.
Taking the leading order approximation with respect to $\genRot$ results in,
\begin{equation} \label{eq:partition-function-2}
   \partitionFunctionFR\left(\F, \genRotRef\right) \approx e^{-\PEQ\left(\genRotStar\right) / \kB \T} \int_{\yc \in \rvedomain} \df{\yc} \: \exp\left(-\frac{1}{\kB \T}\sum_{i=1}^\numChains \chainFreeEnergy_i\left(\right|\F \genRotStar \Xvec_i - \yc\left|\right)\right).
\end{equation}
For brevity, let $\Wchains\left(\F, \genRot, \yc\right) = \sum_{i=1}^\numChains \chainFreeEnergy_i\left(\right|\F \genRot \Xvec_i - \yc\left|\right)$.\footnote{The analytical form of the first and second derivatives of $\Wchains$ with respect to $\yc$ is provided in \Fref{app:cross-link-chain-free-energy-derivatives-junction-position} for both Gaussian chains and Kuhn and Gr\"{u}n chains.}
Expanding about $\yc = \ycStar$,
\begin{equation} \label{eq:partition-function-3}
   \partitionFunctionFR\left(\F, \genRotRef\right) \approx e^{-\left(\Wchains\left(\F, \genRotStar, \ycStar\right) + \PEQ\left(\genRotStar\right)\right) / \kB \T} \int_{\yc \in \rvedomain} \df{\yc} \: \exp\left(-\frac{1}{2 \kB \T} \left(\yc - \ycStar\right) \cdot \evalAt{\frac{\partial^2 \Wchains}{\partial \yc \partial \yc}}{\genRot = \genRotStar, \yc = \ycStar} \left(\yc - \ycStar\right)\right),
\end{equation}
where higher order terms, $\orderOf{\left(\yc - \ycStar\right)^3}$, have been neglected.
Extending the domain of integration to $\yc \in \Reals^3$ results in Gaussian integrals such that
\begin{equation} \label{eq:partition-function-free-rotation-4}
   \partitionFunctionFR\left(\F, \genRotRef\right) \approx \left(\frac{\left(2 \pi \kB \T\right)^{3/2}}{\sqrt{\det \left(\evalAt{\frac{\partial^2 \Wchains}{\partial \yc \partial \yc}}{\genRot = \genRotStar, \yc = \ycStar}\right)}}\right) \exp\left(-\frac{\Wchains\left(\F, \genRotStar, \ycStar\right)}{\kB \T} - \frac{\PEQ\left(\genRotStar\right)}{\kB \T}\right) 
\end{equation}
where $\det \left(\partial^2 \Wchains / \partial \yc \partial \yc\right)$ is the determinant of the Hessian of the free energy of the chains with respect to $\yc$.
Then the free energy of a cross-link in the free rotation limit is given by
\begin{equation} \label{eq:crosslink-free-energy-free-rotation-0}
    \boxed{
   \clinkFreeRotationFreeEnergy\left(\F, \genRotRef\right) = \Wchains\left(\F, \genRotStar, \ycStar\right) + \PEQ\left(\genRotStar\right) + \frac{\kB \T}{2} \ln \left(\det \left(\evalAt{\frac{\partial^2 \Wchains}{\partial \yc \partial \yc}}{\genRot = \genRotStar, \yc = \ycStar}\right)\right) - \frac{3 \kB \T}{2} \ln \left(2 \pi \kB \T\right).
   }
\end{equation}
For simplicity, terms that have negligible effect on the mechanics are dropped.
\begin{enumerate}[(1)]
\item The last term has no influence on the mechanics.
\item The second term is also negligible, by assumption (\Eqref{eq:free-rotation-limit}); and, further, it will be shown that for many networks of interest $\genRotStar$ does not change significantly with deformation; thus, even when $\kNetTors / 2 \kB \T$ is not small, the second term likely has a negligible coupling to mechanics.\footnote{The last term may influence thermal properties and the second term may be more relevant for certain multiphysics problems \hl{(e.g., when an external electric field induces net torques on dipolar chains, coupling the rotational degree of freedom to both the applied field and deformation)}.}
\item The second-to-last term is related to fluctuations of the cross-link junction position, and is more subtle.
It has been argued in previous work that it has no effect on the mechanics of networks of Gaussian chains~\cite{treloar1975physics}, and that its significance for networks with other types of chains (e.g., Kuhn and Gr\"un chains) is also small~\cite{treloar1946elasticity,treloar1954photoelastic}.
Unless otherwise stated, this junction fluctuation term is also dropped.
It will, however, be revisited and analyzed when appropriate.
\end{enumerate}
After simplifying, \Eqref{eq:crosslink-free-energy-free-rotation-0} takes the form
\begin{equation} \label{eq:crosslink-free-energy-free-rotation}
\boxed{
        \clinkFreeRotationFreeEnergy\left(\F\right) = \Wchains\left(\F, \genRotStar, \ycStar\right) = \inf_{\genRot \in \SOThree, \yc \in \rvedomain} \sum_{i=1}^{\numChains} \chainFreeEnergy_i\left(\left|\F \genRot \Xvec_i - \yc\right|\right).
        }
\end{equation}

\paragraph*{Frame averaging limit.}
Assume \Eqref{eq:FA-min} and \Eqref{eq:frame-averaging-limit} hold.
In this limit, one can take the saddle point approximation at $\genRot = \genRotRef,~\yc = \ycStarQO$.
Taking the leading order approximation with respect to $\genRot$ results in,
\begin{equation} \label{eq:partition-function-FA-2}
   \partitionFunctionFA\left(\F, \genRotRef\right) \approx \int_{\yc \in \rvedomain} \df{\yc} \: \exp\left(-\frac{\Wchains\left(\F, \genRotRef, \yc\right)}{\kB \T}\right).
\end{equation}
Expanding about $\yc = \ycStarQO$,
\begin{equation} \label{eq:partition-function-4}
   \partitionFunctionFA\left(\F, \genRotRef\right) \approx e^{-\Wchains\left(\F, \genRotRef, \ycStarQO\right) / \kB \T} \int_{\yc \in \rvedomain} \df{\yc} \: \exp\left(-\frac{1}{2 \kB \T} \left(\yc - \ycStarQO\right) \cdot \evalAt{\frac{\partial^2 \Wchains}{\partial \yc \partial \yc}}{\genRot = \genRotRef, \yc = \ycStarQO} \left(\yc - \ycStarQO\right)\right), 
\end{equation}
where higher order terms, $\orderOf{\left(\yc - \ycStarQO\right)^3}$, have been neglected.
Extending the domain of integration to $\yc \in \Reals^3$ results in Gaussian integrals such that
\begin{equation} \label{eq:partition-function-frame-average-4}
   \partitionFunctionFA\left(\F, \genRotRef\right) \approx \left(\frac{\left(2 \pi \kB \T\right)^{3/2}}{\sqrt{\det \left(\evalAt{\frac{\partial^2 \Wchains}{\partial \yc \partial \yc}}{\genRot = \genRotRef, \yc = \ycStarQO}\right)}}\right) \exp\left(-\frac{\Wchains\left(\F, \genRotRef, \ycStarQO\right)}{\kB \T}\right)
\end{equation}
and
\begin{equation} \label{eq:crosslink-free-energy-frame-average-0}
\boxed{
   \clinkFrameAveragingFreeEnergy\left(\F, \genRotRef\right) = \Wchains\left(\F, \genRotRef, \ycStarQO\right) + \frac{\kB \T}{2} \ln \left(\det \left(\evalAt{\frac{\partial^2 \Wchains}{\partial \yc \partial \yc}}{\genRot = \genRotRef, \yc = \ycStarQO}\right)\right) - \frac{3 \kB \T}{2} \ln \left(2 \pi \kB \T\right).
   }
\end{equation}
Similar to \Eqref{eq:crosslink-free-energy-free-rotation}, the last term is dropped for simplicity, and because it has no bearing on the mechanical behavior.
The second term (i.e., the cross-link junction fluctuation term) is also dropped -- both for simplicity and because, again, it has been found that it often has a negligible effect on mechanical behavior~\cite{treloar1975physics,treloar1946elasticity,treloar1954photoelastic}.
After simplifying, \Eqref{eq:crosslink-free-energy-frame-average-0} takes the form
\begin{equation} \label{eq:crosslink-free-energy-frame-average}
\boxed{
        \clinkFrameAveragingFreeEnergy\left(\F, \genRotRef\right) = \Wchains\left(\F, \genRotRef, \ycStarQO\right) = \inf_{\yc \in \rvedomain} \sum_{i=1}^{\numChains} \chainFreeEnergy_i\left(\left|\F \genRotRef \Xvec_i - \yc\right|\right).
        }
\end{equation}

\begin{figure}
	\centering
	\includegraphics[width=0.85\linewidth]{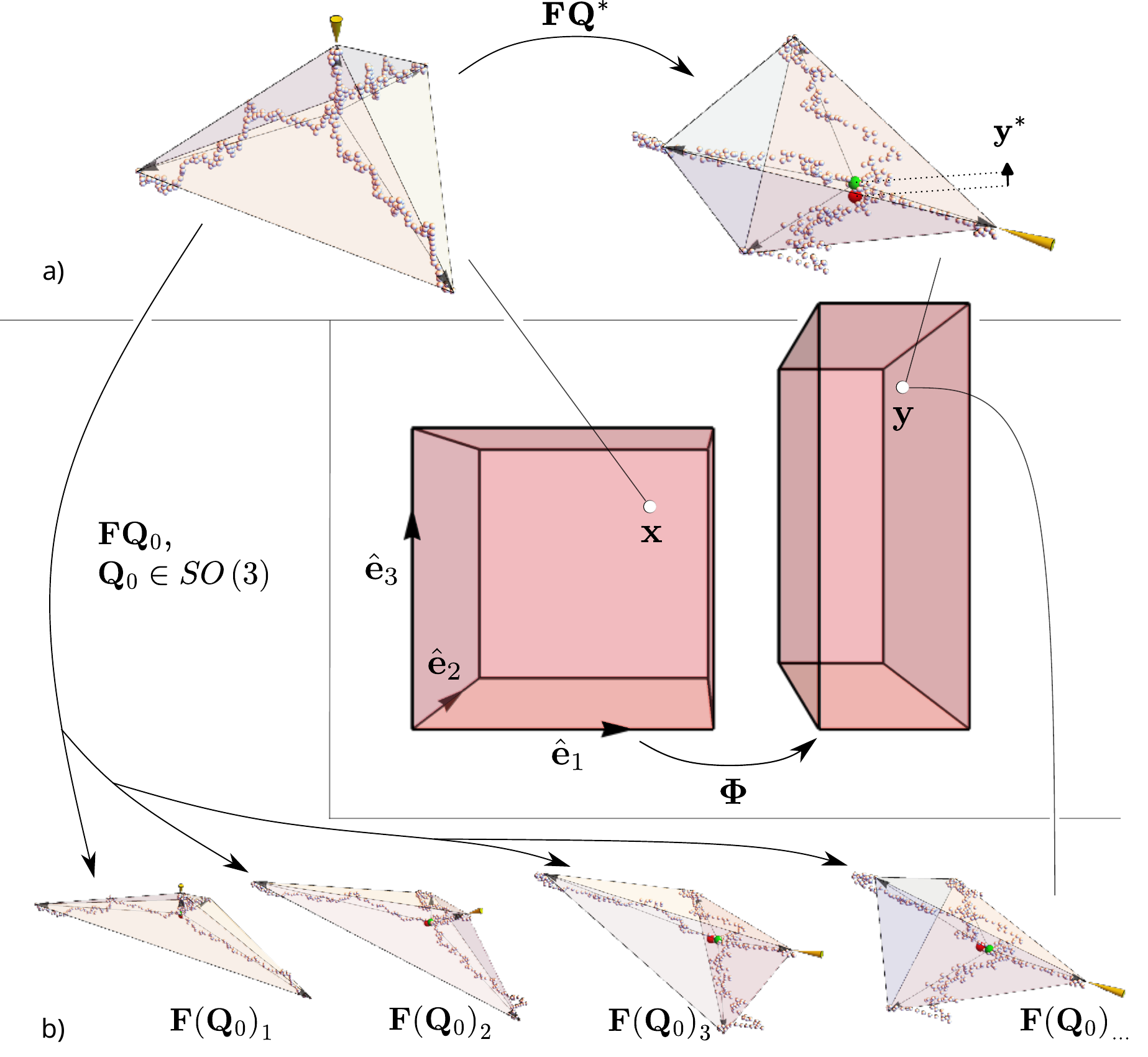}
	\caption{
    \begin{hlbreakable}\capTitle{Cross-link RVE response to deformation.}\end{hlbreakable}
    a) The polydisperse $4$-chain RVE gets mapped under $\F \genRotStar$ for the free rotation limit.
    In the deformed RVE, the origin and cross-link position are denoted by a red ball and green ball, respectively.
    b) In the frame averaging limit, the RVE response to \hl{(fixed)} $\F$ is averaged over all possible frames, represented by $\genRotRef \in \SOThree$.
    The chain deformations, stretches, and cross-link position are, in general, different for each $\left(\genRotRef\right)_\Box$.}
	\label{fig:constitutive-models}
\end{figure}

\begin{hlbreakable}For both the free rotation limit and the frame averaging limit, the cross-link RVE exhibits locally non-affine deformation and its chains share load in a generalized manner independent from restrictive homogenization assumptions.
These effects emerge from affinely deforming each chain end while allowing the cross-link structure to relax to equilibrium.
\Figref{fig:constitutive-models} displays a schematic physically depicting the cross-link deformation for the free rotation and frame averaging limits.\end{hlbreakable}
In addition, the nondimensional representation of the cross-link partition function (\Eqref{eq:partition-function-free-rotation-4} and \Eqref{eq:partition-function-frame-average-4}) and cross-link free energy (\Eqref{eq:crosslink-free-energy-free-rotation-0} and \Eqref{eq:crosslink-free-energy-frame-average-0}) for both limits is provided in \Fref{app:nondimensional-free-energy-representation}.

\begin{remark}
\emph{Frame averaging limit for isotropic materials.}
For isotropic materials, it is often convenient to represent the cross-link free energy of the frame averaging limit as the integral of \Eqref{eq:crosslink-free-energy-frame-average} over all $\genRotRef$.
Using \Eqref{eq:prob-dens-preferred-clnk-orientation} from Assumption \ref{ass:isotropy}, this leads to the following form
\begin{equation} \label{eq:crosslink-free-energy-frame-average-integral-orientations}
\boxed{
    \clinkFrameAveragingFreeEnergy\left(\F\right) = \frac{1}{\volSOThree} \int_{\genRotRef \in \SOThree} \df{\genRotRef} \left(\inf_{\yc \in \rvedomain} \sum_{i=1}^{\numChains} \chainFreeEnergy_i\left(\left|\F \genRotRef \Xvec_i - \yc\right|\right)\right).
    }
\end{equation}
In a similar fashion, the frame averaging over all $\ycStarQO$ (from \Eqref{eq:FA-min}) is given by
\begin{equation} \label{eq:crosslink-displacement-frame-average-integral-orientations}
\ycStarFrameAveraging\left(\F\right) = \frac{1}{\volSOThree} \int_{\genRotRef \in \SOThree} \df{\genRotRef} \: \ycStarQO.
\end{equation}
\end{remark}

\begin{remark}
\emph{Applicability of the free rotation and frame averaging limits.}
    The free rotation and frame averaging limits may be more or less applicable to different types of polymer networks.
    For example, some biopolymer networks consisting of semi-flexible polymers with bending rigidity may resist shear and rotation at cross-links~\cite{storm2005nonlinear,broedersz2014modeling,hatami2018effect} and, as such, the frame averaging limit may provide a good model. 
    Meanwhile, many authoritative works on soft polymer networks (e.g., elastomers) make reference to the ``locally fluid-like freedom'' for molecular motion within the polymer network (e.g.,~\cite{warner2007liquid,treloar1975physics}).
    For soft polymer networks consisting of freely-jointed chains, there is experimental evidence that suggests the free rotation limit is a sound approximation~\cite{arruda1993threee}.
    An added benefit of focusing on the limits, beyond computational tractability, is that it also removes the need to fit an additional model parameter, $\kNetTors$.
\end{remark}

\begin{remark}
\emph{Relationship to existing polymer network models.}
The frame averaging limit shares structural similarities with the full network model~\cite{wu1992improved,wu1993improved}: both integrate over a continuous uniform distribution (chains over $\unitSphere$ vs. cross-link orientations over $\SOThree$) and can employ similar quadrature techniques.
However, the frame averaging limit represents a more physical realization -- rather than imposing infinite chain connectivity at a single junction, it averages over all orientations of finite cross-link structures where loads are shared among a discrete number of chains.

From a homogenization perspective, the junction relaxation in our framework is analogous to nonlinear RVE methods where boundaries are deformed while interiors equilibrate~\cite{milton2022theory,caulfield2024twinning}, naturally satisfying force balance.
This contrasts with restrictive affine (Voigt-type) or equal force (Reuss-type) theories~\cite{von2002mesoscale,verron2017equal,li2020variational}.

Historically, junction relaxation appears in the classical $4$-chain~\cite{flory1943statistical,treloar1975physics} and $3$-chain models~\cite{adolf1987computer,elias2006non}.
While earlier discrete models oriented the RVE relative to the principal frame to satisfy frame indifference, Arruda-Boyce~\cite{arruda1993threee} was the first to use the orientation that optimally distributes elastic energy across chains, which was later generalized in~\cite{grasinger2023polymer}.
Polydispersity in fluctuating-junction models was explored by Kloczkowski et al.~\cite{kloczkowski2002effect} for bimodal chains.
Our framework unifies these ideas -- junction relaxation, rotational fluctuations, and polydispersity -- while also satisfying frame indifference.

For detailed discussion of historical developments and comparisons with specific models, see \Fref{app:past-work}.
\end{remark}

\subsection{Junction position relaxation}
\paragraph*{Uniqueness.} For many networks of interest, the solution to the junction positional relaxation for general $\F$ and $\genRot$ is unique.
\begin{proposition} \label{prop:uniqueness}
    For RVEs consisting of chains with free energies ($\chainFreeEnergy_i,~i = 1, \dots, \numChains$) that are all convex, non-decreasing functions of stretch, with at least $1$ that is strictly convex, strictly increasing, the solution to
    \begin{equation}
        \inf_{\yc \in \rvedomain} \left\{ \Wchains\left(\F, \genRot, \yc\right) \right\}
    \end{equation}
    where
    \begin{equation}
        \Wchains\left(\F, \genRot, \yc\right) = \sum_{i=1}^{\numChains} \chainFreeEnergy_i\left(\left|\F \genRot \Xvec_i - \yc\right|\right)
    \end{equation}
    exists and is unique.
\end{proposition}
\begin{proof}
    By construction, $\RVEdomain$ is a compact, convex set.
    The deformed RVE, $\rvedomain = \F \genRot \RVEdomain$, is also compact, convex since these properties are conserved under linear transformation.
    The minimum exists by the extreme value theorem.

    Fix $\F$ and $\genRot$.
    The compositions $\chainFreeEnergy_i\left(\left|\F \genRot \Xvec_i - \yc\right|\right)$ are convex (strictly convex), non-decreasing (strictly increasing) functions composed with convex functions; thus, they are each convex (strictly convex) functions of $\yc$.
    As a sum of convex functions, with at least $1$ strictly convex, $\Wchains$ is a strictly convex function of $\yc$.
    A strictly convex function on a compact, convex set has a unique global minimum.
\end{proof}
This means that, for polydisperse RVEs consisting of Gaussian chains, Kuhn and Gr\"un chains, wormlike chains, and many others, the solution for the junction position is unique (for fixed $\F$ and $\genRot$).

\paragraph*{Exact solution for Gaussian RVEs.}
For RVEs consisting of Gaussian chains, the equilibrium equation for the junction position, $\partial \Wchains / \partial \yc = \nullvec$ is linear (see~\Fref{app:cross-link-chain-free-energy-derivatives-junction-position}), allowing the exact solution:
\begin{equation} \label{eq:gaussian-chain-unique-junction-position-relaxation}
    \ycStar = \F \genRot\frac{\left(\sum_{i=1}^{\numChains} \frac{\Xvec_i}{\n_i}\right)}{\left(\sum_{i=1}^{\numChains} \frac{1}{\n_i}\right)},
\end{equation}
where uniform $\mLen$ is here assumed.

\subsection{Network mechanics and frame indifference}
\label{sec:network-mechanics-frame-indifference}
As mentioned previously, we assume that -- although the exact structure of each and every individual cross-link is likely unknown -- a statistical distribution of cross-link structures may be known.
Recall that $\probDensInst\left(\inst\right)$ is the probability density of finding a cross-link with structure given by $\inst = \collection{\left(\n_i, \cLen_i, \Xvec_i, \chainFreeEnergy_i\right)}_{i=1}^{\numChains}$ in the network of interest.
Then the free energy density for the network can be modeled as
\begin{equation} \label{eq:network-free-energy-density}
    \allClinksFreeEnergyDensity\left(\F\right) = \frac{\crossDensity}{2} \int_{\inst \in \chainSpace} \df{\inst} \: \probDensInst \clinkFreeEnergy = -\frac{\crossDensity}{2} \kB \T \int_{\inst \in \chainSpace} \df{\inst} \: \probDensInst \ln \left(\partitionFunction\left(\F, \genRotRef\right)\right),
\end{equation}
where $\crossDensity$ is the (volumetric) cross-link number density such that the product $\crossDensity\probDens\left(\inst\right)$ equals the number density of cross-link structure $\inst$.
The factor of $1/2$ is here because each chain in the RVE is assumed to be connected to $2$ distinct cross-linkers in the network.

\paragraph*{Frame indifference.}
Frame indifference can be shown when the cross-link response utilizes the general partition function, \Eqref{eq:partition-function-1}, or an approximation thereof which is faithful to its symmetries, and when $\probDensInst$ satisfies Assumption \ref{ass:isotropy}; that is, where the orientation of cross-links within an elastic embedding is uniformly distributed.
\begin{proof}
Given \Eqref{eq:partition-function-1} and Assumption \ref{ass:isotropy}, it suffices to show that $\partitionFunction\left(\F, \genRotRef\right) = \partitionFunction\left(\strDiag, \genRotRef'\right)$ where, recall, $\strDiag = \diag\left(\pStretch{1}, \pStretch{2}, \pStretch{3}\right) = \pFrame \strTens \pFrame^T$ is a diagonalization of the right stretch tensor, $\strTens$, and we make the definition $\genRotRef' = \pFrame \genRotRef$. 
Again, $\F = \polRot \strTens$.
Then
\begin{equation}
   \partitionFunction\left(\F, \genRotRef\right) = \int_{\yc \in \rvedomain} \df{\yc} \int_{\genRot \in \SOThree} \df{\genRot} \: \exp\left(-\frac{1}{\kB \T}\sum_{i=1}^\numChains \chainFreeEnergy_i\left(\right|\polRot \pFrame^T \strDiag \pFrame \genRot \Xvec_i - \yc\left|\right) - \frac{\PEQ\left(\genRot\right)}{\kB \T}\right).
\end{equation}
Let $\yc' = \left(\polRot \pFrame^T\right)^{-1} \yc = \pFrame \polRot^T \yc$ and $\genRot' = \pFrame \genRot$.
The argument to the chain free energy takes the form
\begin{equation}
    \left|\polRot \pFrame^T \strDiag \pFrame \genRot \Xvec_i - \yc\right| = \left|\polRot \pFrame^T \left(\strDiag \pFrame \genRot \Xvec_i - \yc'\right)\right| = \left|\strDiag \genRot' \Xvec_i - \yc'\right|
\end{equation}
and, consequently,
\begin{equation}
\begin{split}
   \partitionFunction\left(\F, \genRotRef\right) &= \int_{\yc \in \rvedomain} \df{\yc} \int_{\genRot \in \SOThree} \df{\genRot} \: \exp\left(-\frac{1}{\kB \T}\sum_{i=1}^\numChains \chainFreeEnergy_i\left(\left|\strDiag \genRot' \Xvec_i - \yc'\right|\right) - \frac{\PEQ\left(\genRot\right)}{\kB \T}\right) \\
   &= \int_{\yc \in \rvedomain} \df{\yc} \int_{\genRot \in \SOThree} \df{\genRot} \: \exp\left(-\frac{1}{\kB \T}\sum_{i=1}^\numChains \chainFreeEnergy_i\left(\left|\strDiag \genRot' \Xvec_i - \yc'\right|\right) - \frac{\kNetTors}{2\kB \T}\lVert \pFrame^T \left(\genRot' - \genRotRef'\right) \rVert^2\right).
\end{split}
\end{equation}
The tensor norm, $\lVert . \rVert$, has not yet been specified.
For norm choices that appropriately model torsional elasticity, we want $\lVert \pFrame^T \left(\genRot' - \genRotRef'\right) \rVert = \lVert \genRot' - \genRotRef'\rVert$ to hold.
The Frobenius norm is one such example and is also the most natural choice for modeling elasticity.
In all such cases,
\begin{equation}
   \partitionFunction\left(\F, \genRotRef\right) = \int_{\yc' \in \rvedomain} \df{\yc'} \int_{\genRot' \in \SOThree} \df{\genRot'} \: \exp\left(-\frac{1}{\kB \T}\sum_{i=1}^\numChains \chainFreeEnergy_i\left(\left|\strDiag \genRot' \Xvec_i - \yc'\right|\right) - \frac{\kNetTors}{2\kB \T}\lVert \left(\genRot' - \genRotRef'\right) \rVert^2\right),
\end{equation}
and, as a result, $\partitionFunction\left(\F, \genRotRef\right) = \partitionFunction\left(\strDiag, \genRotRef'\right)$.
Thus, the resulting constitutive model is isotropic and frame indifferent.
\end{proof}

\subsection{Cross-link structures} \label{sec:cross-link-structures}

Having addressed the response of a cross-link to external loads, we now turn our attention to instantiation of cross-link structures. 
Per Assumptions \ref{ass:chain-ends}-\ref{ass:isotropy}, chains have their expected lengths, $\rmsChainLength$, upon cross-linking, and the orientations of chains are maximally spaced out.
To construct $\inst$, let the junction position at the moment of cross-linking be at the origin.
For each $\Xvec_i \in \RVEdomain$, define the unit vector $\dir{\Xvec}_i=\Xvec_i/\left|\Xvec_i\right|$. 
We place these unit vectors according to the Thomson problem: the equilibrium configuration of electrostatically repulsive particles on the unit sphere~\cite{thomson1904xxiv}.
Ideally, we also place these unit vectors to maximize reflectional symmetry about the origin (especially if the Thomson problem is not able to be satisfied).
These two considerations maximize angular separation between chains and produces an isotropic distribution of $\collection{\dir{\Xvec}_i}_{i=1}^\numChains$, consistent with excluded volume and entropic considerations prior to cross-linking.
The classical $4$-chain~\cite{flory1943statistical, treloar1943elasticity} and $6$-chain~\cite{grasinger2023polymer} models correspond to Thomson solutions for $\numChains = 4$ and $6$, respectively.
The classical $8$-chain~\cite{arruda1993threee} model does not correspond to the Thomson solution for $\numChains = 8$, but it does maximize the reflectional symmetry of $\collection{\dir{\Xvec}_i}_{i=1}^8$ about the origin.

In this work, we construct the polydisperse $4$-chain model as $\inst = \collection{\left(\n_i, \cLen_i, \Xvec_i, \chainFreeEnergy_i\right)}_{i=1}^4$ where
\begin{equation} \label{eq:four-chain-xs}
    \begin{split}
    \Xvec_1 &= \left(\rmsChainLength\right)_1 \left(0, 0, 1\right), \quad \quad
    \Xvec_2 = \left(\rmsChainLength\right)_2 \left(0, \frac{2\sqrt{2}}{3}, -\frac{1}{3}\right), \\
    \Xvec_3 &= \left(\rmsChainLength\right)_3 \left(\sqrt{\frac{2}{3}}, -\frac{\sqrt{2}}{3}, -\frac{1}{3}\right), \quad \quad
    \Xvec_4 = \left(\rmsChainLength\right)_4 \left(-\sqrt{\frac{2}{3}}, -\frac{\sqrt{2}}{3}, -\frac{1}{3}\right).
    \end{split}
\end{equation}
This will be the primary RVE of interest throughout.
For reference, \Fref{app:polydisperse-cross-link-structures} provides polydisperse cross-link RVEs for $\numChains \in [3, 8]$.
As in \Fref{sec:prelude}, $\rmsChainLength$ is calculated as a function of $\n$, $\cLen$, and $\chainFreeEnergy$, so $\inst$ implicitly determines all $\rmsChainLength$.

\subsection{Cross-link statistics} \label{sec:cross-link-statistics}

\begin{hlbreakable}We next address parameterization of the probability density functions $\probDensInst$ and $\probDens\left(\inst\right)$.\end{hlbreakable}
\begin{hlbreakable}Recall the three sources of network irregularity we are accounting for: scattered chain orientations, polydispersity in chain monomer number, and variation in the number of chains per cross-link junction.\end{hlbreakable}
Chain orientation scatter has been addressed in \Fref{sec:cross-link-mechanics} and \Fref{sec:network-mechanics-frame-indifference} (Assumption \ref{ass:isotropy}).
\begin{hlbreakable}To capture polydispersity and variation in cross-link coordination,\end{hlbreakable} we define probability density functions $\probDensChain = \probDensChain\left(\n\right)$ and $\probDensClinker = \probDensClinker\left(\numChains\right)$ for chain monomer number and cross-linker functionality, respectively.\footnote{Although $\n$ and $\numChains$ are discrete random variables, we use a continuous formulation since they can be represented as weighted sum of Dirac deltas, e.g., $\probDensChain\left(\n\right) = \sum_{i=1}^\infty p_{\n, i} \: \dirac{\n}{i}$ and $\probDensClinker\left(\numChains\right) = \sum_{j=1}^\infty p_{\numChains, j} \: \dirac{\numChains}{j}$ where $\sum_i p_{\n, i} = \sum_j p_{\numChains, j} = 1$. Approximating $\n$ and $\numChains$ as continuous (e.g., Gaussian) is sometimes convenient.}
We restrict attention to networks where all monomers have the same length, $\mLen$, so $\n_i$ uniquely determines $\cLen_i$ and $\left|\Xvec_i\right|$.\footnote{Although chain ends are placed at $\left|\Xvec_i\right| = \left(\rmsChainLength\right)_i$ from the origin to establish $\RVEdomain$, the initial chain length $\rmag_i = \left|\Xvec_i - \yc\right|$ may differ from $\left(\rmsChainLength\right)_i$ due to initial relaxation of the cross-link position $\yc \in \RVEdomain$.}\footnote{Dispersity in initial chain conformation could be considered within the modeling framework by allowing $\n_i$ and $\left|\Xvec_i\right|$ (or $\Xvec_i$, more generally) to vary independently through a joint probability distribution.}

To construct $\probDens(\inst)$, we invoke Flory's ideal network assumptions \cite{flory1953principles}: during cross-linking,
\begin{inparaenum}[(1)] \item all functional groups of the same type are equally reactive,
\item all groups react independently, and
\item no intramolecular reactions occur in finite species. \end{inparaenum}
Consequently, \emph{the number of monomers in each chain are independently and identically distributed} (i.i.d.).
\begin{hlbreakable}$\probDens(\inst)$ is thus parameterized as
\begin{equation} \label{eq:frame-invariant-chainspace-integration}
    \probDens(\inst) = \probDens(\collection{(\n_i, \cLen_i, \Xvec_i, \chainFreeEnergy_i)}_{i=1}^{\numChains}) = \probDensClinker(\numChains) \prod_{i=1}^{\numChains} \probDensChain(\n_i).
\end{equation}
Given \Eqref{eq:crosslink-isotropy}, integrating $\probDensInst$ over $\chainSpace$ follows as
\begin{align} \label{eq:chainspace-integration}
    \int_{\inst \in \chainSpace} \df{\inst} \: \probDensInst & = \int_{\genRotRef \in \SOThree} \df{\genRotRef} \: \probDensF{\genRotRef}\left(\genRotRef\right) \int_{0}^\infty \df{\numChains} \int_{0}^\infty ... \int_{0}^\infty \prod_{i=1}^{\numChains} \df{\n_i} \: \probDens\left(\collection{\left(\n_i, \cLen_i, \Xvec_i, \chainFreeEnergy_i\right)}_{i=1}^{\numChains}\right) \nonumber \\
    & = \frac{1}{\volSOThree} \int_{\genRotRef \in \SOThree} \df{\genRotRef} \int_{0}^\infty \df{\numChains} \: \probDensClinker(\numChains) \int_{0}^\infty ... \int_{0}^\infty \prod_{i=1}^{\numChains} \df{\n_i} \: \probDensChain(\n_i) = 1.
\end{align}\end{hlbreakable} 

Characterizing $\probDensChain$ experimentally prior to cross-linking is feasible, but whether monomer numbers remain i.i.d. after cross-linking is an open question.
Molecular dynamics simulations~\cite{ye2020molecular,barney2022fracture,zhang2024predicting,jang2015comparison} or \emph{in silico} network synthesis (e.g., Monte Carlo methods~\cite{bernhard2025phantom,bernhard2025pylimer}) could provide insight into realistic cross-link distributions.
Alternatively, $\probDensClinker$, $\probDensChain$, and $\probDens$ might be constructed from polymerization statistics — Miller-Macosko theory~\cite{miller1976new} or Pearson-Graessely theory~\cite{pearson1978structure} for $\probDensClinker$, and Flory-Stockmayer theory \cite{flory1953principles} for $\probDensChain$.
For simplicity, we retain the i.i.d. assumption on monomer numbers and assume simple forms for $\probDensClinker$ and $\probDensChain$.
Connecting this framework to experimental or \emph{in silico} network synthesis statistics is left for future work.

\section{Free rotation methods} \label{sec:FR-methods}

\subsection{$\SOThree$ representations and numerical optimization} \label{sec:FR-optimization}

There are numerical pathologies associated with the Euler angle representation when optimizing over rotations~\cite{kuehnel2003minimization}.
We instead use the exponential, or ``axis-angle'', representation.
The Rodrigues vector, $\rodVec \in \Reals^3$, describes the angle of rotation, $\varphi = \left|\rodVec\right|$, and an axis about which to rotate, $\dir{\vvec{u}} = \rodVec / \varphi$.
The rotation, $\genRot$, is obtained by taking the exponential of the generating skew-symmetric tensor
\begin{equation} \label{eq:rod}
	\tens{A} = \begin{pmatrix}
			0 & -\dir{u}_3 & \dir{u}_2 \\
			\dir{u}_3 & 0 & -\dir{u}_1 \\
			-\dir{u}_2 & \dir{u}_1 & 0
		\end{pmatrix}; \quad \genRot\left(\rodVec\right) = \exp\left(\varphi \tens{A}\right) = \identity + \sin \varphi \tens{A} + \left(1 - \cos \varphi\right) \tens{A}^2.
\end{equation}

While the form of $\clinkFreeRotationFreeEnergy$ given by \Eqref{eq:crosslink-free-energy-free-rotation} is convenient to work with analytically, we have found that the following equivalent form has better convergence numerically:
\begin{equation} \label{eq:num-free-energy-free-rotation}
    \clinkFreeRotationFreeEnergy\left(\F\right) = \inf_{\rodVec \in \Ball{2\pi}{\nullvec}, \yc \in \rvedomain} \sum_{i=1}^{\numChains} \chainFreeEnergy_i\left(\left|\F \genRot \Xvec_i - \genRot \yc\right|\right), 
\end{equation}
where $\Ball{2\pi}{\nullvec}$ is the ball of radius $2\pi$ centered at the origin, $\nullvec = \left(0, 0, 0\right)$.
The key differences in the formulation used for numerics are: \begin{inparaenum}[(1)] \item the explicit use of the Rodrigues representation and \item rotating the cross-link along with the chain ends as opposed to only rotating the chain ends. \end{inparaenum}
\Eqref{eq:num-free-energy-free-rotation} was approximated in Mathematica using \texttt{FindMinimum} for local optimization and \texttt{NMinimize} for global optimization.
\texttt{FindMinimum} uses a series of interior point methods for constrained (local) optimization; and \texttt{NMinimize} uses Nelder-Mead methods, supplemented by differential evolution.
The local optimization toolset was used for all of the $4$-chain numerical results because, as will be elaborated on shortly, the exact solution is known for the monodisperse case, which provides a good initial guess for all of the $4$-chain numerical examples presented herein.
Global optimization methods were used for the $6$-chain models because, in the $6$-chain case, the monodisperse solution is not explicitly known.
The choices for optimization methods were made for simplicity of implementation; however, we remark that the exponential representation for rotations is also amenable to gradient-based methods~\cite{kuehnel2003minimization}.

\subsection{Special case of monodispersity} \label{sec:FR-monodispersity}

\begin{hlbreakable}We now investigate RVEs representative of monodisperse polymer networks.
We consider a polymer network to be monodisperse if all its chains are composed of the same number of monomers $\n$ with the same monomer length $\mLen$, are described by the same free energy function $\chainFreeEnergy$, and thus have the same initial average root-mean-square chain length $\rmsChainLength$.\end{hlbreakable}
For certain monodisperse networks where the RVEs are given by the classical $4$-, $6$-, and $8$-chain models, the solution to the inner minimization of $\clinkFreeRotationFreeEnergy$ is known in closed form.
This solution is useful as a starting point for both analytical approximation of, and as an initial guess for numerical approximation of, the free rotation limit for polydisperse networks.
\begin{proposition}[Known solution for inner minimization and unification of discrete, monodisperse polymer network models] \label{prop:unification-FR}
\begin{hlbreakable}Consider a monodisperse RVE such that $\n_1 = \dots = \n_{\numChains} = \n$, $\cLen_1 = \dots = \cLen_{\numChains} = \cLen$ (implying $\mLen_1 = \dots = \mLen_{\numChains} = \mLen$), $\chainFreeEnergy_1 = \dots = \chainFreeEnergy_{\numChains} = \chainFreeEnergy$, and $\left|\Xvec_1\right| = \dots = \left|\Xvec_{\numChains}\right| = \Rmag$ where, without loss of generality, the origin is assumed to coincide with the center of mass (i.e., $\sum_{i=1}^{\numChains} \Xvec_i = \nullvec$).\end{hlbreakable}
Let \begin{equation} \label{eq:partition-stretch}
    \chainStretchStar = \frac{\Rmag}{\cLen} \sqrt{\frac{\pStretch{1}^2+\pStretch{2}^2+\pStretch{3}^2}{3}},
\end{equation}
and
\begin{equation}
    \left\{\genRotStar, \ycStar\right\} = \arg \inf_{\genRot \in \SOThree, \yc \in \rvedomain} \sum_{i=1}^{\numChains} \chainFreeEnergy\left(\left|\F \genRot \Xvec_i - \yc\right|\right)
\end{equation}
be the minimizing rotation and cross-link position, respectively, for a given RVE, $\collection{\left(\n, \cLen, \Xvec_i, \chainFreeEnergy\right)}_{i=1}^{\numChains}$, and deformation pair.
Provided the additional following conditions are satisfied:
\begin{enumerate}[(1)]
    \item the chain free energy is the same convex, non-decreasing function of stretch squared, $\chainStretch^2$, for each chain, and the chain free energy does not depend on the direction of stretch,
    \item the chain directions in the RVE, $\Xvec_i / \Rmag$, are consistent with one of the classical $4$-chain, $6$-chain, or $8$-chain RVEs (e.g., \Eqref{eq:four-chain-xs}, \Eqref{eq:six-chain-xs}, \Eqref{eq:eight-chain-xs}),\footnote{The classical $3$-chain RVE is not included here because it does not satisfy $\sum_{i=1}^{\numChains} \Xvec_i = \nullvec$.} 
\end{enumerate} 
we have the following results:
\begin{enumerate}[i)]
    \item a solution to the inner optimization problem for $\clinkFreeRotationFreeEnergy$ (i.e., $\left\{\genRotStar, \ycStar\right\}$) is known, and is such that $\chainStretch_i = \left|\F \genRotStar\Xvec_i - \ycStar\right| / \cLen = \chainStretchStar$ for all $i = 1, ..., \numChains$.
    It can be formulated explicitly as $\ycStar = \nullvec$ for all cases, and \begin{subequations} \label{eq:chain-Q}
    \begin{equation} \label{eq:$4$-chain-Q}
        \genRotStar = \genRot\left(\frac{\pi}{4} \: \euclid{1}\right)\genRot\left(\arccos\sqrt{\frac{2}{3}} \: \euclid{2}\right)\genRot\left(-\frac{\pi}{2} \: \euclid{3}\right)\pFrame \text{ for the $4$-chain RVE},
    \end{equation}
    and
    \begin{equation} \label{eq:8-chain-Q}
        \genRotStar = \pFrame \text{ for the $8$-chain RVE}.
    \end{equation}
    \end{subequations}
    \item all aforementioned RVEs produce the same constitutive model when the network has a consistent density of chains per unit volume (i.e., after an appropriate renormalization of $\crossDensity \rightarrow \crossDensity / \numChains$), which is equivalent to the classical $8$-chain model (e.g.,~\cite{arruda1993threee}).
\end{enumerate}

\end{proposition}

The principal argument is divided into the following parts: \begin{inparaenum}[(1)] \item the equipartition property of convex functions, \item there is a lower bound on the sum of square chain stretches, $\sum_i \chainStretch_i^2$, with respect to both rotations of the RVE and translations of the cross-link junction, and \item a rotation and cross-link position exist that satisfy an equipartition of stretch. \end{inparaenum}
\begin{lemma}[Equipartition property] \label{lem:equipartition}
Given a convex, non-decreasing function $f$ and $n$ variables, $x_1, \dots, x_n$, whose sum is bounded from below by $y$, a solution to \begin{equation*}
    \min_{x_1, \dots, x_n} \sum_{i=1}^n f\left(x_i\right) \; 
    \text{ subject to } \: \sum_{i=1}^n x_i \geq y
\end{equation*}
is $x_1 = \dots = x_n = \frac{y}{n}$.
\end{lemma}
\begin{proof}
This follows from the definition of convexity,
\begin{equation*}
    \sum_{i=1}^n f\left(x_i\right) = n\left(\frac{1}{n} \sum_{i=1}^n f\left(x_i\right)\right) \geq n f\left(\frac{1}{n} \sum_{i=1}^n x_i\right) \geq n f\left(\frac{y}{n}\right).
\end{equation*}
and the last step is due to the non-decreasing property of $f$. 
\end{proof}
The significance of equipartition is that it establishes the optimality of the $\genRotStar$ and $\ycStar$ where the sum of the chain stretches squared is both minimal and equally partitioned between each of the $\numChains$ chains.
We now proceed with the remainder of the proof of \Fref{prop:unification-FR}.
\begin{proof}
    The first step is to consider the lower bound on the sum of chain stretches squared.
    One can show that this quantity is invariant with respect to rotations.
    Importantly, it establishes that the lower bound is achieved if and only if $\ycStar = \nullvec$.
    \paragraph*{Lower bound on $\sum_i \chainStretch_i^2$ and a solution for $\ycStar$.}
    Consider $\sum_i \chainStretch_i^2$ for general $\genRot$ and $\yc$:
    \begin{equation}
        \sum_{i=1}^{\numChains} \chainStretch_i^2 = \frac{1}{\cLen^2} \sum_{i=1}^{\numChains} \left|\F \genRot \Xvec_i - \yc\right|^2 = \frac{1}{\cLen^2} \left(\underbrace{\sum_{i=1}^{\numChains} \left(\F \genRot \Xvec_i\right) \cdot \left(\F \genRot \Xvec_i\right)}_{\text{conserved}} - \underbrace{2\sum_{i=1}^{\numChains} \left(\F \genRot \Xvec_i\right) \cdot \yc}_{=0} + \underbrace{\sum_{i=1}^{\numChains} \yc \cdot \yc}_{\geq 0}\right).
    \end{equation}
    The first term in the parentheses is conserved with respect to $\genRot \in \SOThree$~\cite{grasinger2023polymer}.
    (For completeness, this result is reproduced in \Fref{app:unification}.)
    The second term vanishes.
    Indeed,
    \begin{equation}
        2\sum_{i=1}^{\numChains} \left(\F \genRot \Xvec_i\right) \cdot \yc = 2\sum_{i=1}^{\numChains} \Xvec_i \cdot \left(\trans{\genRot} \trans{\F} \yc\right) = 2 \left(\trans{\genRot} \trans{\F} \yc\right) \cdot \underbrace{\left(\sum_{i=1}^{\numChains} \Xvec_i\right)}_{=\nullvec}.
    \end{equation}
    Finally, the last term is clearly nonnegative.
    What is instructive about this analysis is that for monodisperse RVEs that satisfy the conditions of the proposition, $\ycStar = \nullvec$ is a part of a solution for all possible $\F$.

    \paragraph*{$\genRotStar$ that satisfies equipartition of stretch.}
    The final step is to establish that, provided $\ycStar = \nullvec$, a $\genRotStar$ exists such that equipartition of stretch is satisfied.
    This is indeed the case and, in fact, an explicit solution can be given for the $4$- and $8$-chain cases~\cite{grasinger2023polymer}.
    
    \item ($4$-chain case)
    To show the $4$-chain case, we decompose the rotation of interest $\genRotStar$ as $\genRotStar = \genRot' \pFrame$ where, recall, $\pFrame$ rotates the Euclidean basis to align with the principal frame.
    Then it suffices to show that $\genRot'$ exists such that the chain stretches are all equal.
    Let $\genRot' = \genRot\left(\eulerx \euclid{1}\right) \genRot\left(\eulery \euclid{2}\right) \genRot\left(\eulerz \euclid{3}\right)$. 
    Then the chain with end-to-end vector $\Xvec_1 \left(= \Rmag \left(0, 0, 1\right)\right)$ (see \Eqref{eq:four-chain-xs}) is deformed such that $\Xvec_1 \rightarrow \F \genRot' \pFrame \Xvec_1$ and
    \begin{equation} \label{eq:stretch-41}
        \chainStretch_1 = \frac{\Rmag}{\cLen} \sqrt{\pStretch{1}^2 \sin^2 \eulery + \pStretch{2}^2 \sin^2 \eulerx \cos^2 \eulery + \pStretch{3}^2 \cos^2 \eulerx \cos^2 \eulery}.
    \end{equation}
    Here $\eulerx$ and $\eulery$ can be chosen such that $\chainStretch_1 = \chainStretchStar$.
    One such solution is $\eulerx = \pi / 4$ and $\eulery = \arccos \sqrt{2 / 3}$.
    Substituting this solution for $\eulerx$ and $\eulery$ into $\genRot'$, we see that $\Xvec_2 \left(= \Rmag \left(0, 2\sqrt{2}/3, -1/3\right)\right)$ is deformed such that
    \begin{equation} \label{eq:stretch-42}
        \chainStretch_2 = \frac{\Rmag}{9 \cLen} \sqrt{
            3\left(\pStretch{1} + 4 \pStretch{1} \sin \eulerz\right)^2 + 
                \pStretch{2}^2\left(6 \cos \eulerz + \sqrt{3}\left(1 - 2\sin \eulerz\right)\right)^2 +
                \pStretch{3}^2\left(6 \cos \eulerz - \sqrt{3}\left(1 - 2\sin \eulerz\right)\right)^2
            }.
    \end{equation}
    Now $\eulerz$ can be chosen such that $\chainStretch_2 = \chainStretchStar$.
    One such solution is $\eulerz = -\pi / 2$.
    Substituting into $\genRot'$ we see that $\chainStretch_1 = \dots = \chainStretch_4 = \chainStretchStar$.
    
    \item ($6$-chain case)
    The stretch for each chain in the RVE can be formulated as
    \begin{equation}
        \chainStretch_i = \frac{1}{\cLen} \sqrt{\Xvec_i \cdot \changeCoord{\cGreen} \Xvec_i}, \quad \text{ where } \changeCoord{\cGreen} = \trans{\genRot} \cGreen \genRot
    \end{equation}
    is a real proper orthogonal similarity transformation of $\cGreen$.
    There exists a real proper orthogonal similarity transformation where all of the elements on the diagonal are equal~\cite{horn1985matrix}. 
    The existence of this similarity transformation can be understood as follows.
    Consider rotating the coordinate system about $\euclid{3}$ by $\varphi$.
    One can permute $\changeCoord{\cGreenSym}_{11}$ and $\changeCoord{\cGreenSym}_{22}$ by taking $\varphi = \pi / 2$.
    Further, this transformation of $\changeCoord{\cGreenSym}_{11}$ and $\changeCoord{\cGreenSym}_{22}$ is continuous with respect to $\varphi$ so that there exists a $\varphi$ such that $\changeCoord{\cGreenSym}_{11} = \changeCoord{\cGreenSym}_{22}$.
    Similar arguments can be made about $\changeCoord{\cGreenSym}_{11}$ and $\changeCoord{\cGreenSym}_{33}$ (by rotating about $\euclid{2}$), and about $\changeCoord{\cGreenSym}_{22}$ and $\changeCoord{\cGreenSym}_{33}$ (by rotating about $\euclid{1}$).
    Thus, there exists a $\genRotStar$ such that $\changeCoord{\cGreenSym}_{11} = \changeCoord{\cGreenSym}_{22} = \changeCoord{\cGreenSym}_{33} = 1 / 3 \Tr \cGreen$.
    Clearly, in this case, the chain stretches satisfy \Eqref{eq:partition-stretch}.
    Unfortunately, the exact form of $\genRotStar$ as a function of $\F$ is not known (at least not to the authors).
    
    \item ($8$-chain case)
    The $\genRotStar$ which satisfies \Eqref{eq:partition-stretch} for the $8$-chain RVE is well-known and is given by $\genRotStar = \pFrame$~\cite{arruda1993threee}.
\end{proof}

\begin{remark} \label{rem:mono-topology}
\emph{Implications for some isotropic, monodisperse polymer networks.}
    Common chain free energies are convex, non-decreasing functions of $\chainStretch^2$~\cite{grasinger2023polymer} (e.g., the Gaussian chain, the Kuhn and Gr\"{u}n approximation to the freely-jointed chain, and the worm-like chain).
    This, along with \Fref{prop:unification-FR}, suggests that isotropic, monodisperse polymer networks which \begin{inparaenum}[(1)] \item are free to locally rotate and \item consist of chains with convex free energies \end{inparaenum} have an elastic response which is insensitive to -- and potentially invariant with respect to -- network topology.
\end{remark}

\begin{remark}
\emph{Caveats for polydisperse polymer networks.}
    There are two important differences for polydisperse polymer networks: \begin{inparaenum}[(1)]
       \item it is unlikely that $\left|\Xvec_1\right| = ... = \left|\Xvec_{\numChains}\right|$ holds; as a result, $\sum_i \chainStretch_i^2$ is no longer conserved with respect to rotations and $\ycStar = \nullvec$ may no longer be an optimal cross-link position, and
       \item the chain free energies differ by chain, so equipartition of stretch is no longer an optimal solution. 
    \end{inparaenum}
    For instance, if the network consists of Gaussian chains with differing numbers of monomers, then each chain has a different chain free energy, $\chainFreeEnergy_i\left(\chainStretch\right) = \left(3 / 2\right) \n_i \kB \T \chainStretch^2$ (i.e., the entropic springs have different stiffnesses).
    However, as mentioned, the known solution for monodisperse networks may serve as inspiration towards initial guesses for numerical approximations of $\clinkFreeRotationFreeEnergy$, and closed-form, analytical approximations of $\clinkFreeRotationFreeEnergy$.
\end{remark}

\subsection{Closed-form approximation}\label{sec:FR-closed-form-approx}

Computation of the free rotation limit may pose some difficulties as it consists of averaging over the space of all cross-link RVEs, and inside the averaging operation is an optimization problem.
Further, the representation for RVEs may be high-dimensional.
Numerical integration can become prohibitively expensive in such a high-dimensional space, especially when the integrand is expensive to compute.
To address these difficulties, we seek to derive analytical, closed-form approximations to the inner optimization problem for $\clinkFreeRotationFreeEnergy$.

Consider a polydisperse RVE whose monodisperse counterpart has an exact solution to the rotational and cross-link positional relaxation problem ($4$-chain and $8$-chain).
We reformulate the inner optimization problem for $\clinkFreeRotationFreeEnergy$ (towards approximation in the limit of a ``small'' amount of polydispersity) by considering perturbations about the known solution for the monodisperse case.
To begin, we assume incompressibility, which implies that $\F = \diag\left(\pStretch{1}, \pStretch{2}, 1 / \pStretch{1} \pStretch{2} \right)$ without loss of generality.
This takes the form
\begin{equation} \label{eq:FR-FED-perturb-approx}
    \clinkFreeRotationFreeEnergyApprox{\nPolydisperseSet} = \inf_{\genRot \in \SOThree, \yc \in \rvedomain} \sum_{i=1}^{\numChains} \chainFreeEnergy_i\left(\left|\F \genRot \Xvec_i - \yc\right|\right) = \inf_{\delRot, \delYc} \sum_{i=1}^{\numChains} \chainFreeEnergy_i\left(\left|\F \genRot\left(\delRot\right) \genRotMono \Xvec_i - \left(\ycMono + \delYc\right)\right|\right)
\end{equation}
where we introduce $\clinkFreeRotationFreeEnergyApprox{\nPolydisperseSet}$ to denote the free energy of an RVE in the free rotation limit with \begin{hlbreakable}$\ncollection{\nPolydisperseSet}$\end{hlbreakable} monomers in its chains, where $\genRotMono$ and $\ycMono$ are the monodisperse solutions to the rotational and cross-link positional relaxation problem, respectively, and where $\genRot\left(\delRot\right)$ is a small additional rotation and $\delYc$ is a perturbation of the cross-link position.
Both $\delRot$ and $\delYc$ are assumed small; that is, $\smallparam = \max\left(\left|\delRot\right|, \left|\delYc\right| / \left(\rmsChainLength\right)_{\max} \right) \ll 1$ where $\left(\rmsChainLength\right)_{\max} = \max\left(\left(\rmsChainLength\right)_1,\dots,\left(\rmsChainLength\right)_{\numChains}\right)$.
For brevity, let the inner free energy cost be denoted by
\begin{equation} \label{eq:inner-FED}
    \clinkInnerFreeRotationFreeEnergy\left(\delRot, \delYc\right) = \sum_{i=1}^{\numChains} \chainFreeEnergy_i\left(\left|\F \genRot\left(\delRot\right) \genRotMono \Xvec_i - \left(\ycMono + \delYc\right)\right|\right),
\end{equation}
and then we approximate derivatives with respect to $\delRot$ and $\delYc$ to linear order in $\delRot$ and $\delYc$ as\footnote{Note that $\delRot$ derivatives are with respect to components of the Rodrigues vector, not rotational derivatives (see~\cite{kuehnel2003minimization}, for example). The approximation made here can be seen as a leading order perturbation, in contrast to a full iteration of Newton's method, which can lead to complex and unstable expressions when the initial guess is not ``close enough''.}
\begin{align}
    \label{eq:perturb-deriv-1}
    \partialx{\clinkInnerFreeRotationFreeEnergy}{\left(\delRot\right)} &= \partialx{\clinkInnerFreeRotationFreeEnergy}{\left(\delRot\right)}\Bigg|_{\left(\nullvec, \nullvec\right)} + 
    \left(\frac{\partial^2 \clinkInnerFreeRotationFreeEnergy}{\partial \delYc \: \partial \delRot}\Bigg|_{\left(\nullvec, \nullvec\right)}\right) \delYc + \left(\frac{\partial^2 \clinkInnerFreeRotationFreeEnergy}{\partial \delRot \: \partial \delRot}\Bigg|_{\left(\nullvec, \nullvec\right)}\right) \delRot +  \orderOf{\smallparam^2}, \\
    \label{eq:perturb-deriv-2}
    \partialx{\clinkInnerFreeRotationFreeEnergy}{\left(\delYc\right)} &= \partialx{\clinkInnerFreeRotationFreeEnergy}{\left(\delYc\right)}\Bigg|_{\left(\nullvec, \nullvec\right)} + \left(\frac{\partial^2 \clinkInnerFreeRotationFreeEnergy}{\partial \delYc \: \partial \delYc}\Bigg|_{\left(\nullvec, \nullvec\right)}\right) \delYc +
    \left(\frac{\partial^2 \clinkInnerFreeRotationFreeEnergy}{\partial \delRot \: \partial \delYc}\Bigg|_{\left(\nullvec, \nullvec\right)}\right) \delRot + \orderOf{\smallparam^2}.
\end{align}
Dropping higher order terms (i.e., $\orderOf{\smallparam^2}$), setting \Eqref{eq:perturb-deriv-1} and \Eqref{eq:perturb-deriv-2} equal to $\nullvec$, and solving for $\delRot$ and $\delYc$, one can arrive at an approximation that minimizes $\clinkInnerFreeRotationFreeEnergy$. We denote the solution for $\delRot$ and $\delYc$ as $\widetilde{\delRot}$ and $\widetilde{\delYc}$, respectively, such that
\begin{subequations} \label{eq:general-RVE-analytical-model}
\begin{align}
    \widetilde{\delRot} &= \widetilde{\delRot}\left(\nPolydisperseSet\right), \\
    \widetilde{\delYc} &= \widetilde{\delYc}\left(\nPolydisperseSet\right).
\end{align}
As derived, the approximation above, by itself, does not satisfy material frame indifference.
The reason is as follows: the rotation which minimizes the free energy for the monodisperse case is not unique; in other words, there are other equally valid choices of $\genRotMono$ to perturb about via $\delRot$ ($\genRotStar$ as expressed in \Eqref{eq:chain-Q} is only one such choice).
Different choices of rotation can result in different approximations for the polydisperse case because the varying chain lengths break some symmetry.
For example, consider the rotation of a cube such that its edges are aligned with the principal directions, $\pDir_i$.
Naively, one could expand about all rotations that align the chains in the (undeformed) RVE with the diagonals of the cube, and derive an approximate solution for each.
Then, the best approximation can be chosen as the one with minimal free energy.
Luckily, further derivations are unnecessary.
Different choices of $\genRotMono$ about which to expand correspond with certain permutations of the monomer numbers, $\ncollection{\nPolydisperseSet}$.
Therefore, approximations to the rotational and cross-link position relaxations can be expressed as
\begin{align}
     \sigma^* &= \arg \min_{\sigma \in \genGroup} \clinkInnerFreeRotationFreeEnergy\left(\widetilde{\delRot}\left(\sigma \cdot \ncollection{\nPolydisperseSet}\right), \widetilde{\delYc}\left(\sigma \cdot \ncollection{\nPolydisperseSet}\right)\right) \\
     \collection{\delRot, \delYc} & = \collection{\widetilde{\delRot}\left(\sigma^* \cdot \ncollection{\nPolydisperseSet}\right), \widetilde{\delYc}\left(\sigma^* \cdot \ncollection{\nPolydisperseSet}\right)} 
\end{align}
\end{subequations}
where $\genGroup \leq \mathbb{S}_{\numChains}$ is the symmetry group of the RVE structure, which is a subgroup of $\mathbb{S}_{\numChains}$, the symmetric group of $\numChains$ elements, and $\sigma \cdot \generic$ denotes the group action of $\sigma$ on $\generic$. In other words, the approximate solution for the cross-link rotational and positional relaxation is taken to be the one for which the inner free energy is minimal, where the minimization is over all permutations of the monomers numbers that are consistent with the underlying symmetry of the RVE.
For the $4$-chain RVE, $\genGroup \cong \mathbb{S}_{4}$.

The free energy of the polydisperse RVE, $\clinkFreeRotationFreeEnergyApprox{\nPolydisperseSet}$, can then be approximated by substituting \Eqref{eq:general-RVE-analytical-model} into \Eqref{eq:inner-FED}. 

\begin{remark}
\emph{Distinction between linearization about the monodisperse RVE solution and linear elasticity.}
    The closed form approximation given by \Eqref{eq:general-RVE-analytical-model} was obtained by linearization of its governing equations; however, the approximation is distinct from \emph{linear elasticity}.
    The typical assumption underlying linear elasticity is that $\left\lVert\takeGrad{\mathbf{u}}\right\rVert$ is small such that $\orderOf{\left\lVert\takeGrad{\mathbf{u}}\right\rVert^2}$ terms can be neglected, where $\mathbf{u}$ is the deformation at a material point.
    Instead, here, \emph{the linearization is about the monodisperse RVE solution, which is known for finite deformations}.
    The approximation can be expected to be at its most accurate when the amount of polydispersity (i.e., variance in $\n_i$) within the RVE is small enough.
    We will show through example that this is indeed the case.
    The closed-form approximation agrees well with numerical solutions at small deformations in all cases. For nearly monodisperse networks, the error increases gradually with deformation. For highly polydisperse networks, the error grows more rapidly.
\end{remark}

\paragraph*{Gaussian networks with cross-link degree 4.}
\begin{hlbreakable}Consider the simplest cross-link RVE for which a closed-form approximation can be found: the polydisperse $4$-chain RVE (with $\Xvec_i$ given by \Eqref{eq:four-chain-xs}) consisting of chains with Gaussian free energy given by \Eqref{eq:gauss-free-energy}.\end{hlbreakable}
Here, the closed-form approximation for the rotational and cross-link positional relaxation problem is found to be
\begin{subequations} \label{eq:$4$-chain-analytical-model}
\begin{align}
   \widetilde{\delRotMag}_1\left(\n_1, \n_2, \n_3, \n_4\right) &= \delRotMag_1\left(\n_1, \n_2, \n_3, \n_4\right) = 0,\\
   \widetilde{\delRotMag}_2\left(\n_1, \n_2, \n_3, \n_4\right) &= \frac{\auxVar{1}\left(\auxVar{4}\pStretch{1}^4\pStretch{2}^4 - \auxVar{2}^2\pStretch{1}^2 - \auxVar{3}^2\pStretch{2}^2\right)}{\auxVar{3}\left(\auxVar{1}^2\pStretch{1}^4\pStretch{2}^4 + \auxVar{2}^2\pStretch{1}^2 + \auxVar{3}^2\pStretch{2}^2\right)}, \\
   \widetilde{\delRotMag}_3\left(\n_1, \n_2, \n_3, \n_4\right) &= \frac{\auxVar{2}\left(\auxVar{1}^2\pStretch{1}^4\pStretch{2}^4 + \auxVar{5}\pStretch{1}^2 + \auxVar{3}^2\pStretch{2}^2\right)}{\auxVar{3}\left(\auxVar{1}^2\pStretch{1}^4\pStretch{2}^4 + \auxVar{2}^2\pStretch{1}^2 + \auxVar{3}^2\pStretch{2}^2\right)}, \\
   \widetilde{\delYcMag}_1\left(\n_1, \n_2, \n_3, \n_4\right) &= -\frac{\mLen \auxVar{3} \auxVar{6} \ngeomean^2 \: \pStretch{1} \pStretch{2}^2}{\sqrt{3} \: \auxVar{7} \left(\auxVar{1}^2 \pStretch{1}^4 \pStretch{2}^4 + \auxVar{2}^2 \pStretch{1}^2 + \auxVar{3}^2 \pStretch{2}^2\right)}, \\
   \widetilde{\delYcMag}_2\left(\n_1, \n_2, \n_3, \n_4\right) &= \frac{\mLen \auxVar{2} \auxVar{6} \ngeomean^2 \: \pStretch{1}^2 \pStretch{2}}{\sqrt{3} \: \auxVar{7} \left(\auxVar{1}^2 \pStretch{1}^4 \pStretch{2}^4 + \auxVar{2}^2 \pStretch{1}^2 + \auxVar{3}^2 \pStretch{2}^2\right)}, \\
   \widetilde{\delYcMag}_3\left(\n_1, \n_2, \n_3, \n_4\right) &= \frac{\mLen \auxVar{1} \auxVar{6} \ngeomean^2 \: \pStretch{1}^3 \pStretch{2}^3}{\sqrt{3} \: \auxVar{7} \left(\auxVar{1}^2 \pStretch{1}^4 \pStretch{2}^4 + \auxVar{2}^2 \pStretch{1}^2 + \auxVar{3}^2 \pStretch{2}^2\right)},\\
   \sigma^* &= \arg \min_{\sigma \in \mathbb{S}_4} \clinkInnerFreeRotationFreeEnergy\left(\widetilde{\delRot}\left(\sigma \cdot \ncollection{\n_1, \n_2, \n_3, \n_4}\right), \widetilde{\delYc}\left(\sigma \cdot \ncollection{\n_1, \n_2, \n_3, \n_4}\right)\right) \\
   \collection{\delRot, \delYc} &= \collection{\widetilde{\delRot}\left(\sigma^* \cdot \ncollection{\n_1, \n_2, \n_3, \n_4}\right), \widetilde{\delYc}\left(\sigma^* \cdot \collection{\n_1, \n_2, \n_3, \n_4}\right)},
\end{align}
where 
\begin{align}
    \auxVar{1} =& \sqrt{\n_1 \n_2 \n_3} - \sqrt{\n_1 \n_2 \n_4} - \sqrt{\n_1 \n_3 \n_4} + \sqrt{\n_2 \n_3 \n_4}, \\
    \auxVar{2} =& \sqrt{\n_1 \n_2 \n_3} - \sqrt{\n_1 \n_2 \n_4} + \sqrt{\n_1 \n_3 \n_4} - \sqrt{\n_2 \n_3 \n_4}, \\
    \auxVar{3} =& \sqrt{\n_1 \n_2 \n_3} + \sqrt{\n_1 \n_2 \n_4} - \sqrt{\n_1 \n_3 \n_4} - \sqrt{\n_2 \n_3 \n_4}, \\
    \auxVar{4} =& 2 \left(-2 \n_2 \n_3 \sqrt{\n_1 \n_4} +\n_2 \n_3 \n_4+\n_1 \left(n_4 \left(\sqrt{\n_2}-\sqrt{\n_3}\right){}^2+\n_2 \n_3\right)\right), \\
    \auxVar{5} =& -2 \left(-2 \n_2 \n_4  \sqrt{\n_1 \n_3} +\n_2 \n_3 \n_4+\n_1 \left(-2 \n_3 \sqrt{\n_2 \n_4} +\n_3 \n_4+\n_2 \left(n_3+\n_4\right)\right)\right), \\
    \auxVar{6} =& 3 \n_2 \n_3 \n_4 - 2 \sqrt{\n_1} \left(\n_3 \n_4\sqrt{\n_2} +\n_2 \left(\n_3 \sqrt{\n_4}+  \n_4\sqrt{\n_3}\right)\right)\\
    & +\n_1 \left(3 \n_3 \n_4+\n_2 \left(-2 \sqrt{\n_3 \n_4}+3 \n_3+3 \n_4\right)-2 \sqrt{\n_2} \left(\n_3 \sqrt{\n_4} + \n_4 \sqrt{\n_3}\right)\right), \nonumber \\
    \auxVar{7} =& \n_1 \n_2 \n_3 + \n_1 \n_2 \n_4 + \n_1 \n_3 \n_4 + \n_2 \n_3 \n_4,
\end{align}
\end{subequations}
\begin{hlbreakable} and $\ngeomean = \left(\n_1 \n_2 \n_3 \n_4\right)^{1/4}$ is the geometric mean of monomer number.
This is distinct from the arithmetic mean of monomer number, $\nmean = \left(\n_1 + ... + \n_4\right) / 4$.
In the context of this work, an interesting measure of polydispersity is the ratio of the geometric mean to the arithmetic mean, $\eff = \ngeomean / \nmean$.
This is often referred to as the ``efficiency'' of a set of values.
The efficiency is a measure of the homogeneity of values in a given dataset; it is $1$ when all of the values are the same and approaches $0$ as the values become increasing disparate.
We consider the influence of $\eff$ on polydisperse cross-link mechanics in \Fref{sec:gaussian-FR}.\end{hlbreakable}

\section{Frame averaging methods} \label{sec:FA-methods}

\subsection{$\SOThree$ representations and numerical optimization} \label{sec:FA-optimization}

Unlike the free rotation limit, the frame averaging limit does not require optimization over rotations, but instead requires integration over $\SOThree$, as per \Eqref{eq:crosslink-free-energy-frame-average-integral-orientations}.
We approximate integration over $\SOThree$ via numerical quadrature using Euler angle parameterization
\begin{equation} \label{eq:so3-quadrature-general}
\frac{1}{\volSOThree} \int_{\genRotRef \in \SOThree} \df{\genRotRef} \genFunc\left(\genRotRef\right) \approxeq \sum_{i=1}^{\numSOThreeQuad} \weightFactorSOThreeQuad_i \genFunc\left(\left(\quadPointSOThreeQuad\right)_i\right),
\end{equation}
where $\genFunc$ is some general (scalar, vectorial, or tensorial) function of $\genRotRef$, $\collection{\left(\quadPointSOThreeQuad\right)_i}_{i=1}^{\numSOThreeQuad}$ is the set of $\SOThree$ quadrature points, $\collection{\weightFactorSOThreeQuad_i}_{i=1}^{\numSOThreeQuad}$ is the set of corresponding weight factors, and $\numSOThreeQuad$ is the number of quadrature points.
Additional details of $\SOThree$ quadrature are provided in \Fref{app:SO3-quadrature}.
The $\SOThree$ quadrature scheme can be constructed from any spherical quadrature rule. 
Since full network and microsphere models~\cite{treloar1979non,wu1992improved,wu1993improved,miehe2004micro} already employ spherical quadrature for integration over chain orientations, existing implementations require only minor modification to perform $\SOThree$ integration.

While the form of $\clinkFrameAveragingFreeEnergy$ given by \Eqref{eq:crosslink-free-energy-frame-average} and \Eqref{eq:crosslink-free-energy-frame-average-integral-orientations} are convenient to work with analytically, we have found that the following equivalent forms have better convergence numerically:
\begin{align} 
    \label{eq:num-crosslink-free-energy-frame-average} \clinkFrameAveragingFreeEnergy\left(\F, \genRotRef\right) & = \inf_{\yc \in \rvedomain} \sum_{i=1}^{\numChains} \chainFreeEnergy_i\left(\left|\F \genRotRef \Xvec_i - \genRotRef \yc\right|\right), \\
    \label{eq:num-crosslink-free-energy-frame-average-integral-orientations} \clinkFrameAveragingFreeEnergy\left(\F\right) & = \frac{1}{\volSOThree} \int_{\genRotRef \in \SOThree} \df{\genRotRef} \left(\inf_{\yc \in \rvedomain} \sum_{i=1}^{\numChains} \chainFreeEnergy_i\left(\left|\F \genRotRef \Xvec_i - \genRotRef \yc\right|\right)\right).
\end{align}

\Eqref{eq:num-crosslink-free-energy-frame-average} and \Eqref{eq:num-crosslink-free-energy-frame-average-integral-orientations} were approximated in Python via the \texttt{NumPy}~\cite{harris2020array} and \texttt{SciPy}~\cite{virtanen2020scipy} packages. 
For the $\SOThree$ quadrature involved in \Eqref{eq:num-crosslink-free-energy-frame-average-integral-orientations}, we need to define a foundational spherical quadrature scheme and the number of spin discretization points, $\numSpinQuad$ (see \Fref{app:SO3-quadrature}).
We use the 74-point, 13-degree Ba\v{z}ant and Oh~\cite{bavzant1986efficient} formula as the foundational spherical quadrature scheme, along with $\numSpinQuad=16$, for all frame averaging limit results in this work (due to symmetry considerations, this amounts to $\numSOThreeQuad=592$).
For our frame averaging limit results, we use the Constrained Optimization BY Quadratic Approximations (COBYQA) algorithm~\cite{rago_thesis,razh_cobyqa} offered by the \texttt{scipy.optimize.minimize} function for local constrained optimization (unless specified otherwise, e.g., \Figref{fig:FA-WGauss-fit}).
Even though our computational implementation is capable of utilizing other local constrained optimization methods\footnote{These include the Constrained Optimization BY Linear Approximation (COBYLA) algorithm~\cite{Zhang_2023} and the constrained trust region method~\cite{conn2000trust} offered by the \texttt{scipy.optimize.minimize} function.} and global constrained optimization methods,\footnote{These include the simplicial homology global optimization~\cite{endres2018simplicial} and differential evolution~\cite{storn1997differential} methods provided by the \texttt{scipy.optimize.shgo} and \texttt{scipy.optimize.differential\_evolution} functions, respectively.} we found the results obtained from the COBYQA method to be both efficient and accurate.
This computational implementation is provided in a freely available \texttt{polydisperse-polymer-networks} \texttt{GitHub} repository \citep{jasonmulderrig_polydisperse_polymer_networks_2025}. 

\subsection{Closed-form approximation}\label{sec:FA-closed-form-approx}

As with the free rotation assumption, we outline a procedure for deriving an analytical, closed-form approximation to the inner optimization problem for $\clinkFrameAveragingFreeEnergy$.
We propose $\ycMono=\nullvec$ as the initial guess for the junction position, motivated by the free rotation monodisperse case. 
For frame averaging monodisperse RVEs with Gaussian chains and $\sum_{i=1}^{\numChains} \Xvec_i = \nullvec$, $\ycStarQO = \nullvec$ is the exact equilibrium solution (as per \Eqref{eq:gaussian-chain-unique-junction-position-relaxation}). 
(This also holds exactly for monodisperse $6$-chain and $8$-chain RVEs with Kuhn-Gr\"un and wormlike chain models.\footnote{It can be seen that this satisfies force balance for the classical $6$ and $8$-chain RVEs because for every chain force $\rdir_i \chainFreeEnergy'_i\left(\left|\F \genRotRef \Xvec_i\right|\right)$ there is an equal and opposite force. By \Fref{prop:uniqueness}, this is the unique minimum of $\Wchains$.}) 
This suggests $\ycMono = \nullvec$ provides a reasonable guess about which to expand for approximating a solution to the polydisperse equilibrium equation.
This takes the form
\begin{align}
    \label{eq:FA-FED-perturb-approx} \clinkFrameAveragingFreeEnergyApprox{\nPolydisperseSet}\left(\F, \genRotRef\right) & = \inf_{\yc \in \rvedomain} \sum_{i=1}^{\numChains} \chainFreeEnergy_i\left(\left|\F \genRotRef \Xvec_i - \yc\right|\right) \approx \inf_{\delYcQO} \sum_{i=1}^{\numChains} \chainFreeEnergy_i\left(\left|\F \genRotRef \Xvec_i - \left(\ycMono + \delYcQO\right)\right|\right), \\
    \clinkFrameAveragingFreeEnergyApprox{\nPolydisperseSet}\left(\F\right) & = \frac{1}{\volSOThree} \int_{\genRotRef \in \SOThree} \df{\genRotRef} \left(\clinkFrameAveragingFreeEnergyApprox{\nPolydisperseSet}\left(\F, \genRotRef\right)\right),
\end{align}
where we introduce $\clinkFrameAveragingFreeEnergyApprox{\nPolydisperseSet}$ to denote the free energy of an RVE in the frame averaging limit with \begin{hlbreakable}$\ncollection{\nPolydisperseSet}$\end{hlbreakable} monomers in its chains, 
$\delYcQO$ is a small perturbation of the cross-link position, and $\smallparam = \left|\delYcQO\right| / \left(\rmsChainLength\right)_{\max} \ll 1$.
For brevity, let the inner free energy cost be denoted by
\begin{equation} \label{eq:inner-FED-FA}
    \clinkInnerFrameAveragingFreeEnergy\left(\delYcQO, \genRotRef\right) = \sum_{i=1}^{\numChains} \chainFreeEnergy_i\left(\left|\F \genRotRef \Xvec_i - \left(\ycMono + \delYcQO\right)\right|\right),
\end{equation}
and then we approximate the derivative with respect to $\delYcQO$ to linear order in $\delYcQO$ as
\begin{equation}
    \label{eq:perturb-deriv-FA-2}
    \partialx{\clinkInnerFrameAveragingFreeEnergy}{\left(\delYcQO\right)} = \partialx{\clinkInnerFrameAveragingFreeEnergy}{\left(\delYcQO\right)}\Bigg|_{\nullvec} + \left(\frac{\partial^2 \clinkInnerFrameAveragingFreeEnergy}{\partial \delYcQO \: \partial \delYcQO}\Bigg|_{\nullvec}\right) \delYcQO + \orderOf{\smallparam^2}.
\end{equation}
To arrive at an approximation that minimizes $\clinkInnerFrameAveragingFreeEnergy$, we solve for $\delYcQO$ by dropping higher order terms (i.e., $\orderOf{\smallparam^2}$) and setting \Eqref{eq:perturb-deriv-FA-2} equal to $\nullvec$,
\begin{equation}
    \label{eq:FA-analytical-model}
    \delYcQO = -\left(\frac{\partial^2 \clinkInnerFrameAveragingFreeEnergy}{\partial \delYcQO \: \partial \delYcQO}\Bigg|_{\nullvec}\right)^{-1}\partialx{\clinkInnerFrameAveragingFreeEnergy}{\left(\delYcQO\right)}\Bigg|_{\nullvec}.
\end{equation}
The free energy of the polydisperse RVE, $\clinkFrameAveragingFreeEnergyApprox{\nPolydisperseSet}$, can then be approximated by substituting $\delYcQO$ into \Eqref{eq:inner-FED-FA}.
Notably, since $\ycMono = \nullvec$, then the analytical form of each of the derivatives in \Eqref{eq:FA-analytical-model} is directly given in \Fref{app:cross-link-chain-free-energy-derivatives-junction-position} (when swapping $\delYcQO$ for $\yc$).
This convenience can then be utilized in deriving an analytical expression for $\delYcQO$.

Recall that the exact solution for $\delYcQO$ for polydisperse RVEs consisting of Gaussain chains is given by \Eqref{eq:gaussian-chain-unique-junction-position-relaxation} (where $\genRot=\genRotRef$). We supply the analytical form of $\delYcQO$ for the case of polydisperse Kuhn and Gr\"{u}n chains in \Fref{app:KG-approx-yc}.
But henceforth in this work, we exclusively utilize the form of $\delYcQO$ supplied in \Eqref{eq:gaussian-chain-unique-junction-position-relaxation}.

\section{Bimodal networks of Gaussian chains} \label{sec:gaussian}
In general, the probability density function over RVEs, $\probDens\left(\inst\right)$, may be complex (which could render the calculation of the network free energy density, as per \Eqref{eq:network-free-energy-density}, computationally costly).
In contrast, traditional (discrete) polymer networks only consider a single molecular weight with a fixed RVE geometry, which is equivalent to $\probDensChain\left(\n\right) = \dirac{\n}{\n_0}$, and has resulted in models that, in many cases, combine simplicity (i.e., a small number of fitted parameters) with quality fits.
Inspired by this, and motivated by the principle of parsimony to keep the number of model parameters as small as possible, we consider the analogous probability density for networks with a bimodal distribution of chain molecular weights that are exclusively joined at tetrafunctional cross-links ($\numChains=4$).

\begin{hlbreakable}We are further motivated by the experimental work of J.E. Mark and colleagues \cite{mark1984dependence,tang1984effect,andrady1980model,llorente1981modelXI,llorente1981modelXIII,mark1994elastomeric} to study bimodal tetrafunctionally cross-linked polymer networks.
Even in such parsimonious polydisperse networks, relatively complex mechanical phenomena emerge with respect to the molar concentration of short chains.
For instance, the maximum value in ultimate strength and toughness for these networks occurs at some intermediate short chain concentration value, indicating a complex load-sharing interaction between short and long chains.
This network parameter also controls how sharply the stress increases in the strain-stiffening regime.\end{hlbreakable}

Assuming that the number of monomers in each chain are independent and identically distributed (as discussed in \Fref{sec:cross-link-statistics}), we here represent the probability density over RVEs as
\begin{subequations}
\begin{align}
\probDens\left(\inst\right) &= \probDens_{4}\left(k\right) \prod_{i=1}^4 \probDensChain\left(\n_i\right), \\
\probDens_{4}\left(\numChains\right) &= \dirac{\numChains}{4}, \\
\probDensChain\left(\n\right) &= p \dirac{\n}{\n_a} + \left(1-p\right) \dirac{\n}{\n_b},
\end{align}
\end{subequations}
where the $\Xvec_i$ are related to $\n_i$ by \Eqref{eq:four-chain-xs}\hl{, and we adopt the convention $\n_a < \n_b$ (i.e., $a$ denotes the shorter chains and $b$ the longer)}.
To make connections to a realistic bimodal distribution of chain molecular weights, $\n_a$ and $\n_b$ can be thought of as the peaks of the distribution, and $p$ can be tuned to account for the relative probability masses (i.e., integration of density ``local'' to each peak) between the peaks.
In a discrete random variable sense, $p$ is the probability of a given chain having $\n_a$ monomers and $1 - p$ is the probability of a given chain having $\n_b$ monomers.

Given the independence and discreteness of monomer numbers, the free energy density takes the simplified form,
\begin{equation} \label{eq:polydisperse-FED}
\begin{split}
    \allClinksFreeEnergyDensity\left(\F\right) & = \frac{\crossDensity}{2} \Bigg(
    p^4\binom{4}{4} \clinkFreeEnergyPolydisperse{\n_a, \n_a, \n_a, \n_a} 
    + p^3 \left(1-p\right) \binom{4}{3} \clinkFreeEnergyPolydisperse{\n_a, \n_b, \n_b, \n_b}  + p^2 \left(1-p\right)^2 \binom{4}{2} \clinkFreeEnergyPolydisperse{\n_a, \n_a, \n_b, \n_b}\\
    &\qquad\qquad + p \left(1-p\right)^3 \binom{4}{1} \clinkFreeEnergyPolydisperse{\n_a, \n_a, \n_a, \n_b} + \left(1-p\right)^4\binom{4}{0} \clinkFreeEnergyPolydisperse{\n_b, \n_b, \n_b, \n_b} \Bigg), \\
\end{split}
\end{equation}
where the binomial coefficients account for all of the permutations of monomers which give equivalent RVEs by symmetry (e.g., $\clinkFreeEnergyPolydisperse{\n_a, \n_b, \n_b, \n_b} = \clinkFreeEnergyPolydisperse{\n_b, \n_a, \n_b, \n_b} = \clinkFreeEnergyPolydisperse{\n_b, \n_b, \n_a, \n_b} = \clinkFreeEnergyPolydisperse{\n_b, \n_b, \n_b, \n_a}$ by symmetry, etc.). 
By considering Gaussian chains, the above model given by \Eqref{eq:polydisperse-FED} allows one to characterize the structure of a bimodal polymer network through $\n_a$, $\n_b$, $p$, $\crossDensity$, and $\mLen$.
\hl{The model is fully specified through the additional parameters temperature, $\T$, and chain free energy function, $\chainFreeEnergy$.
All have direct physical or structural interpretations; no phenomenological fitting constants are introduced.}

\begin{hlbreakable}Before considering \Eqref{eq:polydisperse-FED} in any more depth, we first turn our attention to the three distinct classes of bimodal $4$-chain RVEs that constitute it: the monodisperse RVE ($\clinkFreeEnergyPolydisperse{\n_a, \n_a, \n_a, \n_a}$ and $\clinkFreeEnergyPolydisperse{\n_b, \n_b, \n_b, \n_b}$), which is already known to have the equipartition of stretch solution (\Fref{prop:unification-FR}) for the free rotation limit and $\ycMono=\nullvec$ for the frame averaging limit; the RVE with one chain of one length and three chains of the other length ($\clinkFreeEnergyPolydisperse{\n_a, \n_b, \n_b, \n_b}$ and $\clinkFreeEnergyPolydisperse{\n_a, \n_a, \n_a, \n_b}$); and the RVE with two chains of each length $\clinkFreeEnergyPolydisperse{\n_a, \n_a, \n_b, \n_b}$.
For brevity, we will refer to the second and third RVE classes as the 1-3 and 2-2 RVEs, respectively.
Since the monodisperse RVE is well-known and exactly characterized, we will focus primarily on the 1-3 and 2-2 RVEs (and refer back to the monodisperse RVE only for comparison).\end{hlbreakable}

In what follows, we first investigate the behavior of these RVEs in the free rotation limit, and then briefly highlight the corresponding frame averaging limit behavior.

\subsection{Free rotation limit} \label{sec:gaussian-FR}

\paragraph*{Closed-form approximation for the 1-3 RVE in the free rotation limit.}
For the 1-3 RVE, \Eqref{eq:$4$-chain-analytical-model} takes on the considerably simpler form:
\begin{subequations} \label{eq:bimodal-analytical-1}
\begin{align}
    \widetilde{\delRotMag}_1 &= \delRotMag_1 = 0, \\
    \widetilde{\delRotMag}_2 &= \auxSign{2}\left(\frac{3\left(\pStretch{1}^2 + \pStretch{2}^2\right)}{\pStretch{1}^2 + \pStretch{2}^2 + \pStretch{1}^4 \pStretch{2}^4} - 2\right), \\
    \widetilde{\delRotMag}_3 &= \auxSign{3}\left(\frac{3 \pStretch{1}^2}{\pStretch{1}^2 + \pStretch{2}^2 + \pStretch{1}^4 \pStretch{2}^4} - 1\right), \\
    \widetilde{\delYcMag}_1 &= \auxSign{2} \auxSign{3} \mLen \frac{\sqrt{3 \nA \nB}\left(\sqrt{\nB} - \sqrt{\nA}\right)\pStretch{1}\pStretch{2}\left(\pStretch{1} + \pStretch{2}\right)}{2\left(3\nA + \nB\right)\left(\pStretch{1}^2 + \pStretch{2}^2 + \pStretch{1}^4 \pStretch{2}^4\right)}, \\
    \widetilde{\delYcMag}_2 &= \auxSign{2} \mLen \frac{\sqrt{3 \nA \nB}\left(\sqrt{\nB} - \sqrt{\nA}\right)\pStretch{1}\pStretch{2}\left(\pStretch{1} + \pStretch{2}\right)}{2\left(3\nA + \nB\right)\left(\pStretch{1}^2 + \pStretch{2}^2 + \pStretch{1}^4 \pStretch{2}^4\right)}, \\
    \widetilde{\delYcMag}_3 &= \auxSign{3} \mLen \frac{\sqrt{3 \nA \nB}\left(\sqrt{\nB} - \sqrt{\nA}\right)\pStretch{1}^3\pStretch{2}^3}{2\left(3\nA + \nB\right)\left(\pStretch{1}^2 + \pStretch{2}^2 + \pStretch{1}^4 \pStretch{2}^4\right)},
\end{align}
where $\nA$ and $\nB$ refer to the number of monomers in the 1-3 RVE with multiplicity 1 and 3, respectively; and where $\auxSign{2} = \pm1$, and $\auxSign{3} = \pm1$ can vary independently to give $4$ different solutions.
Remarkably, the free energies that result from the $4$ different solutions are all identical.\footnote{
    Verified in the Mathematica notebook that can be found at \url{https://github.com/grasingerm/polydisperse-network-models}.
}
Similarly, the magnitudes of cross-link displacement and RVE rotations are invariant as
\begin{align}
    \left| \delRot \right| &= \sqrt{5 - \frac{3\left(4 \pStretch{1}^2 \pStretch{2}^2 + \left(1 + 6 \pStretch{1}^6\right)\pStretch{2}^4 + 4 \pStretch{1}^4 \pStretch{2}^6\right)}{\left(\pStretch{1}^2 + \pStretch{2}^2 + \pStretch{1}^4\pStretch{2}^4\right)^2}}, \\
    \left| \delYc \right| &= \frac{\mLen\left|\sqrt{\nA} - \sqrt{\nB}\right| \pStretch{1} \pStretch{2}}{3 \nA + \nB} \sqrt{\frac{
        3 \nA \nB 
        }{
        \left(\pStretch{1}^2 + \pStretch{2}^2 + \pStretch{1}^4\pStretch{2}^4\right)
        }}.
\end{align}
\end{subequations}

\paragraph*{Closed-form approximation for the 2-2 RVE in the free rotation limit.}
Likewise, \Eqref{eq:$4$-chain-analytical-model} takes on the simplified form for the 2-2 RVE,
\begin{subequations} \label{eq:bimodal-analytical-2}
\begin{align}
    \widetilde{\delRot} &= \delRot = \nullvec, \\
    \widetilde{\delYcMag}_i &= \pm \frac{\mLen \eff}{2 \sqrt{3}} \left(\sqrt{\n_b} - \sqrt{\n_a}\right) \pStretch{i}, \\
    \widetilde{\delYcMag}_j &= 0,
\end{align}
\end{subequations}
where $\eff = \ngeomean / \nmean$. Here, there are $6$ separate solutions as $i = 1, 2, 3$ and $j \neq i$.
For the 2-2 RVE, the free energy can be formulated concisely as
\begin{equation}
    \kB \T \left( \left(1 + \eff\right) \pStretch{i}^2 + 2 \sum_{j=1, j\neq i}^3 \pStretch{j}^2\right),
\end{equation}
where $i = 1, 2, 3$ and $\pStretch{3} = 1 / \pStretch{1}^2 \pStretch{2}^2$; of the $6$ solutions for $\delRot$ and $\delYc$, there are only $3$ associated forms of the free energy that are unique.
Of the $3$, the minimum free energy is always given by the case where $\pStretch{i} = \max \ncollection{\pStretch{1}, \pStretch{2}, \pStretch{3}}$.
This is because $\eff \leq 1$.
Note that equality $\eff=1$ occurs if and only if $\n_a=\n_b$, recovering the Neo-Hookean model.
The factor $(1+\eff)$ in front of the maximum principal stretch is reduced relative to the factor $2$ in front of the other principal stretches when $\eff < 1$ (i.e., when polydispersity is significant).
This indicates that the 2-2 RVE exhibits reduced stiffness in the maximum principal stretch direction, with the degree of reduction determined by the polydispersity parameter $\eff$.
\begin{hlbreakable}Physically, the RVE allows relaxation by preferentially aligning softer chains (those with larger monomer number $\n$) along the maximum stretch direction, thereby reducing resistance to deformation in that direction.
This qualitatively resonates with the physical picture of non-affine chain deformation in bimodal networks as discussed in Andrady et al.~\cite{andrady1980model}.\end{hlbreakable}

\begin{figure}
	\centering
    \includegraphics[width=\linewidth]{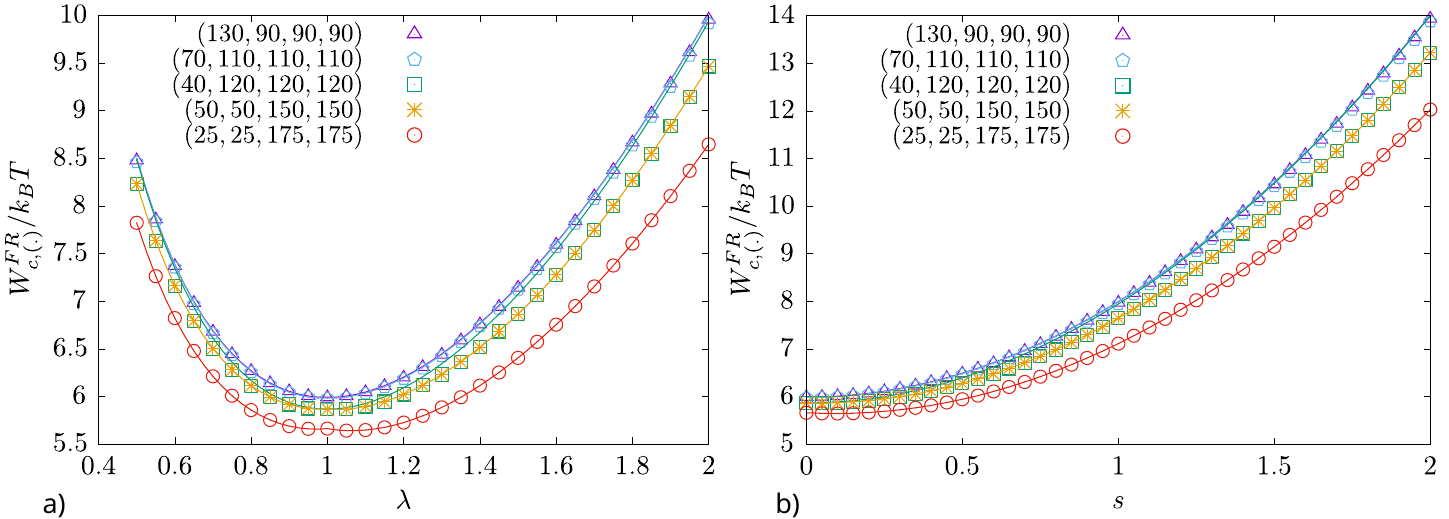}
	\caption{
            \begin{hlbreakable}\capTitle{Accuracy of free rotation RVE closed-form approximations.}\end{hlbreakable}
            Agreement of closed-form approximation for free rotation RVEs with various degrees of polydispersity undergoing uniaxial a) and simple shear b) deformations.
	\label{fig:WGauss-fit}}
\end{figure}

\paragraph*{Junction fluctuation term.}
For networks of Gaussian chains, the junction fluctuation term takes a form that is invariant with respect to the macroscopic deformation (e.g., see \Eqref{eq:gaussian-cross-link-free-energy-hessian-wrt-junction-position}).
This invariance can be physically understood as follows: the Hessian of the cross-link junction positional energy is determined by the tangent stiffnesses of the connected polymer chains about the optimal position, $\ycStar$. 
Because Gaussian chains exhibit a constant stiffness -- lacking the strain stiffening or softening observed in finite extensibility models -- the Hessian remains constant regardless of the individual chain stretches. 
Consequently, the fluctuation term is independent of the continuum-scale deformation, thereby justifying its exclusion from the mechanical stress analysis presented in this work.
Indeed, the determinants of the Hessians for the 1-3 and 2-2 RVEs are given by
\begin{subequations} \label{eq:jft}
\begin{align}
    \det \left(\left.\frac{\partial^2 \Wchains}{\partial \yc \partial \yc}\right|_{\genRot = \genRotStar, \yc = \ycStar}\right) &= \frac{27 \left(\kB \T\right)^3}{\mLen^6} \left(\frac{3}{\nA} + \frac{1}{\nB}\right)^3, \\
    \det \left(\left.\frac{\partial^2 \Wchains}{\partial \yc \partial \yc}\right|_{\genRot = \genRotStar, \yc = \ycStar}\right) &= \frac{27 \left(\kB \T\right)^3}{\mLen^6} \left(\frac{2}{\n_a} + \frac{2}{\n_b}\right)^3,
\end{align}
\end{subequations}
respectively.
Note that each of the terms in the parentheses of \Eqref{eq:jft} appear as a ratio of a monomer number to the number of chains in the RVE with that particular monomer number (e.g., for the 1-3 RVE, there are $3$ chains with $\nA$ monomers and $1$ chain with $\nB$ monomers).

\begin{figure}
	\centering
	\includegraphics[width=\linewidth]{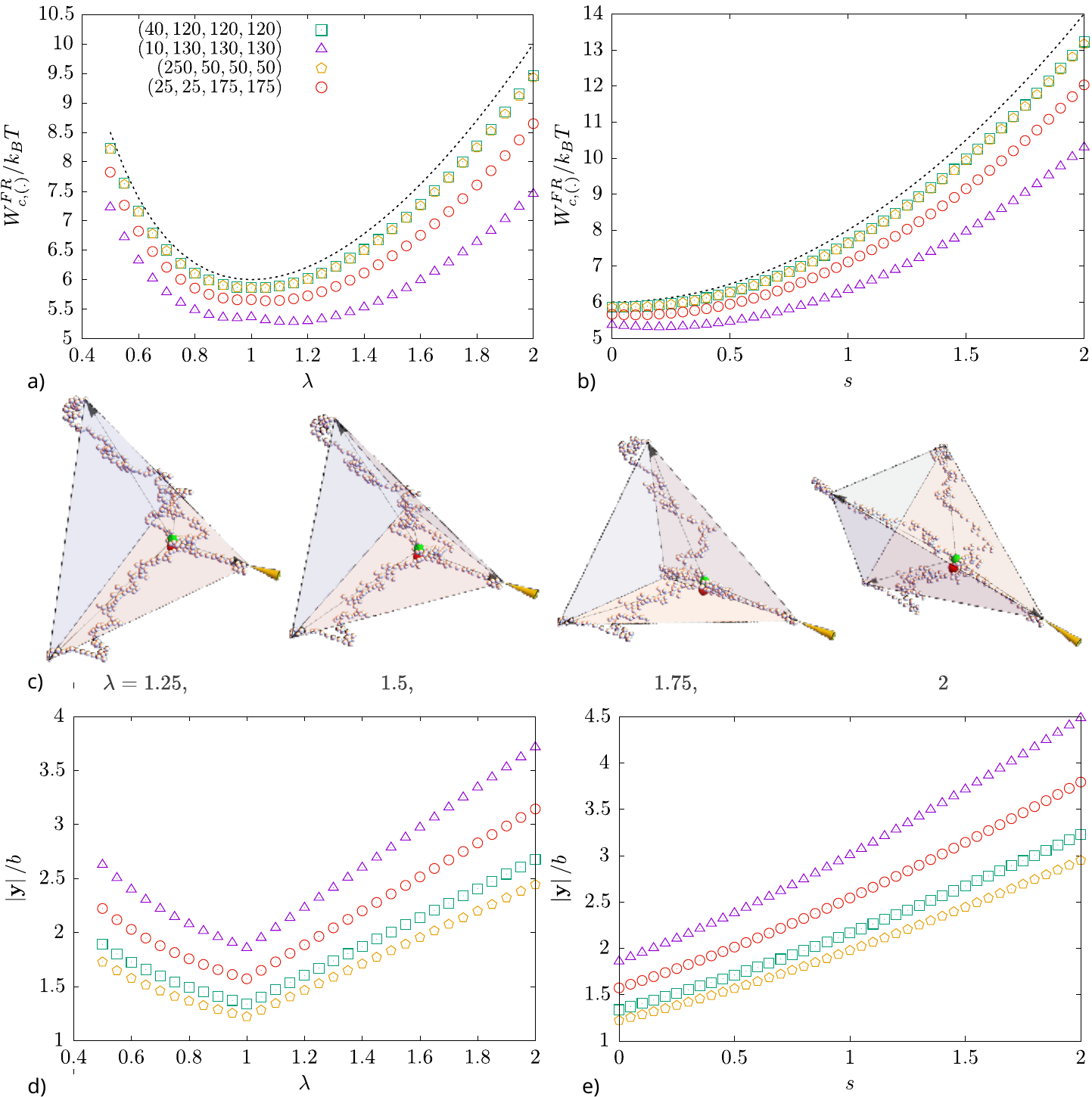}
	\caption{
            \begin{hlbreakable}\capTitle{Free rotation RVE response to deformation.}\end{hlbreakable}
            Free energy response for free rotation RVEs with various degrees of polydispersity for a) uniaxial and b) simple shear deformations.
            c) $\ncollection{40, 120, 120, 120}$ free rotation RVEs under uniaxial deformations of $\pStretchSymbol = 1.25, 1.5, 1.75,$ and $2$.
            Cross-link displacements for polydisperse free rotation RVEs under d) uniaxial and e) simple shear deformations.
	}
	\label{fig:W-var}
\end{figure}

\paragraph*{Accuracy of closed-form approximations.}
The closed-form approximation agrees well with the numerical solutions for the examples considered herein.
First consider the uniaxial deformation, $\F = \diag\left(\pStretchSymbol, 1/\sqrt{\pStretchSymbol}, 1/\sqrt{\pStretchSymbol}\right)$.
\Figref{fig:WGauss-fit} a) shows $\clinkFreeEnergyPolydisperse{.}$ as a function of $\pStretchSymbol$ for $\ncollection{130,90,90,90}$ (\textcolor{purple} {$\triangle$}), $\ncollection{70,110,110,110}$ (\textcolor{blue}{$\pentagon$}), $\ncollection{40,120,120,120}$ (\textcolor{green}{$\square$}), $\ncollection{50,50,150,150}$ (\textcolor{orange}{$*$}), and $\ncollection{25,25,175,175}$ (\textcolor{red}{$\bigcircle$}).
Numerical solutions are represented by markers, whereas closed-form approximations are given by a solid line of the corresponding color.
For the 1-3 RVEs, the approximation is nearly exact for the two cases with lower variance, but there is some disagreement at higher stretches for the $\ncollection{40,120,120,120}$ case.
Despite the high variance, the 2-2 RVEs ($\ncollection{50,50,150,150}$ and $\ncollection{25,25,175,175}$) show almost exact agreement at all the deformations considered.
Note that, for the 1-3 RVEs, the amount of variance in $\n_i$ seems to drive behavior more than whether the $1$ chain length is greater than or less than the length of the remaining $3$ chains. 
In fact, the behavior of the $\ncollection{130,90,90,90}$ and $\ncollection{70,110,110,110}$ are nearly identical, as are the pair $\ncollection{40,120,120,120}$ and $\ncollection{50,50,150,150}$.
The two classes of behavior here appear to be organized primarily by the amount of variance in $\n_i$ (and, correspondingly, $\eff$).
The free energy response of a second deformation mode, simple shear, is shown in \Figref{fig:WGauss-fit} b). 
Here $\F = \identity + s \euclid{1} \otimes \euclid{3}$, and the principal stretches are $\pStretch{1} = \sqrt{2 + s^2 + s \sqrt{4 + s^2}} / \sqrt{2}, \pStretch{2} = 1$, and $\pStretch{3} = \sqrt{2 + s^2 - s \sqrt{4 + s^2}} / \sqrt{2}$.
Similar phenomena can be observed for simple shear, illustrating their generality to other types of deformation.

\begin{figure}
	\centering
	\includegraphics[width=\linewidth]{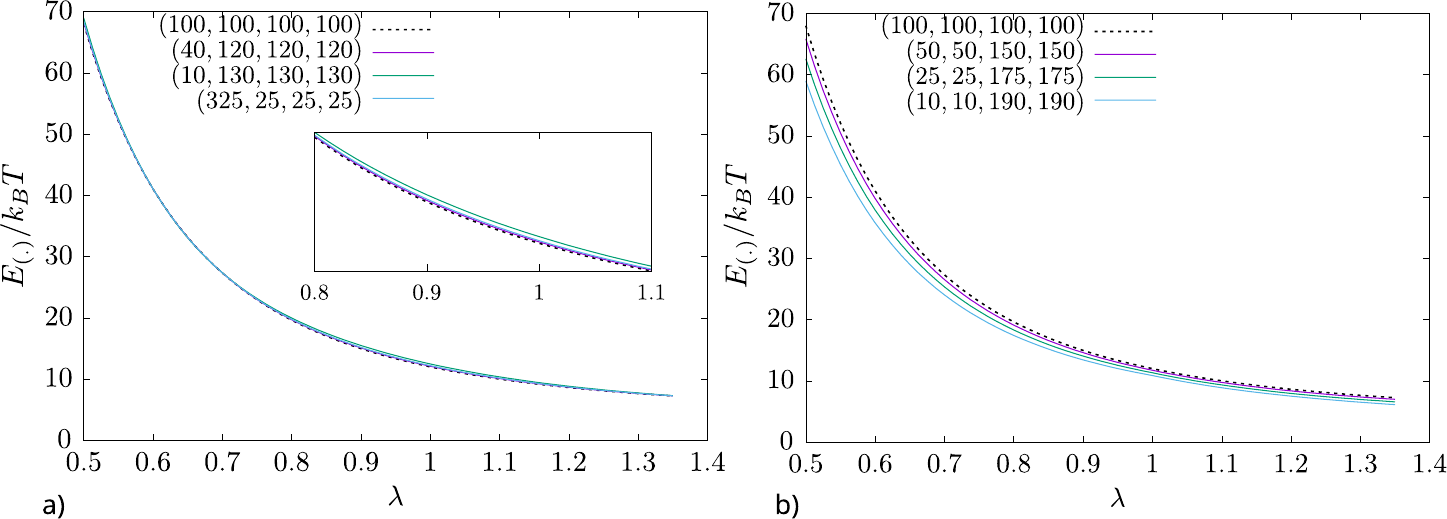}
	\caption{
            \begin{hlbreakable}\capTitle{Tangent stiffness modulus for $4$-chain free rotation RVEs.}\end{hlbreakable}
            Tangent stiffness modulus, $\E = \partial^2\clinkFreeEnergyPolydisperse{.} / \partial \pStretchSymbol^2$, for various polydisperse a) 1-3 and b) 2-2 free rotation RVEs.
	}
    \label{fig:E-var}
\end{figure}

\begin{figure}
	\centering
    \includegraphics[width=\linewidth]{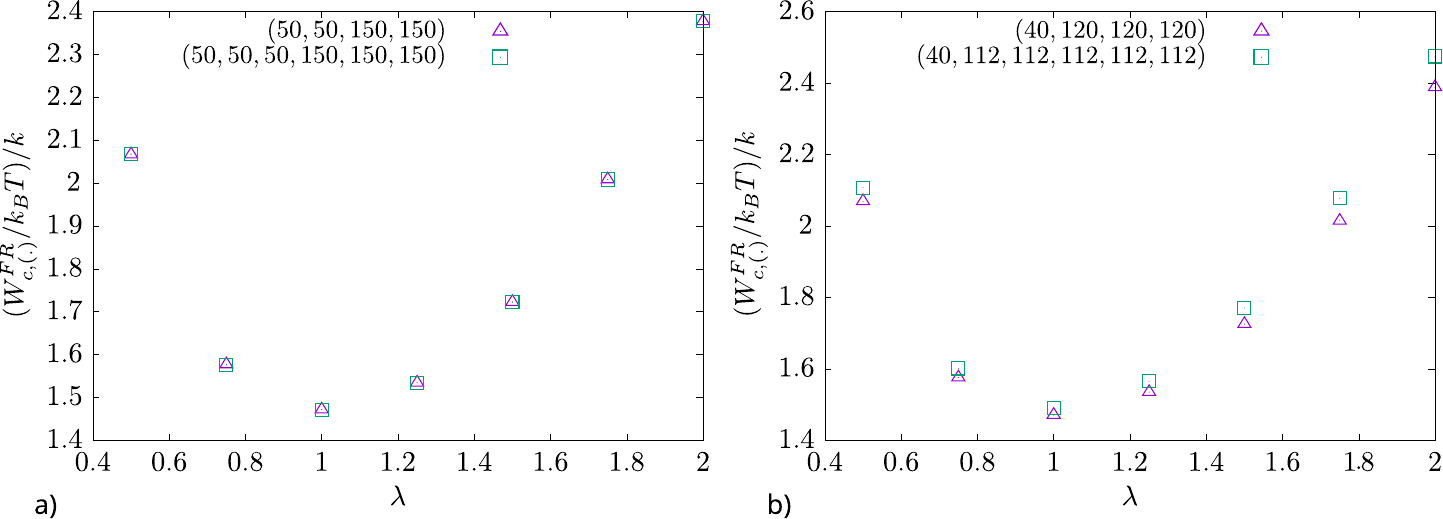}
	\caption{
            \capTitle{Comparison between $4$-chain and $6$-chain free rotation RVEs.}
		  a) For the cases where half of the chains are of one length, and half the other, the behavior is the same.
            b) However, for fixed average but one chain of a different length, there is a difference in behavior between the $4$-chain and $6$-chain RVEs.
	\label{fig:W-4vs6}}
\end{figure}

\paragraph*{Degree of polydispersity.}
To highlight some additional phenomena, \Figref{fig:W-var} shows the free energy response of polydisperse free rotation RVEs with a fixed $\nmean$ and wider range of monomer number variances when subjected to a) uniaxial deformation and b) simple shear deformation.
An increase in monomer number variance leads to lower free energies, and broader, flatter free energy curves, which has some correspondence with softer networks.
We point to the $\ncollection{10,130,130,130}$ cross-link RVE as an example where free energy minima do not occur at $\pStretchSymbol = 1$; it is also clearly nonconvex. 
The free energy in the free rotation limit can exhibit nonconvexity because it is approximated as the minimum over $4$ candidate equilibria, each corresponding to a local minima of the rotation optimization. 
Although each candidate free energy function appears convex, their pointwise minimum is not guaranteed to preserve convexity. 
Nonconvexity occurs when the global minimum switches from one candidate to another as the deformation varies.
This is a potential limitation with either the approximation in the free rotation limit, the method developed herein for the construction of polydisperse RVEs, or both; however, we note that these cases appear isolated to RVEs with large variance, and where at least $1$ of the chains has a small number of monomers (i.e., $\n_i / \n_j \ll 1$ for some $i$ and $j$).
\Figref{fig:W-var} c) shows the rotated and deformed $\ncollection{40,120,120,120}$ RVE when uniaxially deformed with $\pStretchSymbol = 1.25, 1.5, 1.75,$ and $2.0$.
The cross-link position in the reference configuration and the deformed configuration are denoted by red and green points, respectively.
The rotation and displacement of the cross-link increase with deformation.
\Figref{fig:W-var} d) and e) show the magnitude of cross-link displacement, $\left|\yc\right| / \mLen$, as a function of deformation.
The relationship appears nearly linear, where the RVEs with a higher degree of polydispersity have a higher slope; that is, the cross-link displaces more with deformation when there is more variance in $\ncollection{\n_i}_{i=1}^4$.

\begin{figure}
	\centering
    \includegraphics[width=\linewidth]{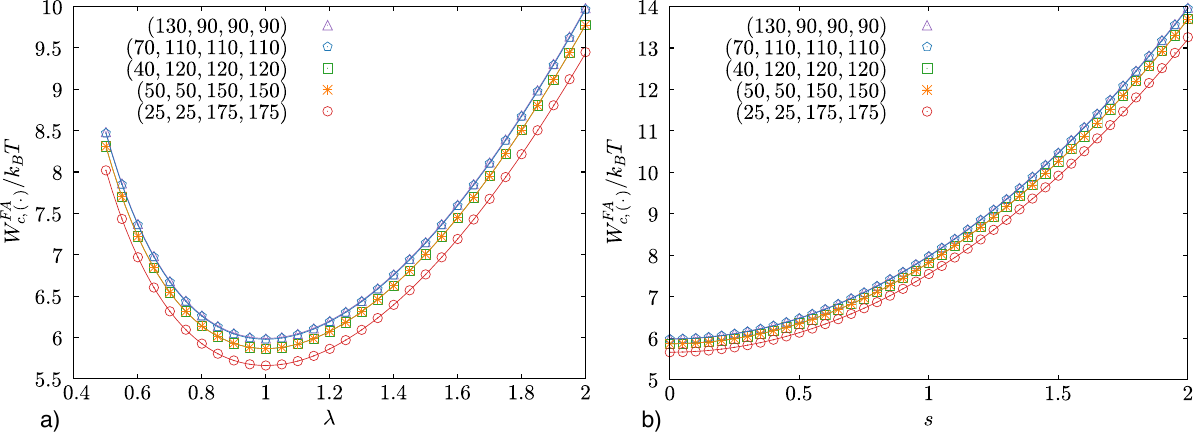}
	\caption{
            \begin{hlbreakable}\capTitle{Accuracy of frame averaging RVE closed-form approximations.}\end{hlbreakable}
            Agreement of closed-form approximation for frame averaging RVEs with various degrees of polydispersity undergoing uniaxial a) and simple shear b) deformations.
            Note that the Constrained Optimization BY Linear Approximation (COBYLA) algorithm~\cite{Zhang_2023} (offered by the \texttt{scipy.optimize.minimize} function) was used here for local constrained optimization to produce these specific frame averaging numerical results.
	\label{fig:FA-WGauss-fit}}
\end{figure}

The change in mechanical properties with polydispersity can be made more precise by investigating the tangent stiffness modulus, $\E_{\ncollection{.}} = \partial^2\clinkFreeEnergyPolydisperse{.} / \partial \pStretchSymbol^2$, with respect to uniaxial deformation.
For simplicity, we only derive an expression for the tangent stiffness modulus of the 1-3 RVE for $1$ of the $4$ candidate minima with respect to rotations.
(The chosen minima corresponds with $s_2 = -s_3 = 1$ in \Eqref{eq:bimodal-analytical-1}.)
The tangent stiffness modulus for the 1-3 RVE is
\begin{equation}
    \E_{\ncollection{\nA, \nB, \nB, \nB}}\big|_{\pStretchSymbol = 1} = 3 \kB \T \left(\frac{15\nA + 7\nB - 6 \sqrt{\nA \nB}}{3\nA + \nB}\right),
\end{equation}
and, for the 2-2 RVE is
\begin{equation}
    \E_{\ncollection{\n_a, \n_a, \n_b, \n_b}}\big|_{\pStretchSymbol = 1} = 2 \kB \T \left(5 + \eff\right), 
\end{equation}
where, again, $\eff = \sqrt{\n_a \n_b} / \nmean$ is the efficiency.
If the average monomer number is constrained such that $4 \nmean = \nA + 3 \nB$ for the 1-3 RVE and $4 \nmean = 2 \n_a + 2\n_b$ for the 2-2 RVE, then, in either case, the stiffness modulus is an extrema if and only if $\nA = \nB = \nmean$ and $\n_a = \n_b = \nmean$, respectively; that is, when the RVE is monodisperse.
Upon inspection of the second derivatives of $\E$ with respect to monomer numbers, it is seen that the monodisperse case is indeed a maxima for the 2-2 RVE, but, surprisingly, \emph{a minima for the 1-3 RVE}.
This can also be seen in \Figref{fig:E-var}, which shows the tangent stiffness modulus as a function of $\pStretchSymbol$ for various polydisperse a) 1-3 and b) 2-2 RVEs.
Although the monodisperse case is a minima for 1-3 RVEs, the difference in stiffness is small between the RVEs considered.
Thus, in a full bimodal network model of the form \Eqref{eq:polydisperse-FED}, a net softening may often occur with increasing polydispersity.
This agrees qualitatively with the results of mesoscale simulations~\cite{lei2022network}.

\begin{figure}
	\centering
	\includegraphics[width=\linewidth]{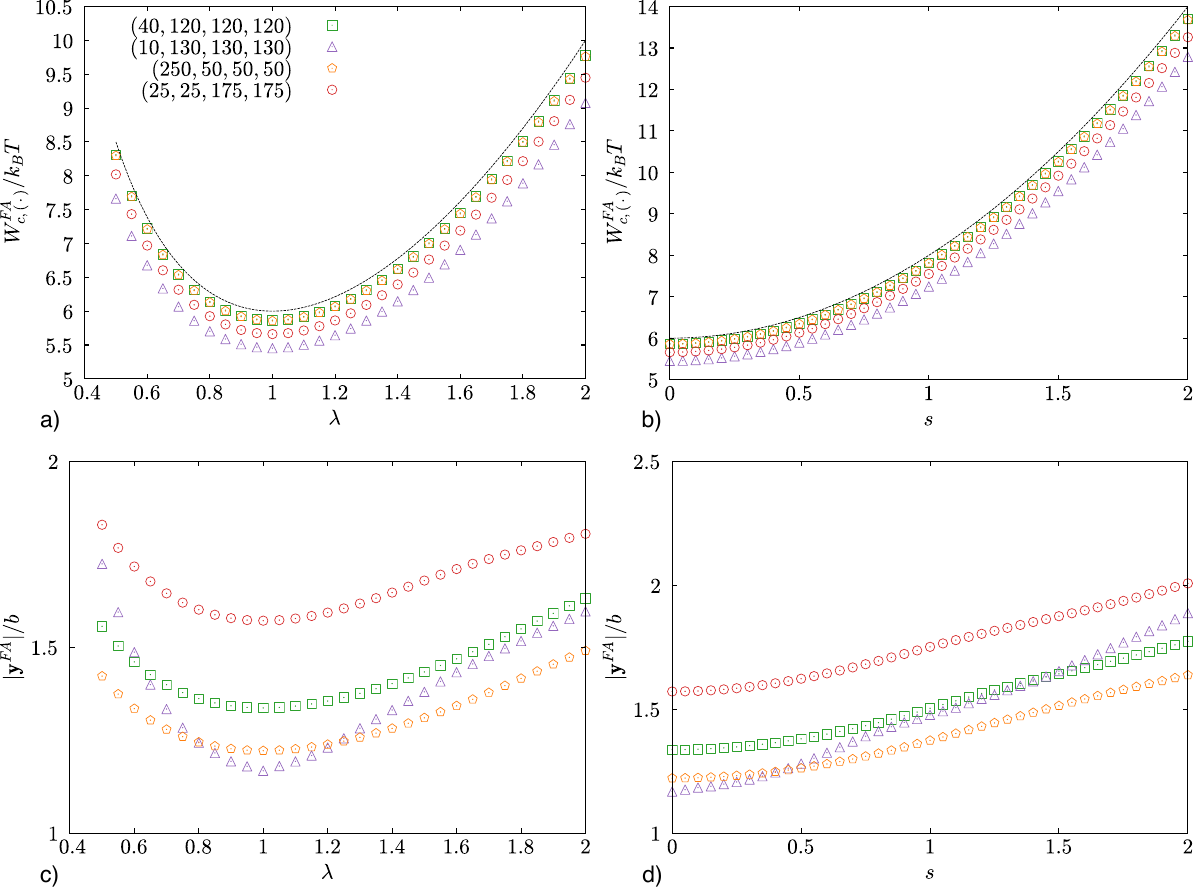}
	\caption{
            \begin{hlbreakable}\capTitle{Frame averaging RVE response to deformation.}\end{hlbreakable}
            Free energy response for frame averaging RVEs with various degrees of polydispersity for a) uniaxial and b) simple shear deformations.
            Cross-link displacements for polydisperse frame averaging RVEs under c) uniaxial and d) simple shear deformations.
	}
	\label{fig:FA-W-var}
\end{figure}

\paragraph*{Cross-links with degree $4$ versus degree $6$.}
For a monodisperse network consisting of FJC (or WLC) chains, the $4$-, $6$-, and $8$-chain RVEs produce the same constitutive model (see e.g., Remark~\ref{rem:mono-topology} or \cite{grasinger2023polymer}). 
Here, we consider the implications of topological differences on polydisperse networks. 
\Figref{fig:W-4vs6} shows the per-chain free energy response for polydisperse $4$- and $6$-chain RVEs undergoing uniaxial deformations.
The average number of monomers, $\nmean$, is held fixed at $100$.
For a), in both the $4$- and $6$-chain RVEs, half of the chains are shorter and half are longer.
The behavior is identical.
However, in b) only $1$ chain is of a smaller number of monomers, $40$.
In this case, the per-chain free energy response is different depending on whether the cross-link consists of $4$ chains or $6$ chains.

\begin{figure}
	\centering
	\includegraphics[width=\linewidth]{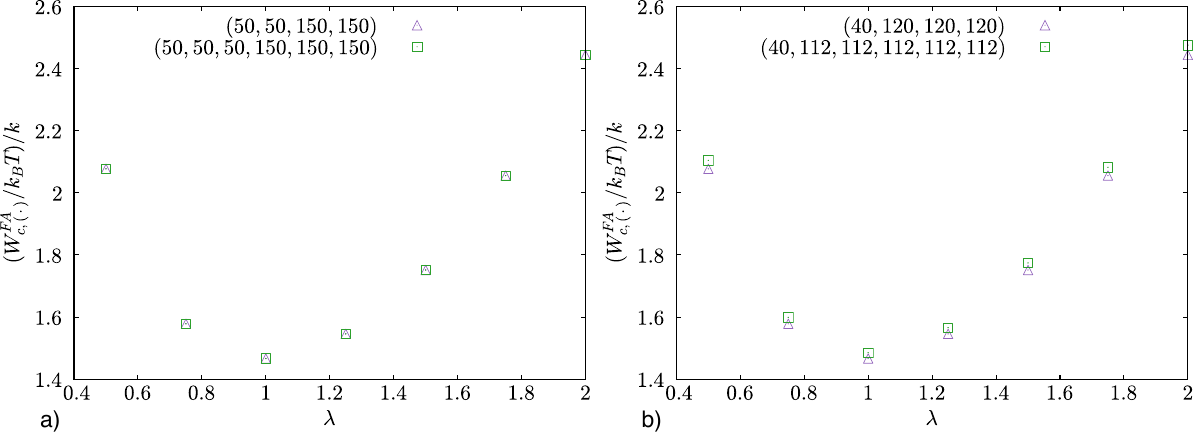}
	\caption{
            \capTitle{Comparison between $4$-chain and $6$-chain frame averaging RVEs.}
            The trends analogously follow those from the free rotation RVEs in \Figref{fig:W-4vs6}.
	\label{fig:FA-W-4vs6}}
\end{figure}

\subsection{Frame averaging limit} \label{sec:gaussian-FA}

In \Figref{fig:FA-WGauss-fit}, we present both the closed-form approximation (\Eqref{eq:gaussian-chain-unique-junction-position-relaxation} where $\genRot=\genRotRef$) and the numerical solution for several different polydisperse frame averaging RVEs in a) uniaxial deformation and b) simple shear deformation.
Here, the closed-form approximation agrees very well with the corresponding numerical solutions.
The free energy minima for each frame averaging RVE also occurs at $\pStretchSymbol = 1$ in uniaxial deformation and $s = 0$ in simple shear.
When compared with the analogous results for the free rotation model in \Figref{fig:WGauss-fit}, it is clear that free energy state of each frame averaging RVE is, at all deformation states, equal to or greater than the free energy state of its free rotation counterpart.

To highlight some additional phenomena, \Figref{fig:FA-W-var} shows the free energy response of polydisperse frame averaging RVEs with a fixed $\nmean$ and wider range of monomer number variances when subjected to a) uniaxial deformation and b) simple shear deformation.
Moreover, \Figref{fig:FA-W-var} c) and d) show the magnitude of frame averaging cross-link displacement, $\left|\ycFrameAveraging\right| / \mLen$, as a function of deformation.
When comparing \Figref{fig:FA-W-var} c) and d) to the analogous free rotation results in \Figref{fig:W-var} d) and e), several intriguing observations can be made.
First, it is clear that, during deformation, less cross-link displacement takes place in the frame averaging limit as compared to the free rotation limit.
With less cross-link displacement under the same macroscopically-imposed deformation conditions, the chains in the frame averaging limit are thus more deformed than their free rotation counterparts, which leads to the frame averaging RVE free energy state being greater than that for the free rotation limit.
Second, there seems to be a rather non-linear relationship between deformation and cross-link displacement in the frame averaging limit (which could possibly be attributed to numerical approximation error from using $\SOThree$ quadrature to approximately integrate over all possible cross-link orientations).
This is in contrast to the very nearly linear relationship between deformation and cross-link displacement in the free rotation limit.
Third, in the free rotation RVE, cross-link displacement increases when there is more variance in $\ncollection{\n_i}_{i=1}^4$.
This trend clearly does not hold for the frame averaging limit.

Finally, \Figref{fig:FA-W-4vs6} shows the per-chain free energy response for polydisperse $4$- and $6$-chain frame averaging RVEs undergoing uniaxial deformations.
These frame averaging RVEs exhibit similar behavior when compared to their free rotation counterparts in \Figref{fig:W-4vs6}.

\section{Finite extensibility and strain stiffening} \label{sec:s-curves}
The elastic response of many soft polymer networks consists of an initial Neo-Hookean regime followed by strain stiffening.
The result is a stress-strain curve that takes a characteristic `S' shape.
The Gaussian model for chain free energies, together with polymer network models, can recover the Neo-Hookean regime and help to elucidate the implications of macromolecular and network structure within this regime~\cite{treloar1975physics}; however, this combination cannot resolve the strain stiffening regime.
In contrast, the Kuhn and Gr\"{u}n chain free energy, $\KGFreeEnergy$, takes into account the finite extensibility of the FJC and, as a result, captures a strain stiffening regime at large deformations (see~\cite{arruda1993threee}).
Ideally, a closed-form approximation would be obtained using the chain free energy $\KGFreeEnergy$ within each RVE of interest, following the procedure outlined in \Fref{sec:FR-closed-form-approx} for the free rotation limit and \Fref{sec:FA-closed-form-approx} for the frame averaging limit.
However, for the free rotation limit, this proves intractable because $\KGFreeEnergy$ involves the inverse Langevin function, for which the exact analytical formula does not exist and approximations are often complex.
Complicating matters further, the inverse Langevin function is nested within $\csch$ and $\ln$ functions.
Instead, we take advantage of the fact that the Gaussian chain is a leading order approximation of $\KGFreeEnergy$.
The approximations $\delRot$ and $\delYc$ derived for Gaussian chains in \Eqref{eq:bimodal-analytical-1} and \Eqref{eq:bimodal-analytical-2} naturally serve as leading order approximations for free rotation RVEs with Kuhn and Gr\"{u}n chains.
A key consideration here is that \Eqref{eq:bimodal-analytical-1} represents $4$ different approximate solutions for the 1-3 RVE and \Eqref{eq:bimodal-analytical-2} represents $6$ for the 2-2 RVE.
As before, the best approximation for a given deformation in either case is the one that results in the minimum free energy.
For the 1-3 RVE, provided one follows the convention\footnote{The conventions chosen here are without loss of generality. Recall: a permutation of monomer numbers is equivalent to changing the frame about which we expand to obtain the closed-form approximation. Fixing a particular permutation and minimizing over the finite number of frames that are optimal for the monodisperse case is sufficient to satisfy material frame indifference.} that the chain with unique monomer number is taken to be chain $1$ (i.e., along $\left(0, 0, 1\right)$), then numerical results reveal that \Eqref{eq:bimodal-analytical-1} with $\auxSign{2} = -1$ and $\auxSign{3} = 1$ is the optimal approximation for the examples presented herein.
For the 2-2 RVE, provided the convention $\ncollection{\n_a, \n_a, \n_b, \n_b}$ is used for the undeformed chain orientations,
\begin{subequations} \label{eq:bimodal-analytical-3}
\begin{align}
    \delRot &= \nullvec, \\
    \delYcMag_1 &= \delYcMag_2 = 0, \\
    \delYcMag_3 &= \frac{\mLen \eff \left(\sqrt{\n_a} - \sqrt{\n_b}\right)}{2 \sqrt{3} \left(\pStretch{1} \pStretch{2}\right)}, 
\end{align}
\end{subequations}
is optimal for the examples presented herein.
A comparison of the closed-form approximation in the free rotation limit with numerical results is given in \Figref{fig:WKG-fit} for a) 1-3 RVEs and b) 2-2 RVEs.
As expected, the approximation remains very accurate at moderate strains (e.g., $\pStretchSymbol \lessapprox 1.3$).
The error increases with the combination of variance in the monomer numbers and large deformations.
While, in this case, there is a slight over prediction of the strain energies, the closed-form solution remains a good leading order approximation for the influence of polydispersity on the elasticity of the network.

Although a closed-form approximation using the Kuhn-Gr\"un chain free energy can be obtained for the frame averaging limit (\Eqref{eq:KG-chain-FA-analytical-model}), we instead use the Gaussian-based approximation from \Eqref{eq:gaussian-chain-unique-junction-position-relaxation}. This choice ensures direct comparability between the two limits and simplifies implementation, as the Gaussian approximation serves as the leading-order term for Kuhn-Gr\"un chains at moderate deformations. \Figref{fig:FA-WKG-fit} shows this approximation agrees well with numerical solutions for both a) 1-3 and b) 2-2 RVEs across all strains considered. 
As in the Gaussian case, the frame averaging limit consistently predicts equal or higher free energies than the free rotation limit.
\begin{figure}
	\centering
	\includegraphics[width=\linewidth]{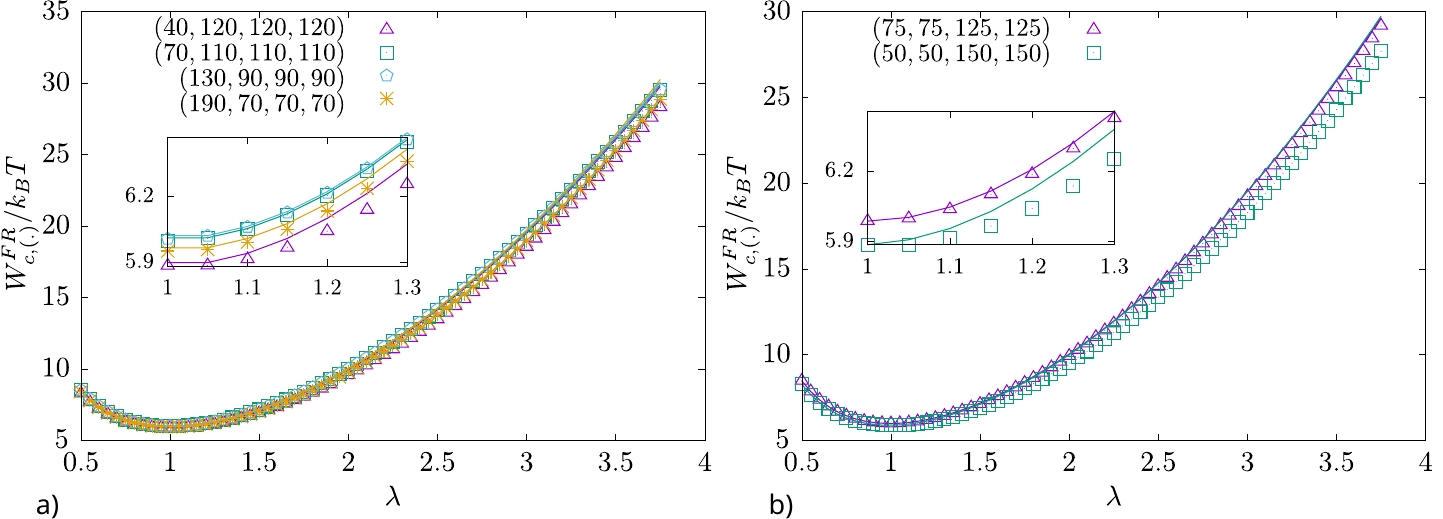}
	\caption{
            \begin{hlbreakable}\capTitle{Accuracy of free rotation RVE closed-form approximations for Kuhn and Gr\"{u}n chains.}\end{hlbreakable}
            Agreement of closed-form approximation for free rotation RVEs with Kuhn and Gr\"{u}n chains (exhibiting finite extensibility) in uniaxial deformation for the a) 1-3 RVEs and b) 2-2 RVEs.
	\label{fig:WKG-fit}}
\end{figure}

\begin{figure}
	\centering
	\includegraphics[width=\linewidth]{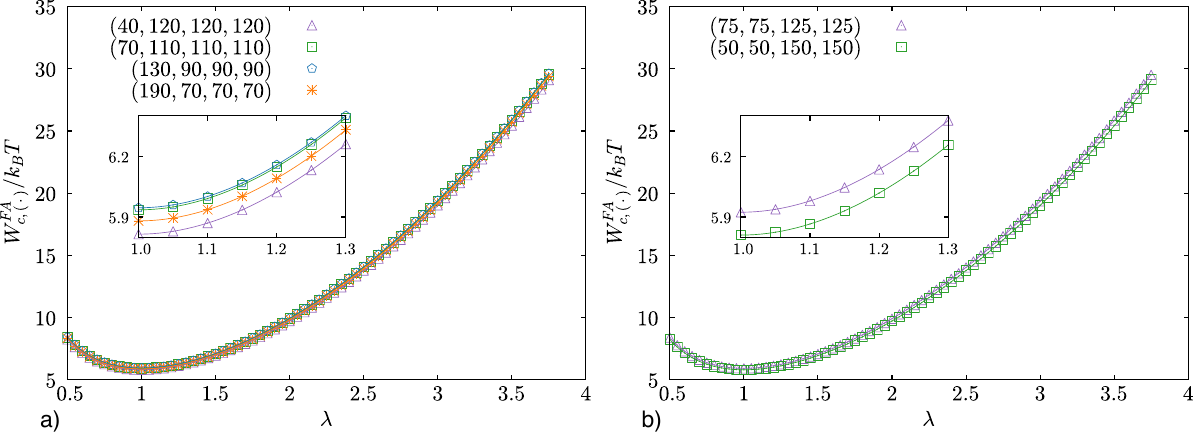}
	\caption{
            \begin{hlbreakable}\capTitle{Accuracy of frame averaging RVE closed-form approximations for Kuhn and Gr\"{u}n chains.}\end{hlbreakable}
            Agreement of closed-form approximation for frame averaging RVEs with Kuhn and Gr\"{u}n chains (exhibiting finite extensibility) in uniaxial deformation for the a) 1-3 RVEs and b) 2-2 RVEs.
	\label{fig:FA-WKG-fit}}
\end{figure}

\paragraph*{`S'-curves.}
Next we consider the features of the `S' shaped stress-strain curves that can be captured with the approach developed herein.
Qualitatively, the $2$ limits agree; a discussion of where and why they diverge will be discussed in the next paragraph.
Since the strain stiffening regime in the `S' is due to the finite extensibility of the chains, we again consider RVEs consisting of Kuhn and Gr\"{u}n chains.
Let $\allClinksFreeEnergyDensity$ be given by \Eqref{eq:polydisperse-FED}.
For the free rotation limit, the 1-3 RVEs are approximated by \Eqref{eq:bimodal-analytical-1} with $\auxSign{2} = -1$ and $\auxSign{3} = 1$, and the 2-2 RVEs are approximated by \Eqref{eq:bimodal-analytical-3}.
For the frame averaging limit, we again follow the approximation provided in \Eqref{eq:gaussian-chain-unique-junction-position-relaxation} for all polydisperse RVEs.
\Figref{fig:S-curves} shows $\PK = \partial \allClinksFreeEnergyDensity / \partial \F$ for the uniaxial stretch, $\pStretchSymbol$.
Recall that $p$ represents the probability that a chain has $\n_a$ monomers, and $1-p$ is the probability of a chain with $\n_b$ monomers.
In \Figref{fig:S-curves} a), $p$ and $\nmean$ are fixed at $p = 0.5$ and $\nmean = 75$, respectively, and the network has increasing disparities between $\n_a$ and $\n_b$, corresponding to increasing degrees of network polydispersity.
Notice that increasing polydispersity leads to the strain stiffening regime (i.e., sharp increase in stress) occurring at lower stretches and lower stresses.
In \Figref{fig:S-curves} b), $\n_a$ and $\n_b$ are fixed at $50$ and $100$, respectively, and $p$ varies as $0.1, 0.25, 0.5, 0.75$, and $0.9$.
Although strain stiffening occurs near the same stretch, the shape of the strain stiffening regime changes with $p$.
More precisely, the onset of strain stiffening is more gradual, and less sharp with increasing $p$.
\begin{hlbreakable}This result qualitatively agrees with experimental findings from J.E. Mark and colleagues regarding the emergence of strain-stiffening in bimodal tetrafunctionally cross-linked PDMS elastomers \cite{mark1984dependence,tang1984effect,andrady1980model,llorente1981modelXI,llorente1981modelXIII,mark1994elastomeric}.\end{hlbreakable}
\begin{hlbreakable}A few key takeaways become apparent here: \begin{inparaenum}[(i)] \item the minimum $\n_i$ appears to controls the stretch at which strain stiffening begins to occur, and \item  $p$ modulates how sharply the stress begins to diverge within that regime.\end{inparaenum}\end{hlbreakable}
The bimodal polymer network model developed herein is one that can capture more complex features in `S' stress-strain curves while only introducing a few additional model parameters that have clear physical meanings, are directly tied to the structure of the polymer network, and can potentially be measured, designed, and perhaps even synthesized with a certain degree of control.
\begin{figure}
	\centering
	\includegraphics[width=\linewidth]{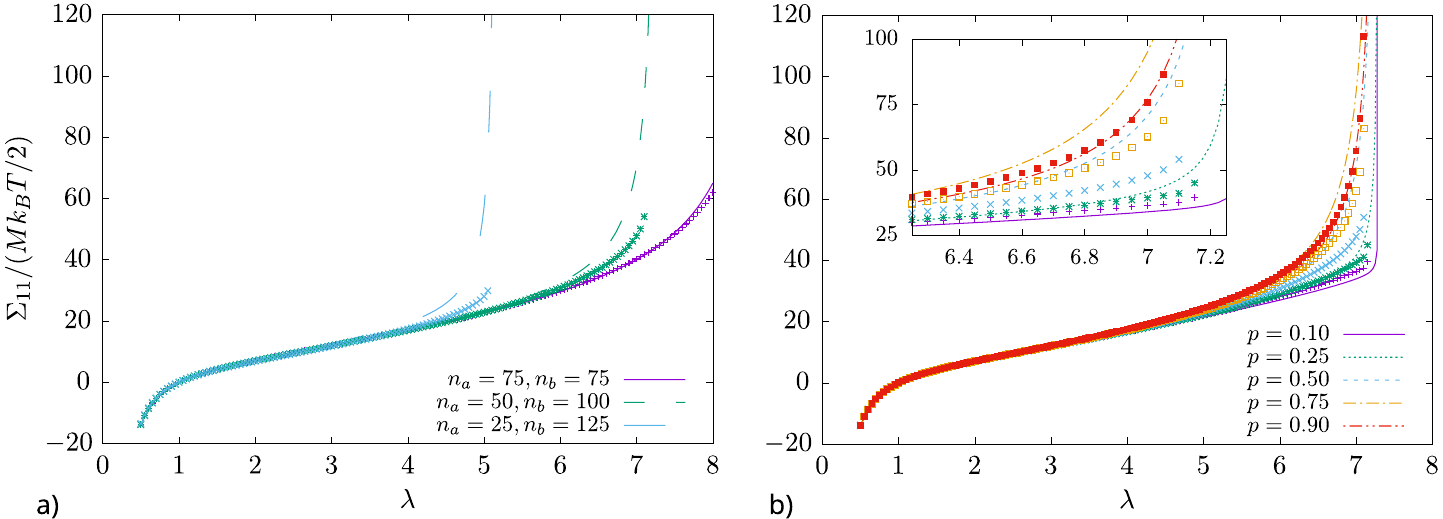}
	\caption{
            \capTitle{`S' shaped stress-strain curves common in soft polymer network elasticity.}
		a) Stress-strain curves for bimodal networks with $p = 0.5$, $\nmean = 75$, and varying $\n_a$ \hl{(monomer number of shorter chains)} and $\n_b$ (longer chains) (lines: free rotation limit; markers: frame averaging limit).
        Greater disparity between $\n_a$ and $\n_b$ leads to earlier onset of strain stiffening.
        b) Stress-strain curves for networks with $\n_a = 50$, $\n_b = 100$, and varying $p$, where $p$ is the probability that a chain has $\n_a$ monomers and $1-p$ the probability of $\n_b$ monomers (lines: free rotation limit; markers: frame averaging limit).
        Increasing $p$ leads to a more gradual, less sharp stress increase in the strain stiffening regime.
	}
	\label{fig:S-curves}
\end{figure}

\paragraph*{Comparison of free rotation and frame averaging limits.}
In principle, the free rotation limit has access to a greater number of relaxation modes than the frame averaging limit; consequently, the free energy and stress predicted by the free rotation limit should theoretically always be less than or equal to its frame averaging counterpart. Surprisingly, \Figref{fig:S-curves} shows that, in many cases, the opposite trend is observed. This remarkable observation is not derived from physical reasons, but is instead an artifact of approximation error inherent to the closed-form solutions derived for the minimum energy frame. To be specific, there is an error in the approximation for the optimal frame that grows with both increasing deformation and an increasing degree of polydispersity. When the approximation for the optimal frame in the free rotation limit is expected to be accurate -- such as when the degree of polydispersity is low -- the free rotation and frame averaging limits agree nearly exactly. Altogether, this suggests that for the elasticity of polydisperse networks, averaging over all frames is generally as effective as minimizing over frames for modeling soft networks; furthermore, averaging over all frames tends to be superior to attempting to approximate the optimal frame for networks with larger degrees of polydispersity and at larger deformations.
This, combined with the simplicity of its formulation, make it an advantageous approach for modeling the elasticity of polydisperse networks.
It remains open how the two limits compare for polydisperse networks with multiphysics loadings, such as electroelasticity (e.g.,~\cite{grasinger2020architected,grasinger2023polymer,grasingerIPflexoelectricity}).
This presents an interesting topic for future work.

\section{Conclusion} \label{sec:conclusion}

\paragraph*{Summary.}
In this work, we introduce a new approach to polymer network modeling that shifts the perspective from polymer chains as the basic unit to cross-links and their connected chains.
A central feature of this framework is that cross-link junction positions are allowed to relax to satisfy local force balance.
Regarding the frame (orientation) of the Representative Volume Element (RVE), we explored two distinct limiting behaviors: the free rotation limit and the frame averaging limit.
The free rotation limit assumes the cross-link rotates to an optimal orientation that minimizes free energy, offering a physically intuitive mechanism that aligns well with quintessential experimental data for unimodal networks.
Conversely, the frame averaging limit incorporates structural heterogeneity by averaging over all possible cross-link orientations.
\hl{We speculate that the free rotation limit is more appropriate for softer materials with compliant internal structure, such as soft elastomers and gels, whereas the frame averaging limit may be especially well-suited to networks with significant torsional or bending stiffness at junctions, such as many biopolymer networks. While the two limits largely agree in their predictions for elastic response, important differences may emerge in multiphysics settings, in the presence of mechanical failure, or wherever rotational kinematics couple strongly to other physics.
Altogether, the overarching framework allows for non-affine deformation and a more efficient distribution of stretches and forces within the network compared to more restrictive assumptions like equal stretch or equal force theories, enabling the direct linking of statistical descriptors of the network structure to macroscopic mechanical response.}

To address potential computational challenges of integrating and optimizing in high-dimensional parameter spaces, we derived closed-form approximations for both limits.
From an analytical and numerical standpoint, however, the frame averaging limit proved to be cleaner to work with.
Physical insights gained from this framework reveal that greater variance in monomer numbers typically leads to a softening at the RVE level.
Furthermore, topological differences are shown to fundamentally alter the mechanical response even when the average chain length is conserved.
Finally, by investigating `S'-shaped stress-strain curves for bimodal networks, it was found that the onset of strain stiffening is governed by the shorter chains, while the sharpness of the stiffening reflects their proportion.

\paragraph*{Outlook.}
Given the observation that the frame averaging limit generally provides a good approximation to the free rotation limit, a potential simple heuristic for modeling ``soft'' polymer networks may be to define the free energy of a RVE as the minimum of their respective predictions.
Such an approach would naturally bound the energy of the system, leveraging the strengths of the frame averaging limit in regimes of high polydispersity while retaining the physical intuition of the free rotation limit where applicable.
While the main thrust of this work was to proceed analytically, it could prove to be both fruitful and interesting to explore the implications of this framework numerically.

A potential issue with the numerical approximation of stress in this context is the difficulty of approximating derivatives of the free energy density via finite differences.
Although possible, this approach is likely computationally expensive because the formulation involves optimization problems nested within multidimensional integration.
The core difficulty lies in the fact that the derivatives are currently ``outside'' these operations.
Ideally, derivatives would be ``pushed inside'' the operators; while this can be readily done with respect to the integrals for non-pathological networks, derivatives cannot be brought inside the inner optimization problems.
To address this, we propose an approximation inspired by the concept of a ``soft max'' (or ``soft min''), as popularized by the machine learning community (e.g.,~\cite{lecun2006tutorial}).
Here, the optimization problem is regularized and reposed as an integration, resembling Boltzmann statistics where a computational parameter, $\beta$, takes the place of inverse temperature:
\begin{align}
    \inf_{x} \: W(x) &\approx \frac{\int \df{x} \: \left( W e^{-\beta W}\right)}{\int \df{x} \: e^{-\beta W}},\\
    \arg \inf_{x} \: W(x) &\approx \frac{\int \df{x} \: \left(x e^{-\beta W}\right)}{\int \df{x} \: e^{-\beta W}},
\end{align}
provided $\beta \gg 1$.
As $\beta$ increases, the approximation approaches the true minimum, though numerical stability must be carefully managed.
The key advantage of replacing the optimization problem with an integral is that it allows one to push derivatives all the way down to the integrand, where they can be taken analytically or handled more easily by alternative means, such as automatic differentiation.
While the motivating example was an approximation of stress, this idea applies to the many quantities of interest that can be posed as derivatives of the free energy density.

\paragraph*{Closure.} Despite the focus on relatively simple polydisperse chain and cross-link distributions, this work opens new possibilities for investigating synthesis-morphology-property relationships in polymer networks in a rational and efficient manner. 
Future research directions include incorporating a distribution of chain conformations within the polymer network homogenization framework, investigating the implications of polydispersity on fracture mechanics and toughness in heterogeneous networks, polydisperse multifunctional polymer networks, and phase transitions in polydisperse biopolymer networks. 
The ability to model these networks on the basis of the statistical descriptors of the structure of the network represents a step toward the targeted design of polymer networks with specific mechanical properties.

\begin{hlbreakable}
\paragraph*{Postscript.} To aid the interested reader, we summarize the step-by-step process needed to computationally implement the model for numerical predictions in \Fref{app:model-implementation}.
%
\end{hlbreakable}

\section*{Software availability}
The code(s) used for analysis and generation of data for this work is available at \url{https://github.com/grasingerm/polydisperse-network-models} and \url{https://github.com/jasonmulderrig/polydisperse-polymer-networks}.

\begin{acknowledgments}
MG thanks Gal deBotton and Kaushik Dayal for insightful discussions and their encouragement in pursuit of the topic.
MG acknowledges the support of the Air Force Research Laboratory.
JPM gratefully acknowledges the support of the National Research Council (NRC) Research Associateship Program (administered by the National Academies of Sciences, Engineering, and Medicine).
MB acknowledges the support of Sandia National Laboratories.
Sandia National Laboratories is a multi-mission laboratory managed and operated by National Technology and Engineering Solutions of Sandia, LLC, a wholly owned subsidiary of Honeywell International, Inc., for the U.S. Department of Energy’s National Nuclear Security Administration under Contract No. DE-NA0003525. Any subjective views or opinions expressed in the paper do not necessarily represent the views of the U.S. Department of Energy or the U.S. Government. The U.S. Government retains and the publisher, by accepting the article for publication, acknowledges that the U.S. Government retains a non exclusive, paid-up, irrevocable, world-wide license to publish or reproduce the published form of this manuscript, or allow others to do so, for U.S. Government purposes.
\end{acknowledgments}

\appendix

\begin{hlbreakable}

\section{Nomenclature} \label{app:nomenclature}

\renewcommand{\arraystretch}{1.2}

\begin{longtable}{ll}

    \toprule
    \textbf{Symbol} & \textbf{Description} \\
    \midrule
    \endfirsthead
    
    \toprule
    \textbf{Symbol} & \textbf{Description} \\
    \midrule
    \endhead
    
    \bottomrule
    \endfoot
    RVE & Representative volume element \\
    $\collection{x_1, \dots, x_n},~\collection{x_i}_{i=1}^n$ & Set with elements $x_1, \dots, x_n$ and set shorthand notation \\
    $\left(x_1, \dots, x_n\right),~\left(x_i\right)_{i=1}^n$ & Tuple with elements $x_1, \dots, x_n$ and tuple shorthand notation \\
    $\n$ & Number of monomers in a chain \\
    $\mLen$ & Monomer length \\
    $\cLen$ & Chain contour length \\
    $\rvec$ & End-to-end chain vector (i.e., the vector from the beginning of the chain to its end) \\
    $\left|\generic\right|$ & \parbox{4.75in}{Absolute value of a scalar; Euclidean norm of a vector; Number of elements in a set (i.e., the cardinality of a set)} \\
    $\rmag$ & End-to-end chain length \\
    $\chainStretch$ & Chain stretch \\
    $\chainConformationProbDens(\rmag)$ & Chain conformation probability density \\
    $\generic_{\text{G}}$ & Parameter associated with Gaussian freely-jointed chains \\
    $\generic_{\text{KG}}$ & Parameter associated with Kuhn and Gr\"{u}n freely-jointed chains \\
    $\Lang,~\invLang$ & Langevin function and its inverse \\
    $\chainDistanceProbDens(\rmag)$ & \parbox{4.75in}{End-to-end chain length probability distribution for a free chain (irrespective of direction in space)} \\
    $\rmsChainLength$ & Root-mean-square end-to-end chain length \\
    $\critChainLength$ & Critical end-to-end chain length \\
    $\chainFreeEnergy$ & Chain free energy \\
    $\kB$ & Boltzmann's constant \\
    $\T$ & Absolute temperature \\
    $\forcevec$ & Chain force vector \\
    $\Xvec$ & Material point in the reference configuration \\
    $\Xvec_i$ & End position of chain $i$ in the reference configuration \\
    $\defMap$ & Deformation map \\
    $\xvec$ & Material point in the deformed configuration \\
    $\xvec_i$ & End position of chain $i$ in the deformed configuration \\
    $\F$ & Deformation gradient \\
    $\delta \Xvec,~\delta \xvec$ & \parbox{4.75in}{Infinitesimal change in position for a material point in the reference and deformed configurations} \\
    $\mathbf{u}$ & Deformation at a material point \\
    $\cGreen$ & Right Cauchy-Green tensor \\
    $\trans{\generic}$ & Transpose of a tensor \\
    $\det\left(\generic\right)$ & Determinant of a tensor \\
    $\SOThree$ & Group of three-dimensional rotations (i.e., proper orthogonal transformations) \\
    $\strTens$ & Right stretch tensor \\
    $\polRot$ & Rotation tensor from the polar decomposition of $\F$ such that $\F = \polRot \strTens$ \\
    $\dir{\generic}$ & Unit vector \\
    $\pStretchSymbol$ & Principal stretch \\
    $\pDir$ & Principal direction \\
    $\delta_{ij}$ & Kronecker delta \\
    $\pFrame$ & \parbox{4.75in}{Rotation tensor which rotates the Euclidean basis to align with the principal frame and diagonalize $\strTens$} \\
    $\strDiag$ & Principal stretch tensor \\
    $\diag\left(\generic\right)$ & Diagonal tensor operator \\
    $\freeEnergyDensity$ & Free energy density \\
    $\PK$ & Nominal stress (i.e., first Piola-Kirchhoff stress) tensor \\
    $\chainSpace$ & Space of polydisperse polymer network cross-links \\
    $\inst$ & Cross-link structure \\
    $\numChains$ & Number of chains connected to the cross-link junction point  \\
    $\yc$ & Position of the cross-link junction point \\
    $\RVEdomain$ & Domain of the cross-link RVE in its reference state \\
    $\Conv\left(\generic\right)$ & Convex hull operator \\
    $\euclid{1},~\euclid{2},~\euclid{3}$ & Euclidean basis vectors for $\Reals^3$ \\
    $\probDensInst$ & Polymer network cross-link probability density function \\
    $\probDens$ & Frame-invariant cross-link structure probability density function \\
    $\genRotRef$ & \parbox{4.75in}{Cross-link orientation tensor relative to a reference set of chain end positions (i.e., the preferred/ground-state orientation of the cross-link within an elastic background)} \\ 
    $\genRot$ & Cross-link rotation tensor \\
    $\probDensF{\genRotRef}$ & Reference cross-link orientation probability density function \\
    $\partitionFunction$ & Partition function \\
    $\rvedomain$ & Domain of the cross-link RVE associated with its current (rotated and deformed) state \\
    $\PEQ\left(\genRot\right)$ & \parbox{4.75in}{Torsional elastic energy from the surrounding network when rotating the cross-link by $\genRot$ prior to deformation} \\
    $\kNetTors$ & Torsional stiffness modulus \\
    $\lVert \generic \rVert$ & Tensor norm \\
    $\generic^{\star}$ & Optimal solution of a parameter \\
    $\ycQO$ & Position of the cross-link junction in the frame averaging limit where $\genRot = \genRotRef$ \\
    $\generic^{FR}$ & Generic parameter associated with the free rotation limit \\
    $\generic^{FA}$ & Generic parameter associated with the frame averaging limit \\
    $\Wchains$ & Total free energy of the chains in a cross-link structure \\
    $\identity$ & Identity tensor \\
    $\otimes$ & Dyadic product \\
    $\orderOf{\generic}$ & Order of \\
    $\clinkFreeEnergy$ & Free energy of a cross-link structure \\
    $\unitSphere$ & The unit 2-sphere (i.e., the boundary of the unit sphere in three-dimensional space) \\
    $\Rdir$ & Chain orientation vector in the full network model \\
    $\allClinksFreeEnergyDensity$ & Polymer network free energy density \\
    $\crossDensity$ & Volumetric cross-link number density \\
    $\probDensChain$ & Chain monomer number probability density function \\
    $\probDensClinker$ & Cross-linker functionality probability density function \\
    $p$ & Probability \\
    $\delta\left(\generic\right)$ & Dirac delta function \\
    $\rodVec$ & Rodrigues vector \\
    $\varphi$ & Angle of rotation \\
    $\dir{\vvec{u}}$ & Axis of rotation \\ 
    $\tens{A}$ & Generating skew-symmetric tensor \\
    $\genRot\left(\rodVec\right)$ & \parbox{4.75in}{Rotation matrix with respect to the Rodrigues vector as per the exponential or ``axis-angle'' representation} \\
    $\nullvec$ & Origin; null vector \\
    $\Ball{2\pi}{\nullvec}$ & Ball of radius $2\pi$ centered at the origin \\
    $\Rmag$ & Root-mean-square end-to-end chain length for chains in a monodisperse cross-link \\
    $\nuVec$ & Orthogonal direction vector \\
    $\changeCoord{\generic}$ & Real proper orthogonal similarity transformation of a tensor \\
    $\Tr\left(\generic\right)$ & Trace \\
    $\eulerx, \eulery, \eulerz$ & Rotation angles \\
    $\ncollection{\nPolydisperseSet}$ & Tuple of $\nPolydisperseSet$ monomer numbers in the first, \dots, $k$th chains of a cross-link \\
    $\clinkFreeEnergyPolydisperse{\nPolydisperseSet}$ & Free energy of a cross-link RVE with $\ncollection{\nPolydisperseSet}$ monomers in its chains \\
    $\delRot$, $\genRot\left(\delRot\right)$ & Perturbation of the Rodrigues vector and the corresponding perturbation of the rotation \\
    $\delYc$ & Perturbation of the cross-link position \\
    $\genRotMono$, $\ycMono$ & \parbox{4.75in}{Monodisperse cross-link solutions to the rotational and cross-link positional relaxation problem} \\
    $\smallparam$ & Approximation error \\
    $\clinkInnerFreeEnergy$ & \parbox{4.75in}{Inner free energy cost for a polydisperse cross-link RVE when perturbing about the known solution for the monodisperse case} \\
    $\widetilde{\generic}$ & Solution for an argument of the infimum of $\clinkInnerFreeEnergy$ \\
    $\sigma$ & Permutation group \\
    $\genGroup$ & Symmetry group of the cross-link RVE structure \\
    $\mathbb{S}_{\numChains}$ & Symmetric group of $\numChains$ elements \\
    $\sigma \cdot \generic$ & Group action of $\sigma$ on $\generic$ \\
    $\ngeomean$ & Monomer number geometric mean for a cross-link with $\ncollection{\nPolydisperseSet}$ monomers in its chains \\
    $\nmean$ & Monomer number arithmetic mean for a cross-link with $\ncollection{\nPolydisperseSet}$ monomers in its chains \\
    $\eff$ & Monomer number efficiency, $\ngeomean / \nmean$. \\
    $\weightFactorSOThreeQuad$ & Weight factor for an $\SOThree$ quadrature point \\
    $\numSOThreeQuad$ & Number of $\SOThree$ quadrature points \\
    $\weightFactorSphQuad$ & Weight factor for a spherical quadrature point \\
    $\numSphQuad$ & Number of spherical quadrature points \\
    $\polarAngle$ & Polar angle (with respect to the positive polar axis in a spherical coordinate system) \\
    $\azimuthalAngle$ & \parbox{4.75in}{Azimuthal angle (the angle of rotation of the radial line about the polar axis in a spherical coordinate system)} \\
    $\left(1, \polarAngle, \azimuthalAngle\right)$ & Spherical quadrature point (in spherical coordinates) \\
    $\spinAngle$ & Spin angle \\
    $\numSpinQuad$ & Number of discretized spin angles \\
    $\delYcQO$ & Perturbation of the cross-link position in the frame averaging limit where $\genRot = \genRotRef$ \\
    $\n_a,\n_b$ & \parbox{4.75in}{Chain monomer numbers in a bimodal polymer network, where the probability of a given chain having $\n_a$ and $\n_b$ monomers is $p$ and $1-p$, respectively} \\
    $\binom{x}{y}$ & Binomial coefficient (``$x$ choose $y$'') \\
    1-3 RVE & Bimodal 4-chain RVE with one chain of one length and three chains of the other length \\
    2-2 RVE & Bimodal 4-chain RVE with two chains of each length \\
    $\nA,\nB$ & Chain monomer numbers in the 1-3 RVE with multiplicity 1 and 3, respectively \\
    $\E$ & Tangent stiffness modulus \\
    $\nSet$ & Domain of $\n$ (discrete) \\
    $\numChainsSet$ & Domain of $\numChains$ (discrete)
\end{longtable}
\end{hlbreakable}

\section{Numerical implementation of the inverse Langevin function} \label{app:numerical-implementation-inv-langevin-func}

In this work, we numerically approximate the inverse Langevin function $\invLang\left(x\right)$ with the Jedynak $R_{9,2}$ approximant~\cite{jedynak2017new},
\begin{equation}
    \invLang\left(x\right) = \frac{x\left(3-1.00651x^2-0.96225x^4+1.47353x^6-0.48953x^8\right)}{(1-x)(1+1.01524x)}.
\end{equation}
Using this inverse Langevin approximant, we obtain the following approximation to the Kuhn and Gr\"{u}n chain free energy in \Eqref{eq:kuhn-grun-free-energy} (as per~\cite{jedynak2017new}),
\begin{align}
    \frac{\KGFreeEnergy\left(\rmag\right)}{\n \kB \T} & = -0.015 + 0.0072\left(\frac{\rmag}{\cLen}\right) + 0.4887\left(\frac{\rmag}{\cLen}\right)^2 - 0.0025\left(\frac{\rmag}{\cLen}\right)^3 - 0.0035\left(\frac{\rmag}{\cLen}\right)^4 - 0.0015\left(\frac{\rmag}{\cLen}\right)^5 \nonumber \\
    & \qquad - 0.1627\left(\frac{\rmag}{\cLen}\right)^6 + 0.001\left(\frac{\rmag}{\cLen}\right)^7 + 0.0603\left(\frac{\rmag}{\cLen}\right)^8 - \ln\left(1-\frac{\rmag}{\cLen}\right) - 0.992\ln\left(0.985+\frac{\rmag}{\cLen}\right).
\end{align}
The above is correspondingly employed within the Kuhn and Gr\"{u}n polymer chain conformation probability density $\KGChainConformationProbDens\left(\rmag\right)$ in \Eqref{eq:FJC}.

The derivative of the inverse Langevin function, $\invLangPrime\left(x\right)$, is also of importance in this work. We approximate $\invLangPrime\left(x\right)$ by taking the derivative of the Jedynak $R_{9,2}$ approximant,
\begin{align}
    \invLangPrime\left(x\right) & = \frac{\partial \invLang\left(x\right)}{\partial x} = \frac{\invLangPrime\left(x\right)_{\text{num}}}{\invLangPrime\left(x\right)_{\text{denom}}}, \\
    \invLangPrime\left(x\right)_{\text{num}} & = \left(3-3.01953x^2-4.81125x^4+10.31471x^6-4.40577x^8\right)(1-x)(1+1.01524x) \nonumber \\
    & \qquad - x(0.01524-2.03048x)\left(3-1.00651x^2-0.96225x^4+1.47353x^6-0.48953x^8\right), \nonumber \\
    \invLangPrime\left(x\right)_{\text{denom}} & = \left((1-x)(1+1.01524x)\right)^2. \nonumber
\end{align}

\section{Recommended chain end positions for polydisperse cross-link RVEs} \label{app:polydisperse-cross-link-structures}

In this Appendix section, we apply the principles introduced in \Fref{sec:cross-link-structures} to define the set of chain end positions, $\collection{\Xvec_i}_{i=1}^\numChains$, and the corresponding set of unit vectors, $\collection{\dir{\Xvec}_i}_{i=1}^\numChains$ ($\dir{\Xvec}_i=\Xvec_i/\left|\Xvec_i\right|$), for polydisperse cross-link RVEs where $\numChains \in [3, 8]$.
Recall that we place these unit vectors according to the Thomson problem: the equilibrium configuration of electrostatically repulsive particles on the unit sphere~\cite{thomson1904xxiv}.
Ideally, we also place these unit vectors to maximize reflectional symmetry about the origin (especially if the Thomson problem is not able to be satisfied).
In what follows, we provide $\collection{\Xvec_i}_{i=1}^\numChains$ that either satisfies or closely satisfies these two considerations for $\numChains \in [3, 8]$.
For each case where $\collection{\Xvec_i}_{i=1}^\numChains$ does not exactly satisfy both of these considerations, it is not known (at least not to the authors) if there even exists some $\collection{\Xvec_i}_{i=1}^\numChains$ that does.
\begin{itemize}
    \item[--] $\numChains=3$: There are many choices for $\collection{\dir{\Xvec}_i}_{i=1}^3$ that solve the Thomson problem and maximize reflectional symmetry, where $\collection{\dir{\Xvec}_i}_{i=1}^3$ define the vertices of an equilateral triangle that rests on some equatorial plane in a unit sphere. One such example, inspired by the $3$-chain models of~\cite{james1943theory},~\cite{elias2006non}, and~\cite{adolf1987computer} is provided below\footnote{This RVE was formulated through the following procedure: \begin{inparaenum}[(i)]
        \item As per the $3$-chain model~\cite{james1943theory}, initially $\collection{\dir{\Xvec}_i = \euclid{i}}_{i=1}^3$ and $\yc = \nullvec$; \item translate $\collection{\dir{\Xvec}_1, \dir{\Xvec}_2, \dir{\Xvec}_3}$ by $\left(-\frac{1}{3}, -\frac{1}{3}, -\frac{1}{3} \right)$ so that $\yc = \nullvec$ is now coincident with the center-of-mass of the RVE; and \item dilate the cube coinciding with the translated $\collection{\dir{\Xvec}_1, \dir{\Xvec}_2, \dir{\Xvec}_3}$ such that $\collection{\left|\dir{\Xvec}_i - \yc\right| = 1}_{i=1}^3$.
    \end{inparaenum}}
    \begin{equation} \label{eq:three-chain-xs}
        \begin{split}
            \Xvec_1 = \left(\rmsChainLength\right)_1 \left(\sqrt{\frac{2}{3}}, -\sqrt{\frac{1}{6}}, -\sqrt{\frac{1}{6}}\right),& \quad \quad
            \Xvec_2 = \left(\rmsChainLength\right)_2 \left(-\sqrt{\frac{1}{6}}, \sqrt{\frac{2}{3}}, -\sqrt{\frac{1}{6}}\right), \\
            \Xvec_3 = \left(\rmsChainLength\right)_3 &\left(-\sqrt{\frac{1}{6}}, -\sqrt{\frac{1}{6}}, \sqrt{\frac{2}{3}}\right).
        \end{split}
    \end{equation} 
    \begin{remark}
        \emph{Special considerations for the $\numChains=3$ RVE domain.}
        Recall that $\RVEdomain$ is defined such that $\RVEdomain \subset \Reals^3$ ($\rvedomain$ is defined similarly). However, for the $\numChains=3$ RVE, the three chains $\collection{\Xvec_i}_{i=1}^3$ always lie on some shared plane, and thus, $\RVEdomain \subset \Reals^2$. We rectify this conflict by noting that \begin{inparaenum}[(1)] \item there exists a rectangular prism (in $\Reals^3$) tightly bounding $\collection{\Xvec_i}_{i=1}^3$, and \item the boundary of said rectangular prism is taken as $\RVEdomain$.\footnote{During deformation, $\rvedomain$ is defined analogously.}
    \end{inparaenum}
    \end{remark}
    \item[--] $\numChains=4$: The Thomson problem is solved and reflectional symmetry is maximized by taking $\collection{\dir{\Xvec}_i}_{i=1}^4$ to be coincident with the vertices of a regular tetrahedron. This aligns with the topology of the $4$-chain model~\cite{flory1943statistical, treloar1943elasticity}:
    \begin{equation}
        \begin{split}
            \Xvec_1 &= \left(\rmsChainLength\right)_1 \left(0, 0, 1\right), \quad \quad
            \Xvec_2 = \left(\rmsChainLength\right)_2 \left(0, \frac{2\sqrt{2}}{3}, -\frac{1}{3}\right), \\
            \Xvec_3 &= \left(\rmsChainLength\right)_3 \left(\sqrt{\frac{2}{3}}, -\frac{\sqrt{2}}{3}, -\frac{1}{3}\right), \quad \quad
            \Xvec_4 = \left(\rmsChainLength\right)_4 \left(-\sqrt{\frac{2}{3}}, -\frac{\sqrt{2}}{3}, -\frac{1}{3}\right).
        \end{split}
    \end{equation}
    \item[--] $\numChains=5$: The Thomson problem is solved by taking $\collection{\dir{\Xvec}_i}_{i=1}^5$ to be coincident with the vertices of an equilateral triangular bipyramid:
    \begin{equation} \label{eq:five-chain-xs}
            \begin{split}
                \Xvec_1 = \left(\rmsChainLength\right)_1 \left(0, 0, 1\right),& \quad \quad
                \Xvec_2 = \left(\rmsChainLength\right)_2 \left(1, 0, 0\right), \\
                \Xvec_3 = \left(\rmsChainLength\right)_3 \left(-\frac{1}{2}, \frac{\sqrt{3}}{2}, 0\right),& \quad \quad
                \Xvec_4 = \left(\rmsChainLength\right)_4 \left(-\frac{1}{2}, -\frac{\sqrt{3}}{2}, 0\right), \\
                \Xvec_5 = &\left(\rmsChainLength\right)_5 \left(0, 0, -1\right).
            \end{split}
        \end{equation}
    \item[--] $\numChains=6$: The Thomson problem is solved and reflectional symmetry is maximized by taking $\collection{\dir{\Xvec}_i}_{i=1}^6$ to be coincident with the vertices of a regular octahedron. This aligns with the topology of the $6$-chain model~\cite{grasinger2023polymer}:
    \begin{equation} \label{eq:six-chain-xs}
        \Xvec_i = \begin{cases}
       \; \; \, \left(\rmsChainLength\right)_i \euclid{i}, &\quad i = 1, 2, 3 \\
               -\left(\rmsChainLength\right)_i \euclid{i-3}, &\quad i = 4, 5, 6
        \end{cases}.
    \end{equation}
    \item[--] $\numChains=7$: The Thomson problem is solved by taking $\collection{\dir{\Xvec}_i}_{i=1}^7$ to be coincide with the vertices of an equilateral pentagonal bipyramid:
    \begin{equation} \label{eq:seven-chain-xs}
            \begin{split}
                \Xvec_1 = \left(\rmsChainLength\right)_1 \left(0, 0, 1\right),& \quad \quad
                \Xvec_2 = \left(\rmsChainLength\right)_2 \left(1, 0, 0\right), \\
                \Xvec_3 = \left(\rmsChainLength\right)_3 \left(\frac{1}{4}\left(\sqrt{5} - 1\right), \sqrt{\frac{5}{8} + \frac{\sqrt{5}}{8}}, 0\right),& \quad \quad
                \Xvec_4 = \left(\rmsChainLength\right)_4 \left(\frac{1}{4}\left(-1 - \sqrt{5}\right), \sqrt{\frac{5}{8} - \frac{\sqrt{5}}{8}}, 0\right), \\
                \Xvec_5 = \left(\rmsChainLength\right)_5 \left(\frac{1}{4}\left(-1 - \sqrt{5}\right), -\sqrt{\frac{5}{8} - \frac{\sqrt{5}}{8}}, 0\right),& \quad \quad
                \Xvec_6 = \left(\rmsChainLength\right)_6 \left(\frac{1}{4}\left(\sqrt{5} - 1\right), -\sqrt{\frac{5}{8} + \frac{\sqrt{5}}{8}}, 0\right), \\
                \Xvec_7 = &\left(\rmsChainLength\right)_7 \left(0, 0, -1\right).
            \end{split}
        \end{equation}
    \item[--] $\numChains=8$: Reflectional symmetry is maximized by taking $\collection{\dir{\Xvec}_i}_{i=1}^8$ to be coincident with the vertices of a regular cube. This aligns with the topology of the $8$-chain model~\cite{arruda1993threee}:
    \begin{equation} \label{eq:eight-chain-xs}
        \Xvec_i = \frac{\left(\rmsChainLength\right)_i}{\sqrt{3}} \left(\pm 1, \pm 1, \pm 1 \right), \quad i = 1, 2, \dots, 8.
    \end{equation}
    The Thomson problem is solved by taking $\collection{\dir{\Xvec}_i}_{i=1}^8$ to be coincident with the vertices of a regular square anti-prism. One can realize such an RVE by twisting any one of the cube faces in the $8$-chain model by $\pi/4$. We provide one such solution below:
    \begin{equation} \label{eq:square-antiprism-eight-chain-xs}
        \Xvec_i = \begin{cases}
       \; \; \, \frac{\left(\rmsChainLength\right)_i}{\sqrt{3}} \left(1, \pm 1, \pm 1 \right), &\quad i = 1, 2, 3, 4 \\
               \frac{\left(\rmsChainLength\right)_i}{\sqrt{2}} \left(-1, \pm 1, 0\right), &\quad i = 5, 6 \\
               \frac{\left(\rmsChainLength\right)_i}{\sqrt{2}} \left(-1, 0, \pm 1\right), &\quad i = 7, 8 \\
        \end{cases}.
    \end{equation}
\end{itemize}

\section{Derivatives of cross-link chain free energy with respect to junction position} \label{app:cross-link-chain-free-energy-derivatives-junction-position}

Recall $\Wchains\left(\F, \genRot, \yc\right) = \sum_{i=1}^\numChains \chainFreeEnergy_i\left(\right|\F \genRot \Xvec_i - \yc\left|\right)$. Also, recall here that $\rvec_i = \F \genRot \Xvec_i - \yc$, and thus, $\rmag_i = \left|\rvec_i\right| = \left|\F \genRot \Xvec_i - \yc\right|$. For Gaussian chains,
\begin{equation}
    \frac{\partial \Wchains}{\partial \yc} = -\frac{3\kB \T}{\mLen}\sum_{i=1}^\numChains \frac{\rvec_i}{\cLen_i},
\end{equation}
\begin{equation} \label{eq:gaussian-cross-link-free-energy-hessian-wrt-junction-position}
    \frac{\partial^2 \Wchains}{\partial \yc \partial \yc} = \frac{3\kB \T}{\mLen^2}\sum_{i=1}^\numChains \frac{1}{\n_i}\identity.
\end{equation}
For Kuhn and Gr\"{u}n chains,
\begin{equation}
    \frac{\partial \Wchains}{\partial \yc} = -\frac{\kB \T}{\mLen}\sum_{i=1}^\numChains \invLang\left(\frac{\rmag_i}{\cLen_i}\right)\frac{\rvec_i}{\rmag_i},
\end{equation}
\begin{equation}
    \frac{\partial^2 \Wchains}{\partial \yc \partial \yc} = \frac{\kB \T}{\mLen^2}\sum_{i=1}^\numChains\frac{1}{\n_i}\left(\invLangPrime\left(\frac{\rmag_i}{\cLen_i}\right)\left(\frac{\rvec_i \otimes \rvec_i}{\rmag_i^2}\right)+\left(\frac{\invLang\left(\frac{\rmag_i}{\cLen_i}\right)}{\frac{\rmag_i}{\cLen_i}}\right)\left(\identity - \left(\frac{\rvec_i \otimes \rvec_i}{\rmag_i^2}\right)\right)\right),
\end{equation}
where $\rmag_i > 0$ must hold.
The numerical implementation of $\invLang$ and $\invLangPrime$ is detailed in \Fref{app:numerical-implementation-inv-langevin-func}.
Note that in this Appendix section we have assumed that all the chains in the cross-link have the same monomer length $\mLen$ (the equations here can be easily reformulated for the case where $\mLen$ varies between chains).

\section{Nondimensional representation of the cross-link partition function and cross-link free energy} \label{app:nondimensional-free-energy-representation}

It is often times convenient (especially for numerical implementation) to represent free energies and their derivatives as nondimensional quantities. With this convention, \Eqref{eq:partition-function-free-rotation-4} and \Eqref{eq:crosslink-free-energy-free-rotation-0} from the free rotation limit each take the following form,
\begin{equation}
   \partitionFunctionFR\left(\F, \genRotRef\right) \approx \left(\frac{\left(2 \pi\right)^{3/2} \mLen^3}{\sqrt{\det \left(\evalAt{\frac{\mLen^2}{\kB \T}\frac{\partial^2 \Wchains}{\partial \yc \partial \yc}}{\genRot = \genRotStar, \yc = \ycStar}\right)}}\right) \exp\left(-\frac{\Wchains\left(\F, \genRotStar, \ycStar\right)}{\kB \T} - \frac{\PEQ\left(\genRotStar\right)}{\kB \T}\right),
\end{equation}
\begin{align}
    \frac{\clinkFreeRotationFreeEnergy\left(\F, \genRotRef\right)}{\kB \T} & = \frac{\Wchains\left(\F, \genRotStar, \ycStar\right)}{\kB \T} + \frac{\PEQ\left(\genRotStar\right)}{\kB \T} + \frac{1}{2}\ln\left(\det \left(\evalAt{\frac{\mLen^2}{\kB \T}\frac{\partial^2 \Wchains}{\partial \yc \partial \yc}}{\genRot = \genRotStar, \yc = \ycStar}\right)\right) \nonumber \\
    & \qquad - \frac{3}{2}\ln\left(2 \pi\right) - 3\ln(\mLen). \label{eq:nondimensional-crosslink-free-energy-free-rotation-0}
\end{align}

In the same vein, \Eqref{eq:partition-function-frame-average-4} and \Eqref{eq:crosslink-free-energy-frame-average-0} from the frame averaging limit each take the following form,
\begin{equation}
   \partitionFunctionFA\left(\F, \genRotRef\right) \approx \left(\frac{\left(2 \pi\right)^{3/2} \mLen^3}{\sqrt{\det \left(\evalAt{\frac{\mLen^2}{\kB \T}\frac{\partial^2 \Wchains}{\partial \yc \partial \yc}}{\genRot = \genRotRef, \yc = \ycStarQO}\right)}}\right) \exp\left(-\frac{\Wchains\left(\F, \genRotRef, \ycStarQO\right)}{\kB \T}\right),
\end{equation}
\begin{align}
    \frac{\clinkFrameAveragingFreeEnergy\left(\F, \genRotRef\right)}{\kB \T} & = \frac{\Wchains\left(\F, \genRotRef, \ycStarQO\right)}{\kB \T} + \frac{1}{2}\ln\left(\det \left(\evalAt{\frac{\mLen^2}{\kB \T}\frac{\partial^2 \Wchains}{\partial \yc \partial \yc}}{\genRot = \genRotRef, \yc = \ycStarQO}\right)\right) \nonumber \\
    & \qquad - \frac{3}{2}\ln\left(2 \pi\right) - 3\ln(\mLen). \label{eq:nondimensional-crosslink-free-energy-frame-average-0}
\end{align}
The last two terms in \Eqref{eq:nondimensional-crosslink-free-energy-free-rotation-0} and \Eqref{eq:nondimensional-crosslink-free-energy-frame-average-0} each have no influence on the mechanics.
Also note that in this Appendix section we have assumed that all the chains in the cross-link have the same monomer length $\mLen$ (the equations here can be reformulated for the case where $\mLen$ varies between chains).

\section{Historical development and related approaches} \label{app:past-work}
\subsection{Historical precedent}
Historically, some of the core ideas of the current approach have been developed elsewhere, albeit separately.
For instance, the influential work of Flory and Rehner~\cite{flory1943statistical}, and later Treloar~\cite{treloar1943elasticity,treloar1943Belasticity,treloar1946elasticity,treloar1975physics,treloar1954photoelastic}, developed a series of $4$-chain polymer network models that included a relaxation of the cross-link junction position. Later on, Kloczkowski, Erman, and Mark~\cite{kloczkowski2002effect} extended this fluctuating-junction $4$-chain model to account for chains with a bimodal distribution of contour lengths.
Adolf and Curro~\cite{adolf1987computer} also incorporated a relaxation of the cross-link junction position in the $3$-chain model setting.
All of these works utilized versions of Assumptions \ref{ass:chain-ends}, \ref{ass:affine}, \ref{ass:thermalize}, and, to some extent, \ref{ass:fluctations}, to establish their models.
However, these RVEs were each fixed relative to a chosen coordinate system and the resulting constitutive models did not satisfy frame indifference.
Considering this context, the work of El\'{i}as-Z\'{u}\~{n}iga~\cite{elias2006non} is noteworthy; here, a fluctuating-junction $3$-chain model is formulated (similar to Adolf and Curro~\cite{adolf1987computer}), where in the undeformed state, the chain ends are located at the vertices of an equilateral tetrahedron cell, and the junction point is located on its centroid. The tetrahedron is further assumed to be contained inside a regular cube. During deformation, the cube subsuming the $3$-chain tetrahedron cell adheres to the principal frame assumption~\cite{grasinger2023polymer} while the junction point relaxes to its equilibrium position.
In the same way, adherence to the principal frame assumption is enforced on the $4$-chain tetrahedron cell in the work of Xing~\cite{xing2025pseudo} (here, the junction position does not fluctuate).
Finally, the seminal work by Arruda and Boyce~\cite{arruda1993threee} -- the ``$8$-chain model'' -- highlighted the importance of considering ``cooperative'' behaviors in polymer networks by allowing the cross-link to rotate to optimally distribute elastic energy across its chains.
This idea was recently extended to cross-links consisting of differing numbers of chains: $3$, $4$, and $6$~\cite{grasinger2023polymer}.
However, neither Arruda and Boyce~\cite{arruda1993threee} nor Grasinger~\cite{grasinger2023polymer} considered polydispersity, or fluctuation and relaxation of the junction position.
This work incorporates each of these ideas, as well as introducing rotational (i.e., twisting) fluctuations of the cross-link.

\subsection{Relationship with prior discrete polymer network models}
Our proposed approach aims to model the mechanics of any arbitrary polydisperse cross-link structure.
Because of this, many past discrete polymer network models are captured, in a generalized sense, by our framework.
Several examples are provided as follows:
\begin{itemize}
    \item \emph{Fluctuating-junction $4$-chain model.}
    The classical $4$-chain model~\cite{flory1943statistical,treloar1943elasticity,treloar1943Belasticity,treloar1946elasticity,treloar1975physics,treloar1954photoelastic} is instantiated via our modeling framework with a tetrafunctional cross-link of uniform Gaussian or Kuhn and Gr\"{u}n polymer chains in the absence of a free rotation assumption (i.e., fluctuations due to cross-link rotations are not accounted for). By amending the aforementioned instantiation such that some chains in the cross-link are long and Gaussian while other chains are short and non-Gaussian, the bimodal $4$-chain model from Kloczkowski, Erman, and Mark~\cite{kloczkowski2002effect} is captured.
    \item \emph{Fluctuating-junction $3$-chain model.} The $3$-chain model from Adolf and Curro~\cite{adolf1987computer} is instantiated with a trifunctional cross-link of uniform Gaussian polymer chains in the absence of a free rotation assumption. The $3$-chain model from El\'{i}as-Z\'{u}\~{n}iga~\cite{elias2006non} is also able to be instantiated by our modeling framework via a trifunctional cross-link of uniform Kuhn and Gr\"{u}n polymer chains. Here, instead of an absence of a free rotation assumption, a restrictive principal frame (rotation) assumption is enforced.
    \item \emph{Arruda-Boyce 8-chain model.} As shown in \Fref{sec:FR-monodispersity}, the Arruda-Boyce 8-chain model is instantiated by our modeling framework via an octafunctional cubic cross-link (see \Eqref{eq:eight-chain-xs}) of uniform Kuhn and Gr\"{u}n polymer chains in the free rotation limit.
    \item \emph{Free rotation assumption from Grasinger~\cite{grasinger2023polymer}.} By assuming uniform polymer chains in the absence of cross-link junction fluctuations, we recover the free rotation assumption from Grasinger~\cite{grasinger2023polymer}.
\end{itemize}

\subsection{Comparing the frame averaging limit with the full network model}
The full network model consists of a continuous uniform distribution of chain end-to-end vectors on the sphere of radius $\rmsChainLength$ (i.e., $\unitSphere(\rmsChainLength) = \collection{\rrmsvec\in\Reals^3~\colon~ \left| \rrmsvec \right|=\rmsChainLength}$).
In the original formulation~\cite{wu1992improved,wu1993improved}, the chain end-to-end vectors are assumed to deform affinely under $\F$ (other deformation assumptions for the full network model have since been formulated, e.g., \cite{miehe2004micro,diani2019fully,mulderrig2021affine,araujo2024micromechanical,araujo2026role,tkachuk2012maximal,rastak2018non,tkachuk2022elastic}).
Considering this, the full network model RVE free energy is
\begin{equation} \label{eq:rve-free-energy-full-network-model}
\rveFullNetworkFreeEnergy\left(\F\right) = \frac{1}{\volUnitSphere} \intOverSphere{\chainFreeEnergy\left(\rmsChainLength\sqrt{\Rdir\cdot\cGreen\Rdir}\right)},
\end{equation}
where $\Rdir\in\unitSphere$ is the chain orientation.

When comparing \Eqref{eq:rve-free-energy-full-network-model} with \Eqref{eq:crosslink-free-energy-frame-average-integral-orientations}, an interesting similarity is observed between the full network model and the frame averaging limit: both frameworks consider a continuous (uniform) probability distribution of polymer network constituents to determine the free energy of the RVE. 
Because of this, similar quadrature techniques can be used to evaluate the integration required for each framework.\footnote{Spherical quadrature is needed to evaluate the integration over $\unitSphere$ for the full network model. For the frame averaging limit, we develop a bespoke quadrature technique for integration over $\SOThree$ (interestingly, this $\SOThree$ quadrature is built directly on top of any arbitrary spherical integration technique). We discuss this $\SOThree$ quadrature in \Fref{sec:FA-optimization} and \Fref{app:SO3-quadrature}.}

The two frameworks differ in several respects. The frame averaging limit RVE permits chains with varying monomer numbers, while the full network model RVE assumes uniform chains. The fluctuating cross-link junction allows non-affine chain deformations in the frame averaging limit, whereas the full network model assumes affine deformation.
One difference merits further discussion. The full network model RVE comprises infinitely many polymer chains emanating spherically from a common junction point (via the continuous uniform distribution of chains on $\unitSphere$)~\cite{grasinger2023polymer}. While this accounts for chains oriented in all directions, the construction is non-physical -- no cross-link connects to infinitely many chains. The frame averaging limit instead uses a continuous distribution of real cross-link structures (each with finitely many chains) over all possible orientations. This also accounts for chains in all directions, but as a consequence of averaging over orientations of physical cross-link geometries rather than by artificially imposing infinite functionality. One could argue that the frame averaging limit is a more physically realistic realization of the full network model.

\subsection{Connections to composite homogenization}
It is instructive to briefly consider the proposed modeling approach through the lens of homogenization theory.
A common starting point for polymer networks is the affine deformation assumption, where end-to-end vectors deform directly with the macroscopic deformation gradient, $\F$.
This is analogous to the Voigt (constant strain) approximation in composites, which is known to provide an upper bound on both the free energy density and the homogenized mechanical response.
While suitable for some monodisperse networks, this approximation can be poor for systems with chains of different lengths or orientations, as it enforces the development of unequal tensions, leading to unbalanced forces and inefficient distribution of elastic energy.
An alternative is provided by equal force theories~\cite{von2002mesoscale,verron2017equal,li2020variational,mulderrig2021affine}, which assume chains aligned in the same direction share the same tension, analogous to the Reuss (constant stress) approximation.
Although this can better match experimental data for polydisperse networks, the Reuss assumption can also be too restrictive.
In contrast, the proposed model -- which involves relaxing the junction position -- is analogous to more advanced nonlinear homogenization techniques where the boundary of an RVE is deformed while its interior relaxes to equilibrium~\cite{milton2022theory,caulfield2024twinning}.
The junction relaxation naturally satisfies the force balance at each cross-link and, we postulate, permits a more realistic stress distribution than either the equal stretch or equal force theories.

\section{Proof of the conservation of sum of stretches squared with respect to rotations, \Fref{prop:stretch-conservation}} \label{app:unification}
This result was first proved in~\cite{grasinger2023polymer}.
It is reproduced here for completeness.
\begin{proposition}[Conservation property] \label{prop:stretch-conservation}
	The sum of squares of the chain stretches (i.e., $\sum_{i=1}^{\numChains} \chainStretch_i^2$) is conserved relative to general rotations \begin{enumerate}[a)] \item for all of the classical $3$, $4$, $6$, and $8$-chain RVEs and \item provided the RVE has reflection symmetries about planes (passing through the origin) normal to three orthogonal directions, $\nuVec_j, \: j = 1, 2, 3$, and, $\sum_{i=1}^{\numChains} \left(\Xvec_i \cdot \nuVec_j\right)^2 = C$ for all $j$.
	\end{enumerate}
\end{proposition}
While there is some overlap between a) and b), they are considered separately because a) establishes the conservation property for well known discrete polymer networks while b) is a sufficient condition for constructing new polymer network models with the same conservation property.
\begin{proof}
	\emph{a)} Let the coordinate systems for the RVEs be chosen such that
	\begin{subequations} \label{eq:Rvecs}
	\begin{equation} \label{eq:three-chain-Rvecs}
		\Xvec_i = \Rmag \euclid{i}, \quad i = 1, 2, 3,
	\end{equation}
	for the $3$-chain RVE;
	\begin{equation} \label{eq:four-chain-Rvecs}
		\begin{split}
		\Xvec_1 &= \Rmag \left(0, 0, 1\right), \quad \quad
		\Xvec_2 = \Rmag \left(0, \frac{2\sqrt{2}}{3}, -\frac{1}{3}\right) \\
		\Xvec_3 &= \Rmag \left(\sqrt{\frac{2}{3}}, -\frac{\sqrt{2}}{3}, -\frac{1}{3}\right), \quad \quad
		\Xvec_4 = \Rmag \left(-\sqrt{\frac{2}{3}}, -\frac{\sqrt{2}}{3}, -\frac{1}{3}\right)
		\end{split}
	\end{equation}
	for the $4$-chain RVE;
	\begin{equation} \label{eq:six-chain-Rvecs}
		\Xvec_i = \begin{cases}
	   \; \; \, \Rmag \euclid{i}, &\quad i = 1, 2, 3 \\
		       -\Rmag \euclid{i-3}, &\quad i = 4, 5, 6
		\end{cases}
	\end{equation}
	for the $6$-chain RVE;
	and
	\begin{equation} \label{eq:eight-chain-Rvecs}
		\Xvec_i = \frac{\Rmag}{\sqrt{3}} \left(\pm 1, \pm 1, \pm 1 \right), \quad i = 1, 2, \dots, 8
	\end{equation}
	\end{subequations}
	for the $8$-chain RVE.
	In each case,
	\begin{equation} \label{eq:sum-of-sq-strs}
		\sum_{i=1}^{\numChains} \chainStretch_i^2 = \frac{1}{\cLen^2} \sum_{i=1}^{\numChains} \left(\F \genRot \Xvec_i\right) \cdot \left(\F \genRot \Xvec_i\right) = \frac{1}{\cLen^2} \sum_{i=1}^{\numChains} \Xvec_i \cdot \changeCoord{\cGreen} \Xvec_i,
	\end{equation}
	where again by $\changeCoord{\cGreen}$ we mean a similarity transformation of $\cGreen$.
	Using \Eqref{eq:Rvecs} and \Eqref{eq:sum-of-sq-strs}, it can be shown directly that
	\begin{equation}
		\sum_{i=1}^{\numChains} \chainStretch_i^2 = \text{const} \times \trace \changeCoord{\cGreen}
	\end{equation}
	for each case, which is invariant with respect to $\genRot$, as desired.\footnote{The algebra was verified in the Mathematica notebook, \texttt{MPS-D-22-01011.nb}, that can be found at \url{https://github.com/grasingerm/MPS-D-22-01011/}}
	
	\emph{b)}
	Expand each $\Xvec$ in the orthonormal basis, $\nuVec_j, \: j = 1, 2, 3$.
	Let $\Xvec_m = \left(\Rmag\right)_{m1} \nuVec_1 + \left(\Rmag\right)_{m2} \nuVec_2 + \left(\Rmag\right)_{m3} \nuVec_3$ where $\left(\Rmag\right)_{mj} = \left(\Xvec_m \cdot \nuVec_j\right)$. Then
	\begin{equation}
		\sum_{m=1}^{\numChains} \chainStretch_m^2 = \frac{1}{\cLen^2} \sum_{m=1}^{\numChains} \Xvec_m \cdot \changeCoord{\cGreen} \Xvec_m = \frac{1}{\cLen^2} \sum_{i=1,j=1}^{3,3} \sum_{m=1}^{\numChains} \left(\Rmag\right)_{mi} \changeCoord{\changeCoord{\cGreenSym}}_{ij} \left(\Rmag\right)_{mj} = \frac{1}{\cLen^2} \sum_{i=1}^{3} \changeCoord{\changeCoord{\cGreenSym}}_{ii} \sum_m^{\numChains} \left(\Rmag\right)_{mi}^2,
	\end{equation}
	where $\changeCoord{\changeCoord{\cGreen}}$ is a similarity transformation of $\changeCoord{\cGreen}$ (and, consequently, is a similarity transformation of $\cGreen$) and the last step is due to the reflection symmetries.
	Clearly $\sum_{i=1}^{\numChains} \chainStretch_i^2 = \left(C / \cLen^2\right) \Tr \changeCoord{\changeCoord{\cGreen}} = \left(C / \cLen^2\right) \Tr \cGreen$, as desired.
\end{proof}

\section{$\SOThree$ quadrature formulation and instantiation}
\label{app:SO3-quadrature}

In this work, integration over $\SOThree$ (the proper orthogonal group of rotations in $3$ dimensions) is evaluated via a numerical quadrature technique we call $\SOThree$ quadrature
\begin{equation}
\frac{1}{\volSOThree} \int_{\genRotRef \in \SOThree} \df{\genRotRef} \genFunc\left(\genRotRef\right) \approxeq \sum_{i=1}^{\numSOThreeQuad} \weightFactorSOThreeQuad_i \genFunc\left(\left(\quadPointSOThreeQuad\right)_i\right),
\end{equation}
where $\genFunc$ is some general (scalar, vectorial, or tensorial) function of orientation $\genRotRef$, $\collection{\left(\quadPointSOThreeQuad\right)_i}_{i=1}^{\numSOThreeQuad}$ is the set of $\SOThree$ quadrature points, $\collection{\weightFactorSOThreeQuad_i}_{i=1}^{\numSOThreeQuad}$ is the set of corresponding weight factors, and $\numSOThreeQuad$ is the number of quadrature points. To instantiate the $\SOThree$ quadrature, we follow the sequential procedure provided below:
\begin{enumerate}[(1)]
    \item Select your spherical quadrature scheme of choice, with $\numSphQuad$ spherical quadrature points and $\collection{\weightFactorSphQuad_j}_{j=1}^{\numSphQuad}$ corresponding weight factors. Common spherical quadrature schemes used in mechanics modeling have included those by Ba\v{z}ant and Oh~\cite{bavzant1986efficient}, Sloan and Womersley~\cite{sloan2004extremal}, and Lebedev~\cite{lebedev1975values,lebedev1976quadratures,lebedev1977spherical,lebedev1992quadrature,lebedev1994quadrature,lebedev1999quadrature} (although many others can be used~\cite{heo2001constructing,fliege1996two,fliege1999distribution,stroud1971approximate,mclaren1963optimal,albrecht1958numerischen,beentjes2015quadrature,badel2004note,neutsch1983optimal,hannay2004fibonacci,sobolev1962number}).
    \item Extract the set of spherical coordinates $\collection{\left(1, \polarAngle_j, \azimuthalAngle_j\right)}_{j=1}^{\numSphQuad}$ from the spherical quadrature scheme. Here, we use the ISO 80000-2:2019~\cite{isoiso} convention (also called the physics convention) for spherical coordinates, where $\polarAngle$ is the polar angle (with respect to the positive polar axis) and $\azimuthalAngle$ is the azimuthal angle (the angle of rotation of the radial line about the polar axis). Depending on the coordinate system used to provide the spherical quadrature points, a coordinate system conversion may need to be performed at this step (e.g., a Cartesian-to-spherical coordinate conversion).
    \item Spherical quadrature is able to be extended to $\SOThree$ quadrature via the following simple step: at each $j$th orientation dictated by the $j$th spherical quadrature point, $\left(1, \polarAngle_j, \azimuthalAngle_j\right)$, we must further integrate over all ``spin orientations''. We capture this concept of ``spin orientation'' via the spin angle $\spinAngle$, where $\spinAngle\in[0, 2\pi)$. In $\SOThree$ quadrature, we must discretize the spin domain of $\spinAngle$. In this work, we choose a simple method where we discretize the spin domain into $\numSpinQuad$ evenly-spaced points, starting at 0 radians. Thus, associated with each $j$th spherical quadrature point are $\numSpinQuad$ discretized spins. The weight factor that corresponds with each of these discretized spin points is equal to $\weightFactorSphQuad_j / \numSpinQuad$.
    \item In light of the details provided in step (3), we begin instantiating $\SOThree$ quadrature by building the following two sets: $\collection{(\spinAngle_i, \polarAngle_i, \azimuthalAngle_i)}_{i=1}^{\numSOThreeQuad}$ and $\collection{\weightFactorSOThreeQuad_i}_{i=1}^{\numSOThreeQuad}$, where $\numSOThreeQuad = \numSphQuad * \numSpinQuad$, and each $\weightFactorSOThreeQuad$ is appropriately calculated with respect to $\weightFactorSphQuad / \numSpinQuad$.
    \item At this point, we reveal that the set $\collection{(\spinAngle_i, \polarAngle_i, \azimuthalAngle_i)}_{i=1}^{\numSOThreeQuad}$ is actually a collection of Euler angles that discretizes rotations in $\SOThree$. We thus convert $\collection{(\spinAngle_i, \polarAngle_i, \azimuthalAngle_i)}_{i=1}^{\numSOThreeQuad}$ to $\collection{\left(\quadPointSOThreeQuad\right)_i}_{i=1}^{\numSOThreeQuad}$ via the ZYZ 3D rotation convention, as follows:
    \begin{equation}
        \left(\quadPointSOThreeQuad\right)_i = \polRot_z(\spinAngle_i)\polRot_y(\polarAngle_i)\polRot_z(\azimuthalAngle_i),
    \end{equation}
    with corresponding rotation matrices
    \begin{equation}
        \polRot_z(\alpha) = \begin{pmatrix}
			\cos{\alpha} & -\sin{\alpha} & 0 \\
			\sin{\alpha} & \cos{\alpha} & 0 \\
			0 & 0 & 1
		\end{pmatrix},~\polRot_y(\beta) = \begin{pmatrix}
			\cos{\beta} & 0 & \sin{\beta} \\
			0 & 1 & 0 \\
			-\sin{\beta} & 0 & \cos{\beta}
		\end{pmatrix}.
    \end{equation}
    \item When appropriate, exploit symmetry in the $\SOThree$ quadrature (stemming from symmetry in the foundational spherical quadrature scheme) for enhanced computational efficiency.
\end{enumerate}

\section{Closed-form approximation for the cross-link position in the frame averaging limit for polydisperse Kuhn and Gr\"{u}n chains} \label{app:KG-approx-yc}

Consider a polydisperse RVE consisting of chains with Kuhn and Gr\"{u}n free energy in the frame averaging limit. We here seek an analytical form for $\delYcQO$ as per \Eqref{eq:FA-analytical-model},
\begin{equation}
    \label{eq:FA-analytical-model-duplicate}
    \delYcQO = -\left(\frac{\partial^2 \clinkInnerFrameAveragingFreeEnergy}{\partial \delYcQO \: \partial \delYcQO}\Bigg|_{\nullvec}\right)^{-1}\partialx{\clinkInnerFrameAveragingFreeEnergy}{\left(\delYcQO\right)}\Bigg|_{\nullvec}.
\end{equation}
Also recall the form for the chain end-to-end vector,
\begin{equation}
    \rvec_i = \F \genRotRef \Xvec_i - \left(\ycMono + \delYcQO\right),
\end{equation}
where $\ycMono = \nullvec$. To denote the fact that $\rvec_i$ is evaluated at $\delYcQO = \nullvec$, as per \Eqref{eq:FA-analytical-model-duplicate}, we define the following shorthand notation:
\begin{equation}
    \rvec_i\big|_{\nullvec} = \F \genRotRef \Xvec_i,\qquad \rmag_i\big|_{\nullvec} = \left|\rvec_i\big|_{\nullvec}\right| = \left|\F \genRotRef \Xvec_i\right|.
\end{equation}
Given all of this, the analytical form for $\delYcQO$ is found to be
\begin{equation} \label{eq:KG-chain-FA-analytical-model}
    \delYcQO = \left(\sum_{i=1}^\numChains\frac{1}{\n_i}\left(\invLangPrime\left(\frac{\rmag_i\big|_{\nullvec}}{\cLen_i}\right)\left(\frac{\rvec_i\big|_{\nullvec} \otimes \rvec_i\big|_{\nullvec}}{\left(\rmag_i\big|_{\nullvec}\right)^2}\right)+\left(\frac{\invLang\left(\frac{\rmag_i\big|_{\nullvec}}{\cLen_i}\right)}{\frac{\rmag_i\big|_{\nullvec}}{\cLen_i}}\right)\left(\identity - \left(\frac{\rvec_i\big|_{\nullvec} \otimes \rvec_i\big|_{\nullvec}}{\left(\rmag_i\big|_{\nullvec}\right)^2}\right)\right)\right)\right)^{-1}\left(\sum_{i=1}^\numChains \invLang\left(\frac{\rmag_i\big|_{\nullvec}}{\cLen_i}\right)\frac{\rvec_i\big|_{\nullvec}}{\rmag_i\big|_{\nullvec}}\right),
\end{equation}
where $\rmag_i\big|_{\nullvec} > 0$ must hold.

\begin{hlbreakable}
\section{Model implementation} \label{app:model-implementation}

An overview of the computational implementation of the model framework is provided as follows:\begin{enumerate}[(I)]
    \item Specify the deformation protocol, i.e., the evolution of $\F$.
    \item Specify the following network-wide material parameters and probability densities: $\T$, $\crossDensity$, $\probDensF{\genRotRef}$, $\probDensChain$, $\probDensClinker$.
    Note that we are interested in isotropic materials in this work, and thus, $\probDensF{\genRotRef} = 1/(\volSOThree)$.
    \item As per $\probDensChain$ and $\probDensClinker$, specify $\nSet$ and $\numChainsSet$, the (discrete) domains of $\n$ and $\numChains$, respectively. 
    \item If employing the frame averaging limit: Instantiate the $\SOThree$ quadrature technique by gathering $\collection{\left(\quadPointSOThreeQuad\right)_i}_{i=1}^{\numSOThreeQuad}$ and $\collection{\weightFactorSOThreeQuad_i}_{i=1}^{\numSOThreeQuad}$ via the procedure in \Fref{app:SO3-quadrature}.
    \item Instantiate $\chainSpace$, the space of polydisperse polymer network cross-links, through the following step-by-step process: \begin{enumerate}[(a)]
        \item To generate every possible $\inst$ with degree $\numChains$, simply sample $\n$ for each of the $\numChains$ chains from $\nSet$ with replacement (this is assuming that the number of monomers in each chain are independent and identically distributed).
        The result of this sequential sampling process for some $\inst$ is captured by $\inst = \ncollection{\nPolydisperseSet}$.
        Repeat for each $\numChains$ in $\numChainsSet$. Gather the resulting cross-links together in a set $\collection{\inst_i}_{i=1}^{\left|\chainSpace\right|}$.
        \item For each $\inst$, calculate its corresponding probability $\probDens(\inst)$ via \Eqref{eq:frame-invariant-chainspace-integration}.
        Gather the resulting probabilities together in a set $\collection{\probDens(\inst_i)}_{i=1}^{\left|\chainSpace\right|}$.
        \item \label{item:spec-polydisperse-cross-link-structures} For each $\inst$, specify the convention for how the initial chain ends $\collection{\Xvec_i}_{i=1}^{\numChains}$ are placed and where the junction point $\yc$ is initially located.
        We provide initial polydisperse cross-link RVE topologies for $\numChains \in [3, 8]$ in \Fref{app:polydisperse-cross-link-structures}. 
        \item For each $\inst$, convert its representation from $\inst = \ncollection{\nPolydisperseSet}$ to $\inst = \collection{\left(\n_i, \cLen_i, \Xvec_i, \chainFreeEnergy_i\right)}_{i=1}^{\numChains}$ by probing each chain in the cross-link.
        For the $i$th chain in $\inst$:\begin{enumerate}[(i)]
            \item Specify $\mLen_i$ and calculate $\cLen_i = \n_i b_i$.
            \item Specify the chain free energy function description $\chainFreeEnergy_i$, e.g., $\GaussFreeEnergy$ for Gaussian FJCs, $\KGFreeEnergy$ for Kuhn and Gr\"{u}n FJCs, etc.
            \item Depending on $\chainFreeEnergy_i$, correspondingly set $\chainConformationProbDens_i\left(\rmag\right)$ and $\left(\critChainLength\right)_i$.
            \item Calculate $\left(\rmsChainLength\right)_i$ via \Eqref{eq:chain-dist-prob-dens} and \Eqref{eq:rms-chain-length}.
            \item Specify $\Xvec_i$ using $\left(\rmsChainLength\right)_i$, as per the convention from step (\ref{item:spec-polydisperse-cross-link-structures}).
        \end{enumerate}
    \item If employing the free rotation limit: the instantiation of $\chainSpace$ is here complete; $\chainSpace = \collection{\inst_i}_{i=1}^{\left|\chainSpace\right|}$ with associated probabilities $\collection{\probDens(\inst_i)}_{i=1}^{\left|\chainSpace\right|}$.
    \item \textit{A remark if employing the frame averaging limit.}
    Technically speaking, $\chainSpace$ corresponds to the set of all frame-invariant cross-link structures (that is, the set of cross-links constructed thus far) oriented at each $\left(\quadPointSOThreeQuad\right)_i$.
    Even though this is the case, we need not explicitly carry out that (expensive) evaluation.
    For the sake of computational implementation, we can sufficiently capture $\chainSpace$ by the sets $\collection{\inst_i}_{i=1}^{\left|\chainSpace\right|}$ and $\collection{\left(\quadPointSOThreeQuad\right)_i}_{i=1}^{\numSOThreeQuad}$ with associated probabilities $\collection{\probDens(\inst_i)}_{i=1}^{\left|\chainSpace\right|}$ and $\collection{\weightFactorSOThreeQuad_i}_{i=1}^{\numSOThreeQuad}$ (respectively).\footnote{\begin{hlbreakable}In the context of the frame averaging limit implementation here, $\left|\chainSpace\right|$ corresponds to the number of frame-invariant cross-link structures in the polymer network.\end{hlbreakable}}
    \end{enumerate}
    Note that the cross-link instantiation procedure presented here could be accomplished through other ways and means, as per the discussion at the end of \Fref{sec:cross-link-statistics}.
    \item Evaluate the polydisperse polymer network mechanical response by calculating $\allClinksFreeEnergyDensity$ through the deformation protocol. At each deformation state $\F$:\begin{enumerate}[(a)]
        \item If employing the free rotation limit: the free energy for a given cross-link $\inst$ is given by \Eqref{eq:crosslink-free-energy-free-rotation},
        \begin{equation*}
            \clinkFreeRotationFreeEnergy\left(\F\right) = \Wchains\left(\F, \genRotStar, \ycStar\right) = \inf_{\genRot \in \SOThree, \yc \in \rvedomain} \sum_{i=1}^{\numChains} \chainFreeEnergy_i\left(\left|\F \genRot \Xvec_i - \yc\right|\right),
        \end{equation*}
        where $\rvedomain = \Conv\left(\collection{\F \Xvec_i}_{i=1}^{\numChains}\right)$ represents the deformed cross-link RVE domain. 
        In certain circumstances, evaluating $\clinkFreeRotationFreeEnergy\left(\F\right)$ can be simplified, as follows: \begin{enumerate}[(i)]
            \item If $\inst$ is monodisperse, $\numChains=4$, and $\collection{\Xvec_i}_{i=1}^{\numChains}$ is given by \Eqref{eq:four-chain-xs}, then $\ycStar=\nullvec$ and $\genRotStar$ is given by \Eqref{eq:$4$-chain-Q}.
            \item If $\inst$ is monodisperse, $\numChains=8$, and $\collection{\Xvec_i}_{i=1}^{\numChains}$ is given by \Eqref{eq:eight-chain-xs}, then $\ycStar=\nullvec$ and $\genRotStar$ is given by \Eqref{eq:8-chain-Q}.
            \item If $\F = \diag\left(\pStretch{1}, \pStretch{2}, 1 / \pStretch{1} \pStretch{2} \right)$ (representing incompressible deformation), $\inst$ is polydisperse, all the chains have the same monomer length $\mLen$ and free energy function description, $\numChains=4$ or $\numChains=8$, and $\collection{\Xvec_i}_{i=1}^{\numChains}$,  is given by \Eqref{eq:four-chain-xs} or \Eqref{eq:eight-chain-xs}, then $\clinkFreeRotationFreeEnergy$ can be analytically approximated by the procedure in \Fref{app:FR-poly-approx-implementation}.
        \end{enumerate}
        When considering all cross-links in the network, $\allClinksFreeEnergyDensity$ is calculated via \Eqref{eq:network-free-energy-density} and \Eqref{eq:chainspace-integration},
        \begin{equation*}
            \allClinksFreeEnergyDensity\left(\F\right) = \frac{\crossDensity}{2}\sum_{i=1}^{\left|\chainSpace\right|}\probDens(\inst_i)\left(\inf_{\genRot \in \SOThree, \yc \in \rvedomain} \sum_{j=1}^{\numChains} \chainFreeEnergy_j\left(\left|\F \genRot \Xvec_j - \yc\right|\right)\right).
        \end{equation*}
        \item If employing the frame averaging limit: the free energy for a given frame-invariant cross-link $\inst$ at orientation $\genRotRef$ is given by \Eqref{eq:crosslink-free-energy-frame-average},
        \begin{equation*}
            \clinkFrameAveragingFreeEnergy\left(\F, \genRotRef\right) = \Wchains\left(\F, \genRotRef, \ycStarQO\right) = \inf_{\yc \in \rvedomain} \sum_{i=1}^{\numChains} \chainFreeEnergy_i\left(\left|\F \genRotRef \Xvec_i - \yc\right|\right),
        \end{equation*}
        where $\rvedomain = \Conv\left(\collection{\F \genRotRef\Xvec_i}_{i=1}^{\numChains}\right)$ represents the deformed and rotated cross-link RVE domain. 
        In certain circumstances, evaluating $\clinkFrameAveragingFreeEnergy\left(\F, \genRotRef\right)$ can be simplified, as follows:
        \begin{enumerate}[(i)]
            \item If $\inst$ is monodisperse with Gaussian chains and $\sum_{i=1}^{\numChains} \Xvec_i = \nullvec$, then $\ycStarQO=\nullvec$.
            This solution can be considered as a leading order approximation for analogous cross-links with Kuhn and Gr\"un chains.
            \item If $\inst$ is a monodisperse $6$-chain and $8$-chain RVE with either Kuhn and Gr\"un and wormlike chains where $\sum_{i=1}^{\numChains} \Xvec_i = \nullvec$, then $\ycStarQO=\nullvec$ exactly.
            \item If $\F = \diag\left(\pStretch{1}, \pStretch{2}, 1 / \pStretch{1} \pStretch{2} \right)$ (representing incompressible deformation), $\inst$ is polydisperse, all the chains have the same monomer length $\mLen$ and free energy function description, and $\sum_{i=1}^{\numChains} \Xvec_i = \nullvec$ in the analogous monodisperse cross-link RVE, then $\clinkFrameAveragingFreeEnergy$ can be analytically approximated by the procedure in \Fref{app:FA-poly-approx-implementation}.
        \end{enumerate}
        When considering all frame-invariant cross-links in all orientations in the network, $\allClinksFreeEnergyDensity$ is calculated via \Eqref{eq:network-free-energy-density} and \Eqref{eq:chainspace-integration},
        \begin{equation*}
            \allClinksFreeEnergyDensity\left(\F\right) = \frac{\crossDensity}{2}\sum_{h=1}^{\left|\chainSpace\right|} \probDens(\inst_h) \left(\sum_{i=1}^{\numSOThreeQuad}  \weightFactorSOThreeQuad_i \left(\inf_{\yc \in \rvedomain} \sum_{j=1}^{\numChains} \chainFreeEnergy_j\left(\left|\F \left(\quadPointSOThreeQuad\right)_i \Xvec_j - \yc\right|\right)\right)\right).
        \end{equation*}
    \end{enumerate}
    \item Calculate $\PK = \partial \allClinksFreeEnergyDensity / \partial \F$ through the deformation protocol via numerical differentiation.
\end{enumerate}

\subsection{Implementation of the free rotation limit closed-form approximation} \label{app:FR-poly-approx-implementation}
If $\F = \diag\left(\pStretch{1}, \pStretch{2}, 1 / \pStretch{1} \pStretch{2} \right)$ (representing incompressible deformation), $\inst$ is polydisperse, all the chains have the same monomer length $\mLen$ and free energy function description, $\numChains=4$ or $\numChains=8$, and $\collection{\Xvec_i}_{i=1}^{\numChains}$, is given by \Eqref{eq:four-chain-xs} or \Eqref{eq:eight-chain-xs}, then $\clinkFreeRotationFreeEnergy$ can be reformulated by considering perturbations about the known solution for the analogous monodisperse $\inst$,
as per \Eqref{eq:FR-FED-perturb-approx}.
We solve for $\clinkFreeRotationFreeEnergy$ by the following procedure: \begin{enumerate}[(I)]
    \item Deduce all of the permutations of the $\ncollection{\nPolydisperseSet}$ chain monomer numbers in $\inst$.
    \item For each cross-link monomer number permutation $\ncollection{\nPolydisperseSet}$: \begin{enumerate}[(a)]
        \item Approximate the derivatives of $\clinkInnerFreeRotationFreeEnergy\left(\delRot, \delYc\right)$ with respect to $\delRot$ and $\delYc$ up to linear order in each, set each of these derivatives equal to zero, and solve for $\delRot$ and $\delYc$,
        as per \Eqref{eq:perturb-deriv-1} and \Eqref{eq:perturb-deriv-2}.
        We denote the emergent solution for $\delRot$ and $\delYc$ as $\widetilde{\delRot}$ and $\widetilde{\delYc}$, respectively.
        \item Calculate the candidate $\clinkFreeRotationFreeEnergy$ value associated with this monomer number permutation as $\clinkInnerFreeRotationFreeEnergy\left(\widetilde{\delRot}, \widetilde{\delYc}\right)$.
    \end{enumerate}
    \item The minimum candidate $\clinkFreeRotationFreeEnergy$ value is the optimal solution for $\clinkFreeRotationFreeEnergy$.
\end{enumerate}
In the case where $\numChains=4$ and the chains are described by the Gaussian chain free energy function, $\widetilde{\delRot}$ and $\widetilde{\delYc}$ are known analytically and are given by \Eqref{eq:$4$-chain-analytical-model}.
This could be considered as a leading order approximation for analogous cross-links with Kuhn and Gr\"{u}n chains.

\subsection{Implementation of the frame averaging limit closed-form approximation} \label{app:FA-poly-approx-implementation}

If $\F = \diag\left(\pStretch{1}, \pStretch{2}, 1 / \pStretch{1} \pStretch{2} \right)$ (representing incompressible deformation), $\inst$ is polydisperse, all the chains have the same monomer length $\mLen$ and free energy function description, and $\sum_{i=1}^{\numChains} \Xvec_i = \nullvec$ in the analogous monodisperse cross-link RVE, then $\clinkFrameAveragingFreeEnergy$ can be reformulated by considering perturbations about the known solution for the analogous monodisperse $\inst$,
as per \Eqref{eq:FA-FED-perturb-approx}.
We solve for $\clinkFrameAveragingFreeEnergy$ by approximating the derivative of $\clinkInnerFrameAveragingFreeEnergy\left(\delYcQO, \genRotRef\right)$ with respect to $\delYcQO$ up to linear order, set the derivative equal to zero, and solve for $\delYcQO$,
as per \Eqref{eq:FA-analytical-model}.
The exact solution for $\delYcQO$ for polydisperse RVEs consisting of Gaussain chains is given by \Eqref{eq:gaussian-chain-unique-junction-position-relaxation} (where $\genRot=\genRotRef$).
This could be considered as a leading order approximation for analogous cross-links with Kuhn and Gr\"{u}n chains.
However, the exact solution for such cross-links is given by \Eqref{eq:KG-chain-FA-analytical-model}.
\end{hlbreakable}

\bibliographystyle{unsrtnat}
\bibliography{master}

\end{document}